\documentclass[11pt,twoside,a4paper,notitlepage]{article}

\def\DONOTINSERTCOMMENTS{}

\usepackage{ifthen}

\newboolean{conferenceversion}
\setboolean{conferenceversion}{false}

\usepackage[a4paper,margin=2.5cm]{geometry}
\usepackage{times}

\usepackage[english]{babel}

\usepackage{fancyhdr}

\usepackage[sf,SF]{subfigure}

\usepackage{bussproofs}

\usepackage{enumitem}

\usepackage{tikz}
\usepackage{tikz-3dplot}
\usetikzlibrary{shapes}

\makeatletter{}%
\usepackage{ifthen}

\makeatletter{}%

\usepackage{ifthen}

\newboolean{detectedSTOC}
\newboolean{detectedFOCS}
\newboolean{detectedElsevier}
\newboolean{detectedNOW}
\newboolean{detectedLMCS}
\newboolean{detectedPoster}
\newboolean{detectedSIAM}
\newboolean{detectedIEEE}
\newboolean{detectedLNCS}
\newboolean{detectedACM}
\newboolean{detectedSigplanconf}
\newboolean{detectedToC}
\newboolean{detectedLIPIcs}
\newboolean{detectedAAAI}
\newboolean{detectedCompCplx}
\newboolean{detectedArticle}
\newboolean{detectedReport}
\newboolean{detectedThesis}

\makeatletter

\@ifclassloaded{sig-alternate}
{\setboolean{detectedSTOC}{true}}
{\setboolean{detectedSTOC}{false}}

\@ifclassloaded{elsarticle}
{\setboolean{detectedElsevier}{true}}
{\setboolean{detectedElsevier}{false}}

\@ifclassloaded{now}
{\setboolean{detectedNOW}{true}}
{\setboolean{detectedNOW}{false}}

\@ifclassloaded{lmcs}
{\setboolean{detectedLMCS}{true}}
{\setboolean{detectedLMCS}{false}}

\@ifclassloaded{IEEEtran} {
  \setboolean{detectedIEEE}{true}
  \ifCLASSOPTIONconference {
    \setboolean{detectedFOCS}{true}
  }
  \else {
    \setboolean{detectedFOCS}{false}
  }
  \fi
}
{
  \setboolean{detectedFOCS}{false}
  \setboolean{detectedIEEE}{false}
}

\@ifclassloaded{siamltex1213} 
{\setboolean{detectedSIAM}{true}}
{\setboolean{detectedSIAM}{false}}

\@ifclassloaded{llncs}
{\setboolean{detectedLNCS}{true}}
{\setboolean{detectedLNCS}{false}}

\@ifclassloaded{acmsmall}
{\setboolean{detectedACM}{true}}
{\setboolean{detectedACM}{false}}

\@ifclassloaded{sigplanconf}
{\setboolean{detectedSigplanconf}{true}}
{\setboolean{detectedSigplanconf}{false}}

\@ifclassloaded{toc}
{\setboolean{detectedToC}{true}}
{\setboolean{detectedToC}{false}}

\@ifclassloaded{lipics}
{\setboolean{detectedLIPIcs}{true}}
{\setboolean{detectedLIPIcs}{false}}

\@ifclassloaded{cc}
{\setboolean{detectedCompCplx}{true}}
{\setboolean{detectedCompCplx}{false}}

\@ifpackageloaded{aaai}
{\setboolean{detectedAAAI}{true}}
{\setboolean{detectedAAAI}{false}}

\@ifclassloaded{sciposter}
{\setboolean{detectedPoster}{true}}
{\setboolean{detectedPoster}{false}}

\@ifclassloaded{article}
{\setboolean{detectedArticle}{true}}
{\setboolean{detectedArticle}{false}}

\makeatother

\ifthenelse{\not \isundefined{\examen} 
  \and \not \isundefined{\disputationsdatum} 
  \and \not \isundefined{\disputationslokal}}   
  {\setboolean{detectedThesis}{true}}
  {\setboolean{detectedThesis}{false}}

\ifthenelse{\boolean{detectedArticle} \or \boolean{detectedThesis}
  \or \boolean{detectedSTOC} \or \boolean{detectedFOCS}
  \or \boolean{detectedSIAM} \or \boolean{detectedIEEE}
  \or \boolean{detectedPoster}}
{\setboolean{detectedReport}{false}}
{\setboolean{detectedReport}{true}}

\ifthenelse{\boolean{detectedLIPIcs}}
{}
{\usepackage[utf8]{inputenc}}

\ifthenelse{\boolean{detectedIEEE}}
{\usepackage[cmex10]{amsmath}
 \interdisplaylinepenalty=2500}
{\usepackage{amsmath}}
\usepackage{amssymb}
\usepackage{amsfonts}

\usepackage{mathtools}

\ifthenelse
{\boolean{detectedElsevier}}
{}
{\usepackage{varioref}}

\usepackage{xspace}

\ifthenelse    
{\boolean{detectedSTOC}      \or \boolean{detectedSIAM}         \or 
  \boolean{detectedLMCS}     \or \boolean{detectedNOW}          \or 
  \boolean{detectedLNCS}     \or \boolean{detectedACM}          \or
  \boolean{detectedToC}      \or \boolean{detectedLIPIcs}       \or
  \boolean{detectedCompCplx} \or \boolean{detectedSigplanconf}}
{}
{
  \usepackage{amsthm}
  \usepackage{amsthm-boldfont}
}

\ifthenelse        
{\boolean{detectedElsevier}  \or \boolean{detectedFOCS}         \or 
  \boolean{detectedSTOC}     \or \boolean{detectedPoster}       \or
  \boolean{detectedSIAM}     \or \boolean{detectedLMCS}         \or
  \boolean{detectedIEEE}     \or \boolean{detectedNOW}          \or
  \boolean{detectedToC}      \or \boolean{detectedThesis}       \or
  \boolean{detectedLNCS}     \or \boolean{detectedACM}          \or 
  \boolean{detectedLIPIcs}   \or \boolean{detectedAAAI}         \or
  \boolean{detectedCompCplx} \or \boolean{detectedSigplanconf}}
{}
{\usepackage{jncaptionsheadings}}

\makeatletter{}%

\usepackage{ifthen}
\usepackage{xspace}
\ifthenelse
{\boolean{detectedElsevier}}
{}
{\usepackage{varioref}}

\DeclareMathAlphabet{\mathsfsl}{OT1}{cmss}{m}{sl}

\newcommand{\formatfunctiontoset}[1]{\mathit{#1}}

\newcommand{\introduceterm}[1]{{\emph{#1}}}

\newcommand{\eqperiod}{\enspace .}
\newcommand{\eqcomma}{\enspace ,}

\newcommand{\wrt}{with respect to\xspace}

\ifthenelse{\boolean{detectedToC}}{}
  {\newcommand{\eg}{for instance\xspace} %
    }

\ifthenelse{\boolean{detectedToC}}{}{\newcommand{\ie}{i.e.,\ }}

\ifthenelse{\boolean{detectedLIPIcs}}
{}
{}

\newcommand{\ifaoif}{if and only if\xspace}

\ifthenelse{\boolean{detectedIEEE}}{}{}

\newcommand{\bigoh}[1]{\mathrm{O} ( #1 )}

\newcommand{\littleoh}[1]{\mathrm{o} ( #1 )}

\newcommand{\bigtheta}[1]{\Theta ( #1 )}

\newcommand{\Bigomega}[1]{\Omega \bigl( #1 \bigr)}
\newcommand{\bigomega}[1]{\Omega ( #1 )}

\newcommand{\complclassformat}[1]%
        {\textrm{\upshape{\textsf{#1}}}\xspace}

\newcommand{\cocomplclass}[1]%
        {\textrm{\upshape{\textsf{co#1}}}\xspace}

\newcommand{\DTIMEadviceclass}[2]%
    {\ensuremath{\complclassformat{DTIME}\bigl(#1\bigr)/{#2}}}

\newcommand{\NP}{\complclassformat{NP}}

\newcommand{\PSPACE}{\complclassformat{PSPACE}}

\newcommand{\EXPTIME}{\complclassformat{EXPTIME}}

\newcommand{\refsec}[1]{Section~\ref{#1}}

\newcommand{\refapp}[1]{Appendix~\ref{#1}}

\newcommand{\reffig}[1]{Figure~\ref{#1}}

\newcommand{\reftwofigs}[2]{Figures~\ref{#1} and~\ref{#2}}

\newcommand{\refth}[1]{Theorem~\ref{#1}}
\newcommand{\reftwoths}[2]{Theorems~\ref{#1} and~\ref{#2}}

\newcommand{\reflem}[1]{Lemma~\ref{#1}}
\newcommand{\reftwolems}[2]{Lemmas~\ref{#1} and~\ref{#2}}

\newcommand{\refpr}[1]{Proposition~\ref{#1}}

\newcommand{\refdef}[1]{Definition~\ref{#1}}
\newcommand{\reftwodefs}[2]{Definitions~\ref{#1} and~\ref{#2}}

\newcommand{\Reflem}[1]{Lemma~\ref{#1}}

\ifthenelse
{\isundefined{\refeq}}
{\newcommand{\refeq}[1]{\eqref{#1}}}
{\renewcommand{\refeq}[1]{\eqref{#1}}}

\newcommand{\ceiling}[1]{\lceil #1 \rceil}
\newcommand{\Ceiling}[1]{\bigl \lceil #1 \bigr \rceil}

\newcommand{\floor}[1]{\lfloor #1 \rfloor}
\newcommand{\Floor}[1]{\bigl \lfloor #1 \bigr \rfloor}
\newcommand{\FLOOR}[1]{\left \lfloor #1 \right \rfloor}

\ifthenelse{\boolean{detectedToC}}{}
{

  \newcommand{\Nplus}     {\mathbb{N}^{+}}
  \newcommand{\Nzero}     {\mathbb{N}_{0}}
  
}

\newcommand{\MAXOFEXPR}[2][]{\max_{#1} \left\{ #2 \right\}}
\newcommand{\MINOFEXPR}[2][]{\min_{#1} \left\{ #2 \right\}}
\newcommand{\Maxofexpr}[2][]{\max_{#1} \bigl\{ #2 \bigr\}}
\newcommand{\Minofexpr}[2][]{\min_{#1} \bigl\{ #2 \bigr\}}

\newcommand{\MAXOFSET}[3][:]{\max \left\{ #2 #1 #3 \right\}}
\newcommand{\MINOFSET}[3][:]{\min \left\{ #2 #1 #3 \right\}}
\newcommand{\Maxofset}[3][:]{\max \bigl\{ #2 #1 #3 \bigr\}}
\newcommand{\Minofset}[3][:]{\min \bigl\{ #2 #1 #3 \bigr\}}

\renewcommand{\MAXOFSET}[3][:]%
     {\ifthenelse{\equal{#1}{;}}%
     {\MAXOFEXPR{ #2 \,;\, #3 }}
     {\ifthenelse{\equal{#1}{:}}%
     {\MAXOFEXPR{ #2 \,:\, #3 }}
     {\max \twincommandJN{\left\{}{#2}{\left#1}{\right}{\,#3}{\right\}}}}}
\renewcommand{\MINOFSET}[3][:]%
     {\ifthenelse{\equal{#1}{;}}%
     {\MINOFEXPR{ #2 \,;\, #3 }}
     {\ifthenelse{\equal{#1}{:}}%
     {\MINOFEXPR{ #2 \,:\, #3 }}
     {\min \twincommandJN{\left\{}{#2}{\left#1}{\right}{\,#3}{\right\}}}}}

\renewcommand{\Maxofset}[3][:]%
     {\ifthenelse{\equal{#1}{;}}%
     {\Maxofexpr{ #2 \,;\, #3 }}
     {\ifthenelse{\equal{#1}{:}}%
     {\Maxofexpr{ #2 \,:\, #3 }}
     {\max \twincommandJN{\bigl\{}{#2}{\bigl#1}{\bigr}{\,#3}{\bigr\}}}}}
\renewcommand{\Minofset}[3][:]%
     {\ifthenelse{\equal{#1}{;}}%
     {\Minofexpr{ #2 \,;\, #3 }}
     {\ifthenelse{\equal{#1}{:}}%
     {\Minofexpr{ #2 \,:\, #3 }}
     {\min \twincommandJN{\bigl\{}{#2}{\bigl#1}{\bigr}{\,#3}{\bigr\}}}}}

\newcommand{\gf}[1]{\mathrm{GF} ( #1 )}
\DeclareMathOperator{\Expop}{E}

\ifthenelse{\boolean{detectedLMCS}}
{}
{}

\newcommand{\twincommandJN}[6]%
    {#1#2#3\vphantom{#2#5}\mspace{-2.05mu}#4.#5#6}

\newcommand{\CondExp}[2]%
    {\Expop\twincommandJN{\bigl[}{#1}{\bigl|}{\bigr}{\,#2}{\bigr]}}
\newcommand{\CONDEXP}[2]%
     {\Expop\twincommandJN{\left[}{#1}{\left|}{\right}{\,#2}{\right]}}

\newcommand{\CondProb}[3][]%
    {\Pr_{#1}\twincommandJN{\bigl[}{#2}{\bigl|}{\bigr}{\,#3}{\bigr]}}
\newcommand{\CONDPROB}[3][]%
    {\Pr_{#1}\twincommandJN{\left[}{#2}{\left|}{\right}{\,#3}{\right]}}

\newcommand{\isdistras}[2]{\ensuremath{#1} \sim \ensuremath{#2}}

\newcommand{\descname}{\formatfunctiontoset{desc}}

\ifthenelse{\isundefined{\descnode}}
{\newcommand{\descnode}[2][]{\descname_{#1}(#2)}}
{\renewcommand{\descnode}[2][]{\descname_{#1}(#2)}}

\newcommand{\setcompact}[1]{{\ensuremath{\bigl\{ #1 \bigr\}}}}

\newcommand{\setdescrcompact}[3][\mid]{{\setcompact{ #2 #1 #3 }}}

\newcommand{\set}[1]{\{ #1 \}}
\newcommand{\Set}[1]{\bigl\{ #1 \bigr\}}

\newcommand{\setdescr}[3][\mid]{\set{ #2 #1 #3 }}
\newcommand{\Setdescr}[3][|]%
     {\ifthenelse{\equal{#1}{;}}%
     {\Set{ #2 \,;\, #3 }}
     {\ifthenelse{\equal{#1}{:}}%
     {\Set{ #2 \,:\, #3 }}
     {\twincommandJN{\bigl\{}{#2\,}{\bigl#1}{\bigr}{\,#3}{\bigr\}}}}}
\newcommand{\SETDESCR}[3][|]%
     {\twincommandJN{\left\{}{#2\,}{\left#1}{\right}{\,#3}{\right\}}}

\newcommand{\Setdescrbrackets}[3][|]%
     {\twincommandJN{\bigl[}{#2}{\bigl#1}{\bigr}{\,#3}{\bigr]}}
\newcommand{\SETDESCRBRACKETS}[3][|]%
     {\twincommandJN{\left[}{#2}{\left#1}{\right}{\,#3}{\right]}}

\newcommand{\Setsize}[1]{\bigl\lvert#1\bigr\rvert}
\newcommand{\setsize}[1]{\lvert#1\rvert}

\newcommand{\union}{\cup}

\newcommand{\unionSP}{\, \union \, }

\newcommand{\disjointunion}{\overset{.}{\cup}}

\newcommand{\DisjointunionInText}%
    {{\smash{\overset{\mbox{\boldmath{.}}}{\bigcup}}}\vphantom{\bigcup}}
\newcommand{\intnfirst}[1]{[{#1}]}

\newcommand{\Lor}{\bigvee}

\newcommand{\olnot}[1]{\overline{#1}}
\newcommand{\stdnot}[1]{\olnot{#1}}
\newcommand{\nvar}{n}

\newcommand{\nclause}{m}

\newcommand{\clwidth}{k}

\newcommand{\randkcnfnclwrepl}[3][\clwidth]%
        {\ensuremath{\mathcal{F}^{#2, #3}_{#1}}}
\newcommand{\randkcnfnclwreplstd}%
        {\randkcnfnclwrepl{\clwidth}{\nvar}{\nclause}}

\newcommand{\israndkcnfnclwrepl}[4]%
  {\isdistras{#1}{\randkcnfnclwrepl[#2]{#3}{#4}}}

\newcommand{\randkcnfprobcl}[3]%
        {\ensuremath{\mathcal{F}^{#2}_{#1} \bigl(#3 \bigr)}}

\newcommand{\pcfor}[4][to]{for #2 := #3 #1 #4 do}
\newcommand{\pcformath}[4][to]%
    {\pcfor[#1]{\ensuremath{#2}}{\ensuremath{#3}}{\ensuremath{#4}}}

\newcommand{\pcassigncompact}[2]{#1 := #2}
\newcommand{\pcassignmathcompact}[2]%
        {\pcassigncompact{\ensuremath{#1}}{\ensuremath{#2}}}

\newcommand{\inductionformat}[1]{\textit{#1}}

\newcommand{\BASE}[1][]
        {\inductionformat
                {%
                        \ifthenelse{\equal{#1}{}}%
                                {Base case: }%
                                {Base case (#1):}%
                }%
        }

\ifthenelse{\isundefined{\DONOTLOADHYPERREFPACKAGE}
  \and \not \boolean{detectedSTOC}     \and \not \boolean{detectedFOCS}
  \and \not \boolean{detectedPoster}   \and \not \boolean{detectedElsevier} 
  \and \not \boolean{detectedSIAM}     \and \not \boolean{detectedACM}
  \and \not \boolean{detectedIEEE}     \and \not \boolean{detectedNOW}
  \and \not \boolean{detectedToC}      \and \not \boolean{detectedThesis}
  \and \not \boolean{detectedLNCS}     \and \not \boolean{detectedLIPIcs}
  \and \not \boolean{detectedAAAI}     \and \not \boolean{detectedSigplanconf}
  \and \not \boolean{detectedCompCplx}}
{\usepackage[colorlinks=true,citecolor=blue]{hyperref}}
{}

\ifthenelse{\boolean{detectedSTOC}}                  
  {%
\makeatletter{}%

\newtheorem{theorem}{Theorem}
\newtheorem{lemma}[theorem]{Lemma}
\newtheorem{proposition}[theorem]{Proposition}
\newtheorem{corollary}[theorem]{Corollary}
\newtheorem{observation}[theorem]{Observation}

\newtheorem{definition}[theorem]{Definition}
\newtheorem{claim}[theorem]{Claim}

\newtheorem{conjecture}{Conjecture}
\newtheorem{openproblem}[conjecture]{Open Problem}

\newdef{remark}{Remark}
\newdef{example}{Example}

\newcounter{unnumber}

}
  {\ifthenelse{\boolean{detectedSIAM}}
    {%
\makeatletter{}%

\newtheorem{observation}[theorem]{Observation}
\newtheorem{claim}[theorem]{Claim}

\newtheorem{conjecture}{Conjecture}
\newtheorem{openquestion}{Open Question}

\newtheorem{remarkinner}[theorem]{Remark}
\newtheorem{exampleinner}[theorem]{Example}

\newcommand{\exampleendmarker}{\qquad$\Diamond$}
\newcommand{\remarkendmarker}{\qquad$\Diamond$}

\newenvironment{example}                        
    {\begin{exampleinner} \rm}
    {\exampleendmarker\end{exampleinner}}

\newenvironment{remark}                        
    {\begin{remarkinner} \rm}
    {\remarkendmarker\end{remarkinner}}

\newcounter{unnumber}

}
    {\ifthenelse{\boolean{detectedNOW}}
      {%
\makeatletter{}%
\newtheorem{standardlocalcounter}{Dummy}[chapter]
\newtheorem{standardglobalcounter}{Dummy}

\newtheorem{theorem}[standardlocalcounter]{Theorem}
\newtheorem{lemma}[standardlocalcounter]{Lemma}
\newtheorem{proposition}[standardlocalcounter]{Proposition}
\newtheorem{corollary}[standardlocalcounter]{Corollary}
\newtheorem{observation}[standardlocalcounter]{Observation}
\newtheorem{fact}[standardlocalcounter]{Fact}

\newtheorem{conjecturelocalcounter}[standardlocalcounter]{Conjecture}
\newtheorem{conjectureglobalcounter}[standardglobalcounter]{Conjecture}
\newtheorem{conjecture}[standardglobalcounter]{Conjecture}
\newtheorem{openquestion}[standardglobalcounter]{Open Question}
\newtheorem{openproblem}[standardglobalcounter]{Open Problem}
\newtheorem{problem}{Problem}

\newtheorem{property}[standardlocalcounter]{Property}
\newtheorem{definition}[standardlocalcounter]{Definition}
\newtheorem{claim}[standardlocalcounter]{Claim}

\newtheorem{algorithm}[standardlocalcounter]{Algorithm}

\newtheorem{remark}[standardlocalcounter]{Remark}
\newtheorem{example}[standardlocalcounter]{Example}

\renewenvironment{proof}[1][Proof]{\par\trivlist
   \item[\hskip \labelsep{\itshape {#1}.}]\prooffont}
   {\hspace*{0pt plus1fill}\fboxsep2.5pt\fboxrule.5pt\raise3pt\hbox{\fbox{}}\endtrivlist}
}
      {\ifthenelse{\boolean{detectedLMCS}}
        {%
\makeatletter{}%

\theoremstyle{plain}    

\newtheorem{theorem}[thm]{Theorem}
\newtheorem{lemma}[thm]{Lemma}
\newtheorem{proposition}[thm]{Proposition}
\newtheorem{corollary}[thm]{Corollary}
\newtheorem{observation}[thm]{Observation}

\newtheorem{conjecture}[thm]{Conjecture}
\newtheorem{problem}[thm]{Problem}

\newtheorem{openquestion}{Open Question}
\newtheorem{openproblem}{Open Problem}

\theoremstyle{definition}

\newtheorem{property}[thm]{Property}
\newtheorem{definition}[thm]{Definition}
\newtheorem{claim}[thm]{Claim}
\newtheorem{remark}[thm]{Remark}
\newtheorem{example}[thm]{Example}

}
        {\ifthenelse{\boolean{detectedIEEE}}
          {%
\makeatletter{}%
\newtheorem{standardlocalcounter}{Dummy}[section]
\newtheorem{standardglobalcounter}{Dummy}

\theoremstyle{plain}    

\newtheorem{theorem}[standardglobalcounter]{Theorem}
\newtheorem{lemma}[standardglobalcounter]{Lemma}
\newtheorem{proposition}[standardglobalcounter]{Proposition}
\newtheorem{corollary}[standardglobalcounter]{Corollary}
\newtheorem{observation}[standardglobalcounter]{Observation}
\newtheorem{fact}[standardglobalcounter]{Fact}

\newtheorem{conjecture}[standardglobalcounter]{Conjecture}
\newtheorem{openquestion}{Open Question}
\newtheorem{openproblem}{Open Problem}
\newtheorem{problem}{Problem}

\theoremstyle{definition}

\newtheorem{property}[standardglobalcounter]{Property}
\newtheorem{definition}[standardglobalcounter]{Definition}
\newtheorem{claim}[standardglobalcounter]{Claim}

\theoremstyle{remark}
\newtheorem{remark}[standardglobalcounter]{Remark}
\newtheorem{example}[standardglobalcounter]{Example}

\newtheoremstyle{meta}%
  {3pt}%
  {3pt}%
  {\scshape \small }%
  {}%
  {\scshape \small }%
  {:}%
  { }%
  {}%

\theoremstyle{meta}
\newtheorem{meta}{Meta comment}

\newtheoremstyle{questions}%
  {3pt}%
  {3pt}%
  {\sffamily \slshape}%
  {}%
  {\bfseries \sffamily \slshape}%
  {:}%
  { }%
  {}%

\theoremstyle{questions}
\newtheorem{questions}{Open questions}

}
          {\ifthenelse{\boolean{detectedLNCS}}
            {%
\makeatletter{}%

\spnewtheorem*{proofsketch}{Proof sketch}{\itshape}{\rmfamily}

\spnewtheorem{observation}{Observation}{\bfseries}{\itshape}
\spnewtheorem{fact}{Fact}{\bfseries}{\itshape}

}
            {\ifthenelse{\boolean{detectedACM}}
              {%
\makeatletter{}%

\newtheorem{observation}[theorem]{Observation}
\newtheorem{fact}[theorem]{Fact}
\newtheorem{claim}[theorem]{Claim}
\newtheorem{openquestion}{Open Question}
\newtheorem{openproblem}{Open Problem}

\newcounter{unnumber}

}
              {\ifthenelse{\boolean{detectedToC}}
                {%
\makeatletter{}%

\theoremstyle{plain}
\newtheorem{observation}[theorem]{Observation}
\newtheorem{openproblem}[theorem]{Open Problem}

\theoremstyle{definition}
\newtheorem{property}[theorem]{Property}

\renewcommand{\refth}[1]{\expref{Theorem}{#1}}
\renewcommand{\reflem}[1]{\expref{Lemma}{#1}}
\renewcommand{\refpr}[1]{\expref{Proposition}{#1}}

\renewcommand{\refdef}[1]{\expref{Definition}{#1}}

\renewcommand{\Reflem}[1]{\expref{Lemma}{#1}}

\renewcommand{\refsec}[1]{\expref{Section}{#1}}

\renewcommand{\refapp}[1]{\expref{Appendix}{#1}}

\renewcommand{\reffig}[1]{\expref{Figure}{#1}}

}
                {\ifthenelse{\boolean{detectedLIPIcs}}
                  {%
\makeatletter{}%

\theoremstyle{plain}    

\newtheorem{fact}[theorem]{Fact}
\newtheorem{proposition}[theorem]{Proposition}
\newtheorem{observation}[theorem]{Observation}
\newtheorem{claim}[theorem]{Claim}
}
                  {\ifthenelse{\boolean{detectedCompCplx}}
                    {%
\makeatletter{}%
}
                    {\ifthenelse{\boolean{detectedSigplanconf}}
                      {%
\makeatletter{}%

\usepackage{amsthm}          
\usepackage{amsthm-boldfont}
\usepackage{url}

\newtheorem{standardlocalcounter}{Dummy}[section]
\newtheorem{standardglobalcounter}{Dummy}

\theoremstyle{plain}    

\newtheorem{theorem}[standardlocalcounter]{Theorem}
\newtheorem{lemma}[standardlocalcounter]{Lemma}
\newtheorem{proposition}[standardlocalcounter]{Proposition}
\newtheorem{corollary}[standardlocalcounter]{Corollary}
\newtheorem{observation}[standardlocalcounter]{Observation}
\newtheorem{fact}[standardlocalcounter]{Fact}

\newtheorem{conjecturelocalcounter}[standardlocalcounter]{Conjecture}
\newtheorem{conjectureglobalcounter}[standardglobalcounter]{Conjecture}
\newtheorem{conjecture}[standardglobalcounter]{Conjecture}
\newtheorem{openquestion}[standardglobalcounter]{Open Question}
\newtheorem{openproblem}[standardglobalcounter]{Open Problem}
\newtheorem{problem}[standardglobalcounter]{Problem}
\newtheorem{question}[standardglobalcounter]{Question}

\theoremstyle{definition}

\newtheorem{property}[standardlocalcounter]{Property}
\newtheorem{definition}[standardlocalcounter]{Definition}
\newtheorem{claim}[standardlocalcounter]{Claim}
\newtheorem{subclaim}[standardlocalcounter]{Subclaim}

\newtheorem{algorithm}[standardlocalcounter]{Algorithm}

\theoremstyle{remark}
\newtheorem{remark}[standardlocalcounter]{Remark}
\newtheorem{example}[standardlocalcounter]{Example}

}
                      {\ifthenelse{\boolean{detectedArticle} 
                          \or \boolean{detectedElsevier}}
                        {%
\makeatletter{}%
\newtheorem{standardlocalcounter}{Dummy}[section]
\newtheorem{standardglobalcounter}{Dummy}

\theoremstyle{plain}    

\newtheorem{theorem}[standardlocalcounter]{Theorem}
\newtheorem{lemma}[standardlocalcounter]{Lemma}
\newtheorem{proposition}[standardlocalcounter]{Proposition}
\newtheorem{corollary}[standardlocalcounter]{Corollary}
\newtheorem{observation}[standardlocalcounter]{Observation}

\newtheorem{conjecturelocalcounter}[standardlocalcounter]{Conjecture}
\newtheorem{conjectureglobalcounter}[standardglobalcounter]{Conjecture}
\newtheorem{conjecture}[standardglobalcounter]{Conjecture}
\newtheorem{openquestion}[standardglobalcounter]{Open Question}
\newtheorem{openproblem}[standardglobalcounter]{Open Problem}
\newtheorem{problem}[standardglobalcounter]{Problem}

\theoremstyle{definition}

\newtheorem{property}[standardlocalcounter]{Property}
\newtheorem{definition}[standardlocalcounter]{Definition}
\newtheorem{claim}[standardlocalcounter]{Claim}
\newtheorem{subclaim}[standardlocalcounter]{Subclaim}

\theoremstyle{remark}
\newtheorem{remark}[standardlocalcounter]{Remark}
\newtheorem{example}[standardlocalcounter]{Example}

}
                        {%
\makeatletter{}%
\newtheorem{standardlocalcounter}{Dummy}[chapter]
\newtheorem{standardglobalcounter}{Dummy}

\theoremstyle{plain}    

\newtheorem{theorem}[standardlocalcounter]{Theorem}
\newtheorem{lemma}[standardlocalcounter]{Lemma}
\newtheorem{proposition}[standardlocalcounter]{Proposition}

\theoremstyle{definition}

\newtheorem{definition}[standardlocalcounter]{Definition}
\newtheorem{claim}[standardlocalcounter]{Claim}

\theoremstyle{remark}

\newtheoremstyle{meta}%
  {3pt}%
  {3pt}%
  {\scshape \small }%
  {}%
  {\scshape \small }%
  {:}%
  { }%
  {}%

\theoremstyle{meta}

\newtheoremstyle{questions}%
  {3pt}%
  {3pt}%
  {\sffamily \slshape}%
  {}%
  {\bfseries \sffamily \slshape}%
  {:}%
  { }%
  {}%

\theoremstyle{questions}

}}}}}}}}}}}}

\ifthenelse{\boolean{detectedThesis}}
{}
{\usepackage{ifpdf}}

\ifthenelse{\boolean{detectedToC}}
{}
{\usepackage{graphicx}}

\ifpdf         
\fi

\ifthenelse
{\boolean{detectedSTOC}  \or \boolean{detectedThesis} \or 
 \boolean{detectedLNCS}  \or \boolean{detectedToC}    \or 
 \boolean{detectedLIPIcs} \or \boolean{detectedAAAI}}
{}
{
  \setcounter{secnumdepth}{3}
  \setcounter{tocdepth}{3}
}

\newcount\hours
\newcount\minutes
\def\SetTime{\hours=\time
\global\divide\hours by 60
\minutes=\hours
\multiply\minutes by 60
\advance\minutes by-\time
\global\multiply\minutes by-1 }
\SetTime
\def\now{\number\hours:\ifnum\minutes<10 0\fi\number\minutes}

\makeatletter{}%
\newcommand{\boundary}[1]{\ensuremath{\partial #1}}

\DeclareFontFamily{OT1}{pzc}{}
\DeclareFontShape{OT1}{pzc}{m}{it}{<-> s * [1.200] pzcmi7t}{}
\DeclareMathAlphabet{\mathpzc}{OT1}{pzc}{m}{it}

\newcommand{\deriveswithall}%
        {\vdash_{\!\!\!{\scriptscriptstyle \forall}}} 
\newcommand{\notderiveswithall}%
        {\nvdash_{\!\!\!{\scriptscriptstyle \forall}}} 

\newcommand{\clcfgtransitioncrammed}[2]%
        {\ensuremath{#1 \!\rightsquigarrow\! #2}}

\newcommand{\varx}{\ensuremath{x}}

\newcommand{\setsofvarsorlitlarge}[2]%
        {\mathit{#1}\left({#2}\right)}
\newcommand{\setsofvarsorlit}[2]%
        {\mathit{#1}({#2})}
\newcommand{\setsofvarsorlitcompact}[2]%
        {\mathit{#1}\bigl({#2}\bigr)}

\newcommand{\setsofvarsorlitsup}[3]%
        {\mathit{#1}^{#2}({#3})}
\newcommand{\setsofvarsorlitsuplarge}[3]%
        {\mathit{#1}^{#2}\left({#3}\right)}
\newcommand{\setsofvarsorlitsupcompact}[3]%
        {\mathit{#1}^{#2}\bigl({#3}\bigr)}

\newcommand{\Vars}[1]{\setsofvarsorlitcompact{Vars}{#1}}

\newcommand{\derivabbrev}[2]{\bigl( #1 \vdash #2 \bigr)}
\newcommand{\derivabbrevsmall}[2]{( #1 \vdash #2 )}
\newcommand{\derivabbrevcompact}[2]{\bigl( #1 \vdash #2 \bigr)}

\newcommand{\refutabbrevsmall}[1]{\derivabbrevsmall{#1}{\falsenum}}
\newcommand{\refutabbrevcompact}[1]{\derivabbrevcompact{#1}{\falsenum}}

\renewcommand{\refutabbrevsmall}[1]{\derivabbrevsmall{#1}{\!\emptycl}}
\renewcommand{\refutabbrevcompact}[1]{\derivabbrevcompact{#1}{\!\emptycl}}

\renewcommand{\refutabbrevsmall}[1]{\derivabbrevsmall{#1}{\!\bot}}
\renewcommand{\refutabbrevcompact}[1]{\derivabbrevcompact{#1}{\!\bot}}

\newcommand{\genericrefsmall}[3]%
    {{\mathit{#1}}_{#2}\refutabbrevsmall{#3}}
\newcommand{\genericrefcompact}[3]%
    {{\mathit{#1}}_{#2}\refutabbrevcompact{#3}}
\newcommand{\genericderiv}[4]%
    {{\mathit{#1}}_{#2}\derivabbrev{#3}{#4}}
\newcommand{\genericderivsmall}[4]%
    {{\mathit{#1}}_{#2}\derivabbrevsmall{#3}{#4}}
\newcommand{\genericderivcompact}[4]%
    {{\mathit{#1}}_{#2}\derivabbrevcompact{#3}{#4}}
\newcommand{\generictaut}[3]%
    {{\mathit{#1}}_{#2}\derivabbrev{}{#3}}
\newcommand{\generictautcompact}[3]%
    {{\mathit{#1}}_{#2}\derivabbrevcompact{}{#3}}
\newcommand{\generictautsmall}[3]%
    {{\mathit{#1}}_{#2}\derivabbrevsmall{}{#3}}

\newcommand{\formulaformat}[1]{\ensuremath{\mathit{#1}}}
\renewcommand{\formulaformat}[1]{\mathit{#1}}

\newcommand{\transitionarrow}{\rightsquigarrow}
\newcommand{\pebcfgtransition}[2]%
    {\ensuremath{#1 \transitionarrow #2}}
\newcommand{\pebcfgtransitionsqueeze}[2]%
    {#1 \! \transitionarrow \! #2}

\newcommand{\formatpebblingprice}[1]{\textsl{\textsf{#1}}}

\newcommand{\Pebblingprice}[1]%
    {\formatpebblingprice{Peb}\bigl(#1\bigr)}
\newcommand{\pebblingpricecompact}[1]%
    {\formatpebblingprice{Peb}\bigl(#1\bigr)}

\newcommand{\Bwpebblingprice}[1]%
    {\formatpebblingprice{BW-Peb}\bigl(#1\bigr)}
\newcommand{\bwpebblingpricecompact}[1]%
    {\formatpebblingprice{BW-Peb}\bigl(#1\bigr)}

\newcommand{\pebpersistentsymbol}{\bullet}
\newcommand{\pebvisitingsymbol}{\emptyset}

\newcommand{\bwpebpricepersistent}[1]%
    {\formatpebblingprice{BW-Peb}^{\pebpersistentsymbol}(#1)}
\newcommand{\Bwpebpricepersistent}[1]%
    {\formatpebblingprice{BW-Peb}^{\pebpersistentsymbol}\bigl(#1\bigr)}
\newcommand{\bwpebpricevisiting}[1]%
    {\formatpebblingprice{BW-Peb}^{\pebvisitingsymbol}(#1)}
\newcommand{\Bwpebpricevisiting}[1]%
    {\formatpebblingprice{BW-Peb}^{\pebvisitingsymbol}\bigl(#1\bigr)}

\newcommand{\pebpricepersistent}[1]%
    {\formatpebblingprice{Peb}^{\pebpersistentsymbol}(#1)}
\newcommand{\Pebpricepersistent}[1]%
    {\formatpebblingprice{Peb}^{\pebpersistentsymbol}\bigl(#1\bigr)}
\newcommand{\pebpricevisiting}[1]%
    {\formatpebblingprice{Peb}^{\pebvisitingsymbol}(#1)}
\newcommand{\Pebpricevisiting}[1]%
    {\formatpebblingprice{Peb}^{\pebvisitingsymbol}\bigl(#1\bigr)}

\newcommand{\bwpebblingpriceempty}[1]%
    {\formatpebblingprice{BW-Peb}^{\pebvisitingsymbol}(#1)}
\newcommand{\bwpebblingpriceemptycompact}[1]%
    {\formatpebblingprice{BW-Peb}^{\pebvisitingsymbol}\bigl(#1\bigr)}

\newcommand{\pebdeg}{\ensuremath{d}}

\newcommand{\pebaxcompact}[2]%
        [\pebdeg]{\ensuremath{\formulaformat{Ax}^{#1} \bigl(#2 \bigr)}}

\newcommand{\pqrxvar}[6]%
    {\ensuremath{\stdnot{\varx({#1})}_{#2} \lor \stdnot{\varx({#3})}_{#4} \lor %
    \sourceclausexvar[#6]{#5}}}

\newcommand{\pqr}[6]%
    {\ensuremath{\stdnot{#1}_{#2} \lor \stdnot{#3}_{#4} \lor %
    \sourceclausenodisplay[#6]{#5}}}
\newcommand{\pqrstd}{\pqr{p}{i}{q}{j}{r}{l}}
\newcommand{\pqrall}[6]%
        {\setdescrcompact
        {\pqr{#1}{#2}{#3}{#4}{#5}{#6}}{#2,#4 \in \intnfirst{\pebdeg}}}
\newcommand{\pqrallstd}%
        {\setdescrcompact{\pqrstd}{i,j \in \intnfirst{\pebdeg}}}

\newcommand{\sourceclausexvar}[2][n]%
        {\Lor_{#1 = 1}^{\pebdeg} \varx({#2})_{#1}}
\newcommand{\subsourceclausexvar}[3][n]%
        {\Lor_{#1 = {#2}}^{\pebdeg} \varx({#3})_{#1}}

\newcommand{\sourceclausexvarnodisplay}[2][n]%
        {\textstyle \Lor_{#1 = 1}^{\pebdeg} \varx({#2})_{#1}}

\newcommand{\sourceclausenodisplay}[2][n]%
        {\textstyle \Lor_{#1 = 1}^{\pebdeg} #2_{#1}}

\newcommand{\relativisation}[1]%
    {\ensuremath{\formulaformat{Rel}\bigl(#1 \bigr)}}

\newcommand{\extPHPnot}[2]
    {\ensuremath{\extendedversion{\formulaformat{PHP}}^{#1}_{#2}}}

\newcommand{\GraphOntoPHPnot}[1][G]%
    {\text{$\formulaformat{Onto}$-$\formulaformat{PHP}$}_{#1}}

\renewcommand{\extPHPnot}[2]%
    {\ephpnot{#1}{#2}}

\newcommand{\ephpnot}[2]%
    {\vphantom{\extendedversion{\formulaformat{PHP}}}
      {\smash{\extendedversion{\formulaformat{PHP}}}
        \vphantom{\formulaformat{PHP}}}^{#1}_{#2}}

\newcommand{\efphpnot}[2]%
    {\vphantom{\extendedversion{\formulaformat{FPHP}}}
      {\smash{\extendedversion{\formulaformat{FPHP}}}
        \vphantom{\formulaformat{FPHP}}}^{#1}_{#2}}

\newcommand{\ontophpnot}[2]%
    {\formulaformat{Onto}\text{-}\formulaformat{PHP}^{#1}_{#2}}

\newcommand{\ontofphpnot}[2]%
    {\formulaformat{Onto}\text{-}\formulaformat{FPHP}^{#1}_{#2}}

\newcommand{\extendedversion}[1]{\widetilde{#1}}

\makeatletter{}%
\usepackage{stmaryrd} 

\newcommand{\formatfunctiontosubconfiguration}[1]{\mathsf{#1}}
\newcommand{\formatfunctiontomulti}[1]{\mathcal{#1}}

\DeclareMathOperator{\dummystar}{*}
\newcommand{\pebblingcontrNT}[2][G]%
 {\ensuremath{\dummystar\!\!\formulaformat{Peb}^{#2}_{#1}}}

\newcommand{\somenodetrueclausedeg}[2]{\formulaformat{All}_{#1}^{+}({#2})}

\newcommand{\slashedstrickenletter}[1]{{\backslash\mkern-9mu #1}}
            
\newcommand{\strikethroughcommand}[1]{\slashedstrickenletter{#1}}

\newcommand{\abovevertices}[2][G]%
    {{#1}_{#2}^{\hspace{-0.2 pt}\triangledown}}
\newcommand{\aboveverticesNR}[2][G]%
    {{#1}_{\strikethroughcommand{#2}}^{\hspace{-0.3 pt}\triangledown}}
\newcommand{\belowvertices}[2][G]%
    {{#1}^{#2}_{\hspace{-0.6 pt}\vartriangle}}
\newcommand{\belowverticesNR}[2][G]%
    {{#1}^{\strikethroughcommand{#2}}_{\hspace{-0.6 pt}\vartriangle}}

\newcommand{\lpebblingpricecompact}[1]%
    {\formatpebblingprice{L-Peb}\bigl(#1\bigr)}

\newcommand{\scnot}[2]{#1 \langle #2 \rangle}
\newcommand{\scnotcompact}[2]{#1 \bigl\langle #2 \bigr\rangle}

\newcommand{\spcanonconfcompact}[1]%
        {\formatfunctiontosubconfiguration{canon}\bigl({#1}\bigr)}

\newcommand{\spprojsubsub}[4]%
    {\formatfunctiontosubconfiguration{proj}_{\scnot{#1}{#2}}(\scnot{#3}{#4})}
\newcommand{\spprojsubsubcompact}[4]%
    {\formatfunctiontosubconfiguration{proj}_{\scnot{#1}{#2}}%
    \bigl(\scnot{#3}{#4}\bigr)}
\newcommand{\spprojsubconf}[3]%
    {\formatfunctiontosubconfiguration{proj}_{\scnot{#1}{#2}}({#3})}
\newcommand{\spprojsubconfcompact}[3]%
    {\formatfunctiontosubconfiguration{proj}_{\scnot{#1}{#2}}\bigl({#3}\bigr)}
\newcommand{\spprojconfsub}[3]%
    {\formatfunctiontosubconfiguration{proj}_{#1}(\scnot{#2}{#3})}
\newcommand{\spprojconfsubcompact}[3]%
    {\formatfunctiontosubconfiguration{proj}_{#1}\bigl(\scnot{#2}{#3}\bigr)}
\newcommand{\spprojconfconf}[2]%
    {\formatfunctiontosubconfiguration{proj}_{#1}({#2})}
\newcommand{\spprojconfconfcompact}[2]%
    {\formatfunctiontosubconfiguration{proj}_{#1}\bigl({#2}\bigr)}

\newcommand{\spclossubcompact}[2]%
        {\formatfunctiontoset{cl}\bigl(\scnotcompact{#1}{#2}\bigr)}

\newcommand{\spintersubcompact}[2]%
        {\formatfunctiontoset{int}\bigl(\scnotcompact{#1}{#2}\bigr)}

\newcommand{\spcoversubcompact}[2]%
        {\formatfunctiontoset{cover}\bigl(\scnotcompact{#1}{#2}\bigr)}

\newcommand{\spcoverconfcompact}[1]%
        {\formatfunctiontoset{cover}\bigl({#1}\bigr)}

\newcommand{\spinducedblack}[1]%
    {\formatfunctiontoset{Bl} (#1)}
\newcommand{\spinducedwhite}[1]%
    {\formatfunctiontoset{Wh} (#1)}
\newcommand{\spinducedblackcompact}[1]%
    {\formatfunctiontoset{Bl} \bigl(#1 \bigr)}
\newcommand{\spinducedwhitecompact}[1]%
    {\formatfunctiontoset{Wh} \bigl(#1 \bigr)}

\newcommand{\pathclausedeg}[2][\pebdeg]%
    {\somenodetrueclausedeg[#1]{\vertexpath{#2}}}
\newcommand{\pathclauseNRdeg}[2][\pebdeg]%
    {\somenodetrueclausedeg[#1]{\vertexpathNR{#2}}}

\newcommand{\blacktruthdegexplicit}[4]%
        {\setdescrcompact
        {{\textstyle \Lor_{#2 = 1}^{#3} {#1}_{#2}}}
        {{#1} \in {#4}}}

\newcommand{\binsubtree}[1]{T^{#1}}

\newcommand{\vertexpath}[1]{{P}^{#1}}
\newcommand{\vertexpathNR}[1]{{P}_{*}^{#1}}

\newcommand{\unrelatedNP}[1]%
        {T \setminus \bigl(\binsubtree{#1} \unionSP \vertexpath{#1} \bigr)}
\newcommand{\unrelatedsmallNP}[1]%
        {T \setminus (\binsubtree{#1} \unionSP \vertexpath{#1} )}

\newcommand{\abovelevelblockerminsizecompact}%
    [2]{L_{\succeq{#1}}\bigl({#2}\bigr)}

\newcommand{\necessaryhidingvert}[2]%
{{#1}{\scriptstyle{\llfloor {#2} \rrfloor}}}

\newcommand{\Klawepropertyprefix}{Limited hiding-cardinality\xspace}

\newcommand{\klawepropacronym}{LHC property\xspace}

\newcommand{\nongenklaweprop}%
{non-generalized \Klawepropertyprefix property\xspace}
\newcommand{\nongenklawepropacronym}%
{non-generalized \klawepropacronym}
\newcommand{\nongenklawepropacronymWithParam}%
{(non-generalized) \klawepropacronym}

\newcommand{\siblingnonreachabiblitypropertynoref}%
{Sibling non-reachability property\xspace}
\newcommand{\Siblingnonreachabiblitypropertynoref}%
{Sibling non-reachability property\xspace}
\newcommand{\siblingnonreachabiblityproperty}%
{\siblingnonreachabiblitypropertynoref~%
\ref{property:sibling-non-reachability-property}\xspace}
\newcommand{\Siblingnonreachabiblityproperty}%
{\Siblingnonreachabiblitypropertynoref~%
\ref{property:sibling-non-reachability-property}\xspace}

\newcommand{\introducetermanmpctext}%
    {a \introduceterm{\mpctext{}}\xspace}

\newcommand{\introducetermamultipebblingtext}%
  {a \introduceterm{\multipebblingtext{}}\xspace}

\newcommand{\blobpebblingtext}{blob-pebbling\xspace}

\newcommand{\multipebblingtext}{\blobpebblingtext}

\newcommand{\mpcostblack}[1]%
        {\formatpebblingprice{cost}_{\mpcblacks}( #1 )}
\newcommand{\mpcostwhite}[1]%
        {\formatpebblingprice{cost}_{\mpcwhites}( #1 )}

\newcommand{\blobpebblingpricecompact}[1]%
    {\formatpebblingprice{Blob-Peb}\bigl(#1\bigr)}

\newcommand{\multipebblingpricecompact}[1]%
    {\formatpebblingprice{Blob-Peb}\bigl(#1\bigr)}

\newcommand{\mpcblacks}{\formatfunctiontomulti{B}}
\newcommand{\mpcwhites}{\formatfunctiontomulti{W}}

\newcommand{\mpscnotcompact}[2]%
        {\big[ {#1} \big] \bigl\langle {#2} \bigr\rangle}

\newcommand{\mpctext}{\blobpebblingtext con\-fig\-u\-ra\-tion\xspace}

\newcommand{\chargeablevertices}[1]%
{\formatfunctiontoset{chargeable}({#1}) }
\newcommand{\chargeableverticescompact}[1]%
{\formatfunctiontoset{chargeable}\bigl({#1}\bigr) }

\newcommand{\blackschargedfor}[1][]%
    {\mpcblacks_{#1}}
\newcommand{\whiteschargedfor}[1][]%
    {\mpcwhites_{#1}^{\hspace{-0.3 pt}\vartriangle}}

\newcommand{\whitesbelowjustblocked}%
    {\mpcwhites_{B}^{\hspace{-0.3 pt}\vartriangle}}
\newcommand{\whitesbelowhidden}%
    {\mpcwhites_{H}^{\hspace{-0.3 pt}\vartriangle}}

\newcommand{\whitestight}%
    {\mpcwhites_{T}^{\hspace{-0.3 pt}\vartriangle}}

\makeatletter{}%

\theoremstyle{plain}    

\newtheorem*{lem:expanderexists}{Lemma~\ref{lem:expanderexistnew} (restated)}
\newtheorem*{lem:closedset}{Lemma~\ref{lem:ClosedSet} (restated)}

\makeatletter{}%
\newcommand{\firstplayer}{Player~1\xspace}   %
\newcommand{\secondplayer}{Player~2\xspace}  %

\newcommand{\defi}{:=}

\newcommand{\XOR}{XOR\xspace}
\newcommand{\XORification}{XOR\-ification\xspace}
\newcommand{\xorification}{\XORification}

\newcommand{\xorified}{XORified\xspace}

\newcommand{\CNF}{CNF\xspace}

\newcommand{\FO}{\ensuremath{\mathsf L}}
\newcommand{\FOk}{\ensuremath{\mathsf L\!^{\pebblesk}}}

\newcommand{\FOcnt}{\ensuremath{\mathsf C}}
\newcommand{\FOcntk}{\ensuremath{\mathsf C^{\pebblesk}}}

\newcommand{\tuple}[1]{\ensuremath{\vec #1}}

\newcommand{\colourtype}{\ensuremath{\mathsf t}}
\newcommand{\wldim}{\ensuremath{k}} %
\newcommand{\wldimtext}{\ensuremath{k}} %

\newcommand{\xorwidth}{\ensuremath{\ell}}
\newcommand{\xorwidthalt}{m}
\newcommand{\Boolvala}{a}

\newcommand{\xformf}{\ensuremath{F}}
\newcommand{\xformg}{\ensuremath{H}}

\newcommand{\digraphtoxor}{\formulaformat{xor}}
\newcommand{\digraphxor}{\digraphtoxor(\digraph)}
\newcommand{\digraphxorarg}[1]{\digraphtoxor ( {#1} )}
\newcommand{\Digraphxorarg}[1]{\digraphtoxor \bigl( {#1} \bigr)}

\newcommand{\labelling}{\mathcal{M}}

\newcommand{\pyrquot}[2]{q_{#1}(#2)}
\newcommand{\pyrrem}[2]{r_{#1}(#2)}

\newcommand{\layernumber}{L}
\newcommand{\layernumberalt}{\layernumber'}
\newcommand{\rownumber}{\layernumber}
\newcommand{\rownumberalt}{\layernumberalt}
\newcommand{\pyheight}{h}
\newcommand{\setS}{S}

\newcommand{\dimension}{d}

\newcommand{\strucA}{\mathcal A}
\newcommand{\strucAelement}{u}
\newcommand{\strucB}{\mathcal B}
\newcommand{\strucBelement}{v}

\newcommand{\digraph}{\mathcal G}
\newcommand{\graph}{\mathcal G}
\newcommand{\graphalt}{\mathcal H}
\newcommand{\indegreenumber}{d}
\newcommand{\pyramid}{\mathcal P}
\newcommand{\binarytree}{\mathcal T}
\newcommand{\expandergraph}{\mathcal G}
\newcommand{\restrictedexpander}{\expandergraph'}

\newcommand{\EF}{Ehrenfeucht-Fra\"iss\'e\xspace}
\newcommand{\Spoiler}{Spoiler\xspace}
\newcommand{\Duplicator}{Duplicator\xspace}
\newcommand{\partmap}{p}
\newcommand{\partmapalt}{p'}
\newcommand{\bijection}{f}

\newcommand{\parammax}{\mathrm{hi}}
\newcommand{\parammin}{\mathrm{lo}}
\newcommand{\initialxorwidth}{p}

\newcommand{\pebblesk}{k} 
\newcommand{\pebblesl}{\ell} 

\newcommand{\indexn}{n}

\newcommand{\indexlhs}{m}

\newcommand{\indexrhs}{n}

\newcommand{\roundlowerbound}{r}

\newcommand{\leftsize}{n}  
\newcommand{\rightsize}{m}
\renewcommand{\leftsize}{\indexlhs}%
\renewcommand{\rightsize}{\indexrhs}%

\newcommand{\leftvertexset}{U}
\newcommand{\leftvertexsubset}{U'}
\newcommand{\rightvertexset}{V}
\newcommand{\rightvertexsubset}{V'}
\newcommand{\expansionguarantee}{s}
\newcommand{\expansionfactor}{c}
\newcommand{\expanderdegree}{\Delta}
\newcommand{\Ker}{\operatorname{Ker}}
\renewcommand{\ker}{\operatorname{Ker}}
\newcommand{\closure}{\gamma}
\newcommand{\expandersubgraph}[2]{#1\setminus #2}
\newcommand{\substituted}{[\expandergraph]}
\newcommand{\subst}[2]{#1[#2]}

\newcommand{\constraint}{c}
\renewcommand{\constraint}{C}
\newcommand{\constraintsubst}{\constraint\substituted}
\newcommand{\valuea}{a}
\newcommand{\valueb}{b}

\newcommand{\variables}{\operatorname{Vars}}
\renewcommand{\variables}{\mathit{Vars}}

\newcommand{\pebblestd}{k}
\renewcommand{\pebblestd}{\pebblesk}
\newcommand{\roundstd}{r} %

\newcommand{\xformsubst}{\xformf\substituted}
\newcommand{\xformorig}{\xformf}
\newcommand{\roundorig}{\roundstd}

\newcommand{\pos}{\alpha}
\newcommand{\posalt}{\beta}
\newcommand{\posorig}{\alpha}
\newcommand{\possubst}{\beta}
\newcommand{\startpossuffix}{\mathrm{start}}
\newcommand{\stoppossuffix}{\mathrm{end}}
\newcommand{\posorigstart}{\posorig_{\startpossuffix}}
\newcommand{\kstart}{\leftvertexset_{\startpossuffix}}
\newcommand{\posorigend}{\posorig_{\stoppossuffix}}
\newcommand{\kend}{\leftvertexset_{\stoppossuffix}}

\newcommand{\possubstpart}{\possubst'}
\newcommand{\possubstextended}{\possubst_{\mathrm{ext}}}

\newcommand{\posset}{\mathit{Cons}}

\newenvironment{subproof}[1][\proofname]{%
  \begin{proof}[#1]%
}{%
  \end{proof}%
}

\newcommand{\pebbleslower}{\pebblesk_{\parammin}}
\newcommand{\pebblesupper}{\pebblesk_{\parammax}}
\newcommand{\pebbleslowerlhs}{\pebblesl_{\parammin}}
\newcommand{\pebblesupperlhs}{\pebblesl_{\parammax}}

\newcommand{\pebblezero}{k_0}
\renewcommand{\pebblezero}{K_0}

\newcommand{\alternativetoell}{t}
\newcommand{\roundi}{i}

\newcommand{\leftvertexsubsetsize}{\ell}
\newcommand{\eulernumber}{2.72}
\renewcommand{\eulernumber}{e}
\newcommand{\eulernumberfourdigits}{2.718}
\newcommand{\eulernumbervariable}{e}

\newcommand{\smallepsilon}{\varepsilon}
\newcommand{\smallepsilonalt}{\smallepsilon'}
\newcommand{\nbhd}{\mathcal N}
\newcommand{\mindegree}{\expanderdegree_0}

\makeatletter{}%

\newcommand{\theauthorCB}{the first author\xspace}
\newcommand{\TheauthorCB}{The first author\xspace}

\newcommand{\theauthorJN}{the second author\xspace}

\ifthenelse{\boolean{conferenceversion}}
{
  \newcommand{\mysubsection}[1]{\paragraph{#1}}
  \renewcommand{\mysubsection}[1]{\textbf{\textit{#1}}}
  \renewcommand{\mysubsection}[1]{\paragraph{\textbf{#1}}}

}
{
  \newcommand{\mysubsection}[1]{\subsection{#1}}

}

\ifthenelse{\boolean{conferenceversion}}
{}
{
\newtheoremstyle{metacommenttheoremstyle}%
    {3pt}%
    {3pt}%
    {\sffamily \itshape \scriptsize
    }%
    {}%
    {\bfseries \scshape \footnotesize }%
    {:}%
    { }%
    {}%

\theoremstyle{metacommenttheoremstyle}
}
\newtheorem{jncommentcontainer}{Jakob's comment}
\newtheorem{cbcommentcontainer}{Christoph's comment}

\newcounter{rbcounter}
\newcommand{\randbem}[3]%
{\stepcounter{rbcounter}%
  \parbox[t]{0mm}{$^{\arabic{rbcounter}}$}%
  \marginpar%
  {\textcolor{#1}%
    {\raggedright\footnotesize$\mathbf{#2}^{\arabic{rbcounter}}$: #3}%
  }
}

\ifthenelse{\isundefined{\DONOTINSERTCOMMENTS}}
{ %

  \newcommand{\jncomment}[1]%
  {\begin{jncommentcontainer} \textcolor{blue}{#1} \end{jncommentcontainer}}

  \newcommand{\cbcomment}[1]%
  {\begin{cbcommentcontainer} \textcolor{magenta}{#1} \end{cbcommentcontainer}}

  \newcommand{\jn}[1]{\randbem{blue}{J}{#1}}
  \newcommand{\chr}[1]{\randbem{magenta}{C}{#1}}

}
{ %
  \newcommand{\jncomment}[1]{}
  \newcommand{\cbcomment}[1]{}
  \newcommand{\chr}[1]{}
  \newcommand{\jn}[1]{}
}

\numberwithin{equation}{section}

\begin{document}

\title{Near-Optimal Lower Bounds on Quantifier Depth and
    Weisfeiler--Leman Refinement Steps%
    \thanks{This is the full-length version of a paper with the same
      title which appeared in \emph{Proceedings of the 31st Annual ACM/IEEE 
        Symposium on Logic in Computer Science (LICS~'16)}.}}

\author{%
  Christoph Berkholz \\
  Humboldt-Universität zu Berlin
  \and
  Jakob Nordström \\
  KTH Royal Institute of Technology}

\date{\today}

\maketitle

\thispagestyle{empty}

\pagestyle{fancy}
\fancyhead{}
\fancyfoot{}
\renewcommand{\headrulewidth}{0pt}
\renewcommand{\footrulewidth}{0pt}

\fancyhead[CE]{\slshape 
  NEAR-OPTIMAL LOWER BOUNDS ON QUANTIFIER DEPTH 
}
\fancyhead[CO]{\slshape \nouppercase{\leftmark}}
\fancyfoot[C]{\thepage}

\setlength{\headheight}{13.6pt}

\makeatletter{}%
\ifthenelse{\boolean{conferenceversion}}
{%
  \begin{abstract}  
    We prove near-optimal trade-offs for quantifier depth versus number
    of variables in first-order logic by exhibiting pairs of $n$-element
    structures that can be distinguished by a $k$-variable first-order
    sentence but where every such sentence requires quantifier depth at
    least~$n^{\bigomega{k/\log k}}$. \mbox{Our trade-offs} also apply to
    first-order counting logic, and by the known connection to the
    $k$-dimensional Weisfeiler--Leman algorithm imply near-optimal lower
    bounds on the number of refinement iterations.
    A key component in our proof is the hardness condensation technique
    recently introduced by [Razborov~'16] in the context of proof
    complexity.  We apply this method to reduce the domain size of
    relational structures while maintaining the quantifier depth
    required to distinguish them.
  \end{abstract}
}
{%
  \begin{abstract}  
    We prove near-optimal trade-offs for quantifier depth versus number
    of variables in first-order logic by exhibiting pairs of $n$-element
    structures that can be distinguished by a $k$-variable first-order
    sentence but where every such sentence requires quantifier depth at
    least~$n^{\bigomega{k/\log k}}$. Our \mbox{trade-offs} also apply to
    first-order counting logic, and by the known connection to the
    $k$-dimensional Weisfeiler--Leman algorithm imply near-optimal lower
    bounds on the number of refinement iterations.
    
    A key component in our proof is the hardness condensation technique
    recently introduced by [Razborov~'16] in the context of proof
    complexity.  We apply this method to reduce the domain size of
    relational structures while maintaining the minimal quantifier depth
    to distinguish them in finite variable logics.
  \end{abstract}
}
\makeatletter{}%

\section{Introduction}
\label{sec:intro}

The $\pebblesk$-variable fragment of first-order logic 
\FOk{}
consists of
those first-order sentences that use at most $\pebblesk$~different variables.
A simple example is the  $\FO\!^2$ sentence 
\begin{equation}
  \label{eq:FOkExample}
  \exists x \exists y(Exy \wedge \exists x (Eyx\wedge \exists y
  (Exy\wedge \exists x Eyx))) 
\end{equation}
stating that there exists a directed path of length~$4$ in a
digraph. 
Extending~$\FOk$ with counting quantifiers~$\exists^{\geq i}x$
yields~$\FOcntk$, which can be more economical in terms of
variables. As an illustration,
the $\FO\!^8$~sentence 
\begin{equation}
  \exists x\exists y_1 \cdots \exists y_7 
  \bigl(
  \textstyle\bigwedge_{i\neq j} y_i\neq y_j 
  \land
  \textstyle\bigwedge_i Exy_i 
  \bigr)
\end{equation}
stating the existence of a vertex of degree at least 7 in a graph
can be written more succinctly 
as
the \mbox{$\FOcnt^2$ sentence} 
\ifthenelse{\boolean{conferenceversion}}
{\begin{equation} \exists x \exists^{\geq 7}y Exy \eqperiod \end{equation}}
{\begin{equation} \exists x \exists^{\geq 7}y Exy \eqperiod \end{equation}}
Bounded variable fragments of first order logic have found numerous
applications in finite model theory and related areas (see
\cite{Grohe.1998} for a survey).  Their importance stems from the fact
that the model checking problem (given a finite relational structure
$\strucA$ and a sentence $\varphi$, does $\strucA$ satisfy $\varphi$?)
can be decided in polynomial time \cite{Immerman.1982,Vardi.1995}.
Moreover, the 
equivalence problem (given two finite relational
structures $\strucA$ and $\strucB$, do they satisfy the same
sentences?) for \FOk{} and \FOcntk{}
can be decided in 
\mbox{time $n^{O(\pebblesk)}$ \cite{Immerman.1990}}, \ie polynomial for constant~$\pebblesk$. 

\mysubsection{Quantifier Depth}
If $\strucA$ and $\strucB$ are not equivalent in $\FOk$ or $\FOcntk$,
then there exists a sentence $\varphi$ that defines a distinguishing
property, 
\ie such that
 $\strucA\models\varphi$ and
$\strucB\not\models\varphi$, which certifies that the structures are
non-isomorphic.  
But how complex can such a sentence be?  
In particular, 
what is the minimal quantifier depth of an  \FOk{} or 
\mbox{\FOcntk{} sentence} 
that distinguishes two $n$-element relational structures
  $\strucA$ and $\strucB$?
The best upper bound for the quantifier depth of $\FOk$ and $\FOcntk$
is $n^{k-1}$ \cite{Immerman.1990}, while to the best of our knowledge
the strongest lower bounds have been only linear in $n$
\cite{Cai.1992,Grohe.1996,Furer.2001}.  
In this paper we present a near-optimal lower bound
of~$\indexn^{\Omega(\pebblesk / \log \pebblesk)}$.

\begin{theorem}
  \label{thm:maintheorem}
  There 
  exist
  $\varepsilon>0$, $\pebblezero\in\mathbb N$ such that for all 
  $\pebblesk,\indexn$ 
  with
  $\pebblezero\leq \pebblesk \leq \indexn^{1/12}$ 
  there is a pair of $\indexn$-element 
  \ifthenelse{\boolean{conferenceversion}}
  {$(\pebblesk\!-\!1)$-ary}
  {$(\pebblesk-1)$-ary}
  relational structures 
  $\strucA_\indexn, \strucB_\indexn$
  that 
  can be distinguished in
  $\pebblesk$-variable first-order logic
  but satisfy the same
  $\FO^\pebblesk$ and $\FOcnt^\pebblesk$ sentences up to
  quantifier depth
  $\indexn^{\varepsilon \pebblesk / \log \pebblesk}$. 
\end{theorem}

Note that any two non-isomorphic $n$-element $\sigma$-structures 
$\strucA$ and~$\strucB$
can always be distinguished by a simple $n$\nobreakdash-variable
first-order sentence of quantifier depth~$n$, namely
\ifthenelse{\boolean{conferenceversion}}
{%
\begin{multline}
\label{eq:distinguishing-formula}
  \exists x_1
  \cdots
  \exists x_n 
  \Biggl(%
    \bigwedge_{i\neq j}x_i\neq x_j
    \ \wedge 
    \!\!\!\!\
    \bigwedge_{\substack{R\in\sigma, \\ (v_{i_1},\ldots,v_{i_{r}})\in
        R^\strucA}}
    \!\!\!\!
    Rx_{i_1},\ldots,x_{i_{r}} \\
    \wedge
    \!\!\!\!
    \bigwedge_{\substack{R\in\sigma, \\ (v_{i_1},\ldots,v_{i_{r}})\notin
        R^\strucA}}
    \!\!\!\!
    \neg Rx_{i_1},\ldots,x_{i_{r}}
    \Biggr)%
  \eqperiod
\end{multline}
}
{%
\begin{equation}
\label{eq:distinguishing-formula}
  \exists x_1
  \cdots
  \exists x_n 
  \Biggl(%
    \bigwedge_{i\neq j}x_i\neq x_j
    \ \wedge \!\!\!\!
    \bigwedge_{\substack{R\in\sigma, \\ (v_{i_1},\ldots,v_{i_{r}})\in
        R^\strucA}}
    \!\!\!\!
    Rx_{i_1},\ldots,x_{i_{r}} 
    \ \wedge \!\!\!\!
    \bigwedge_{\substack{R\in\sigma, \\ (v_{i_1},\ldots,v_{i_{r}})\notin
        R^\strucA}}
    \!\!\!\!
    \neg Rx_{i_1},\ldots,x_{i_{r}}
    \Biggr)%
  \eqperiod
\end{equation}%
}
Since our $\indexn^{\Omega(\pebblesk / \log \pebblesk)}$ lower
bound for $\pebblesk$-variable logics grows significantly 
faster
than this
trivial upper bound~$n$ 
on the quantifier depth as the number of
variables increases,
\refth{thm:maintheorem}
also describes a trade-off in the super-critical regime
above worst-case investigated by 
Razborov~\cite{Razborov16NewKind}: 
If one reduces one complexity measure (the number of variables), then
the other complexity parameter (the quantifier depth) increases
sharply even beyond its worst-case upper bound.

The equivalence problem for $\FOcnt^{\wldim+1}$ is known to be closely
related to the
\emph{\wldimtext\nobreakdash-dimensional Weisfeiler--Leman
  algorithm} (\wldimtext\nobreakdash-WL) for testing non-isomorphism of graphs
and, more generally, relational structures.  It was shown by Cai,
Fürer, and Immerman \cite{Cai.1992} that two structures are
distinguished by \wldimtext-WL if and only if there exists a
$\FOcnt^{\wldim+1}$~sentence that
differentiates between them.
Moreover, the quantifier depth of such a sentence also relates to the
complexity of the WL~algorithm in that the number of iterations 
\wldimtext-WL needs to  
tell $\strucA$ and~$\strucB$  apart
coincides with the
minimal quantifier depth of a distinguishing
$\FOcnt^{\wldim+1}$~sentence.  
Therefore, \refth{thm:maintheorem} also implies a near-optimal
lower bound on the
number of refinement steps 
required in the Weisfeiler--Leman algorithm.
We discuss this next.

\mysubsection{The Weisfeiler--Leman Algorithm}
The Weisfeiler--Leman algorithm, independently introduced by Babai in
1979 and by Immerman and Lander in~\cite{Immerman.1990}
(cf.~\cite{Cai.1992} and~\cite{Babai16GraphIsomorphism} for historic notes), is a
hierarchy of methods for isomorphism testing  that iteratively refine
a partition 
(or colouring) of the vertex set, ending with a \emph{stable colouring} that
classifies \emph{similar vertices}. 
Since no isomorphism can
map non-similar vertices to each
other, this reduces the search space.
Moreover, if two structures end up with different stable colourings,
then we 
can immediately deduce
that the structures are non-isomorphic. 
The $1$\nobreakdash-dimensional Weisfeiler--Leman algorithm, 
better known as \emph{colour refinement}, 
initially colours the vertices according to their degree (clearly, no
isomorphism identifies vertices of different degree).  
The vertex colouring is then refined based on the colour classes of
the neighbours.
For example, two degree-$5$ vertices get different colours in the next
step if they have a different number of degree-$7$ neighbours. 
This refinement step is repeated until the colouring stays stable 
(\ie every pair of equally coloured vertices have the same number of
neighbours in every other colour class).  This algorithm is already
quite strong and is extensively used in practical graph isomorphism
algorithms.

In \wldimtext-dimensional WL this idea is generalized to colourings of
\emph{\wldimtext-tuples} of vertices.
Initially the \wldimtext\nobreakdash-tuples are coloured by their isomorphism
type, \ie two tuples $\tuple{v}=(v_1,\ldots,v_{\wldim})$ and
$\tuple{w}=(w_1,\ldots,w_{\wldim})$ get different colours if the
mapping $v_i\mapsto w_i$ is not an isomorphism 
on the substructures induced on 
$\{v_1,\ldots,v_{\wldim}\}$ and 
 $\{w_1,\ldots,w_{\wldim}\}$.
In the refinement step, we consider for each \wldimtext-tuple
$\tuple{v}=(v_1,\ldots,v_{\wldim})$ and every vertex $v$ the colours
of the tuples
$\tuple{v}_j\defi(v_1,\ldots,v_{j-1},v,v_{j+1},\ldots,v_{\wldim})$,
where 
$v$ is substituted at the $j$th position in the tuple~$\tuple{v}$.
We refer to  the tuple $(c(\tuple{v}_1),\ldots,c(\tuple{v}_\wldim))$
of these $\wldim$ colours  as the \emph{colour type}  
$\colourtype(\tuple{v},v)$ and let $v$ be a
$\colourtype$\nobreakdash-neighbour of 
$\tuple{v}$ if $\colourtype=\colourtype(\tuple{v},v)$. 
Now two tuples $\tuple{v}$ and~$\tuple{w}$ get different colours if
they are already coloured differently, or if there exists a colour
type~$\colourtype$ such that $\tuple{v}$ and $\tuple{w}$ have a
different number of $\colourtype$\nobreakdash-neighbours. 
The refinement step is repeated until the colouring stays stable. 
Since in every round the number of colour classes grows, the
process stops after at most $n^{\wldim}$ steps.  
The colour names can be chosen in such a way that the stable colouring
is canonical, which means that two isomorphic structures end up with
the same colouring, and such a 
canonical stable colouring
can be computed in time~$n^{O(\wldim)}$.

This simple combinatorial algorithm is surprisingly powerful.
Grohe \cite{Grohe12FixedPointDefinability} showed that for every
nontrivial graph class that excludes some minor (such as planar
graphs or graphs of bounded treewidth) there exists some $\pebblesk$ such
that $\pebblesk$-WL computes a different colouring for all non-isomorphic
graphs, and hence solves graph isomorphism in polynomial time on that
graph class. Weisfeiler--Leman has also been used as a subroutine in
algorithms that solve graph isomorphism on all graphs. 
As one part of his very recent graph isomorphism algorithm, 
Babai~\cite{Babai16GraphIsomorphism} applies $\pebblesk$-WL for
polylogarithmic~$\pebblesk$
to relational ($\pebblesk$-ary) structures and makes use of 
the  quasi-polynomial 
running time of this algorithm.

Given the importance of 
the Weisfeiler--Leman procedure,
it is a natural question to
ask whether the trivial $n^{\wldim}$~upper bound on the number of
refinement steps is tight.  
By the correspondence between the number of refinement steps of
\wldimtext-WL and the quantifier depth 
of~$\FOcnt^{\wldim+1}$~\cite{Cai.1992}, our main result implies a
near-optimal lower bound even up to 
polynomial, but still sublinear, values of~$\wldim$
(\ie $\wldim = \indexn^\delta$ for small enough constant~$\delta$).  

\begin{theorem}\label{thm:mainWLtheorem}
  There 
  exist
  $\varepsilon>0$, $\pebblezero\in\mathbb N$ such that for all 
  $\pebblesk,\indexn$ with
  $\pebblezero\leq \pebblesk \leq \indexn^{1/12}$ 
  there is an $\indexn$-element $\wldim$-ary relational structure
  $\strucA_\indexn$
  for which
  the $\wldim$-dimensional Weisfeiler--Leman algorithm
  needs $\indexn^{\varepsilon
    \wldim / \log \wldim}$ refinement steps to compute the
  stable colouring. 
  \end{theorem}

In addition to the near-optimal lower bounds for a specific dimension
(or number of variables)~$\wldim$, we also obtain the following
trade-off between the dimension and the number of refinement steps:   
If we fix two parameters $\ell_1$ and $\ell_2$ (possibly depending
on~$\indexn$) satisfying 
$\ell_1\leq\ell_2 \leq \indexn^{1/6}/\ell_1$,
then there are $\indexn$\nobreakdash-element structures such that $\wldim$-WL
needs $\indexn^{\bigomega{\ell_1/\log \ell_2}}$ refinement steps for
all $\ell_1\leq\wldim\leq\ell_2$. 
A particularly interesting choice of parameters is $\ell_1 = \log^c\indexn$
for some constant $c>1$
and $\ell_2 = \indexn^{1/7}$. 
This implies the following quasi-polynomial lower bound on the number
of refinement steps for Weisfeiler--Leman from polylogarithmic
dimension all the way up to dimension~$\indexn^{1/7}$.

\begin{theorem}\label{thm:mainTheoremWLquasipolyTradeoff}
  For every $c>1$ 
  there is a sequence of $\indexn$-element relational structures
  $\strucA_\indexn$
  for which
  the $\wldim$-dimensional Weisfeiler--Leman algorithm needs
  $n^{\bigomega{\log^{c-1} n}}$ refinement steps to compute the stable
  colouring for all $\wldim$ with $\log^c n \leq \wldim \leq n^{1/7}$.  
\end{theorem}

\mysubsection{Previous Lower Bounds}

In their seminal work~\cite{Cai.1992}, Cai, Fürer and Immerman 
established the existence of non-isomorphic \mbox{$n$-vertex} graphs 
that cannot be distinguished by any first-order counting sentence with
$\littleoh{n}$~variables.
Since every pair of non-isomorphic $n$-element structures can be
distinguished by a $\FOcnt^n$ (or even~$\FO\!^n$) sentence
(as shown in \refeq{eq:distinguishing-formula} above), this result
also implies a linear lower bound on the quantifier depth of $\FOcntk$
\mbox{if $k=\Omega(n)$.}  
For all constant $k\geq 2$, a linear $\Omega(n)$ lower bound on the
quantifier depth of $\FOcntk$ follows implicitly from an
intricate
construction of Grohe \cite{Grohe.1996}, which was used to show that
the equivalence 
problems for $\FOk$ and $\FOcntk$ are complete for
polynomial time. 
An explicit linear lower bound based on a simplified construction was
subsequently 
presented
by Fürer~\cite{Furer.2001}.  

For the special case of $\FOcnt^2$, Krebs and
Verbitsky~\cite{Krebs.2015} recently obtained an improved $(1-o(1))n$
lower bound on the quantifier depth, nearly matching the upper bound~$n$.  
In contrast,
Kiefer and Schweitzer~\cite{Kiefer.2016}
showed that if two \mbox{$n$-vertex}
graphs
can be distinguished by 
a
$\FOcnt^3$~sentence, 
then there is always a distinguishing sentence of
quantifier depth $O(n^2/\log n)$.  Hence, the trivial \mbox{$n^2$~upper}
bound is not tight in this case.
{As far as we are aware, the current paper presents the first lower
  bounds that are super-linear in the domain size~$n$.}

\mysubsection{Discussion of Techniques}
The hard instances we construct are based on
propositional \XOR{} (exclusive or) formulas,
which can  alternatively be viewed as systems of linear equations
over~$\gf{2}$.  
There is a long history of using \XOR formulas
for proving lower bounds in 
different
areas of theoretical
computer science such as, e.g., finite model theory, 
proof complexity, and combinatorial optimization/hardness of
approximation.
Our main technical insight is to combine two methods that, to the
best of our knowledge, have not been used together before, namely
\EF games on structures based on \XOR formulas  
and hardness amplification by variable substitution. 

More than three decades ago, Immerman~\cite{Immerman.1981} presented a way to
encode an \XOR formula into two graphs that are isomorphic if and only if 
the formula is satisfiable. 
This can then be used to show that the
two graphs cannot be distinguished by a sentence with few variables or
low quantifier depth using \EF games.
Arguably the most important application of this method is the result
in~\cite{Cai.1992} establishing 
that a linear number of variables is
needed to distinguish two  
graphs in first-order counting logic. 
Graph constructions based on \XOR formulas have also been used to
prove lower bounds on the quantifier depth of  
$\FOcntk$ \cite{Immerman.1981,Furer.2001}.
We remark that for our result
we have to use a slightly different encoding of 
\ifthenelse{\boolean{conferenceversion}}
{formulas}  
{\XOR formulas}  
into relational structures rather than graphs.

In proof complexity, various flavours of \XOR formulas 
(usually called \introduceterm{Tseitin formulas} when used to encode the 
\introduceterm{handshaking lemma} 
saying that the sum of all vertex degrees in an
undirected graph has to be an even number) 
have been 
employed
to obtain lower bounds for
proof systems such as  
resolution~\cite{Urquhart87HardExamples},
polynomial calculus~\cite{BGIP01LinearGaps},
and
bounded-depth Frege~\cite{Ben-Sasson02HardExamples}.
Such formulas have also played an important role in many lower bounds for
the Positivstellensatz/sums-of-squares proof system
\cite{Grigoriev01LinearLowerBound,KI06LowerBounds,Schoenebeck08LinearLevel}  
corresponding to the Lasserre semidefinite programming 
hierarchy, which has been the focus of much recent interest in the
context of combinatorial optimization.%
\footnote{%
  No proof complexity is needed in this paper, 
  and so readers unfamiliar with these proof systems need not
  worry---this is just 
  an informal overview.} 
Another use of \XOR in proof complexity has been for hardness amplification,
where one takes a (typically non-\XOR) formula that is moderately hard with
respect to some complexity measure, substitutes all variables by
exclusive ors
over pairwise distinct sets of variables, 
and then shows that the new
\introduceterm{\xorified{}} formula must be very hard \wrt some
other (more important) complexity measure.
This technique was perhaps first made explicit
in~\cite{Ben-Sasson02SizeSpaceTradeoffsJOURNALREF}  
(attributed there to personal communication with 
Michael~Alekhnovich and Alexander~Razborov, with a note that it is also 
very similar in spirit to an approach used in \cite{BW01ShortProofs}) 
and has later appeared in, e.g., 
\cite{BP07Complexity,BN08ShortProofs,BN11UnderstandingSpace,BNT13SomeTradeoffs,FLMNV13TowardsUnderstandingPC}.

An even more crucial role in proof complexity is played by well-connected
so-called \introduceterm{expander graphs}.
For instance, given a 
formula in conjunctive normal form (CNF)
one can look at its bipartite clause-variable incidence graph (CVIG),
or some variant of the CVIG derived from the combinatorial structure
of the formula, and prove that if this graph is an expander, then this
implies that the formula must be hard for proof systems such as
resolution~\cite{BW01ShortProofs} and polynomial
calculus~\cite{AR03LowerBounds,MN15GeneralizedMethodDegree}.

In a striking recent paper~\cite{Razborov16NewKind}, Razborov combines 
\xorification
and expansion in a simple (with hindsight)
but amazingly powerful way. Namely, instead of replacing every
variable by an XOR over new, fresh variables, 
he recycles variables from a much smaller pool, thus decreasing the
total number of variables.
This means that the hardness amplification proofs no longer work,
since they crucially use that all new substitution variables
are distinct.
But here expansion come into play. If the pattern of variable
substitutions is described by a strong enough bipartite expander, it
turns out that locally there is enough ``freshness'' even among the
recycled variables to make the hardness amplification go through over
a fairly wide range of the parameter space. 
And since the formula has not only become harder but has also had the
number of variables decreased, 
this can be viewed as a kind of 
\introduceterm{hardness compression}
or
\introduceterm{hardness condensation}.

What we do in this paper is to first revisit
Immerman's old quantifier depth lower bound for first-order counting
logic~\cite{Immerman.1981} and observe that the construction can be
used to obtain an improved scalable lower bound for the
\mbox{$\pebblesk$-variable} fragment.  We then translate Razborov's hardness
condensation technique~\cite{Razborov16NewKind} into the language of
finite variable logics and use it---perhaps somewhat amusingly applied
to \xorification of \XOR formulas, which is usually not the case in
proof complexity---to reduce the domain size of relational structures
while maintaining the minimal quantifier depth required to distinguish
them.

\mysubsection{Outline of This Paper}
The rest of this paper is organized as follows.
In \refsec{sec:prelims} we describe how to translate \XOR formulas to
relational structures and play combinatorial games on these
structures.
This then allows us to state our main technical lemmas
in \refsec{sec:lower-bound-proof} and show how these lemmas yield our
results.
Turning to the proofs of these technical lemmas,  
in \refsec{sec:pyramids}
we present a version of Immerman's quantifier depth lower bound for
\XOR formulas, and in
\refsec{sec:hardness-condensation}
we apply Razborov's hardness condensation technique to these formulas.
Finally, in
\refsec{sec:conclusion}
we make some concluding remarks and discuss possible directions for
future research.
\ifthenelse{\boolean{conferenceversion}}
{Due to space constraints, we omit some of the more standard technical
  proofs in this conference version, referring the reader to the
  upcoming full-length version for the missing details.}
{Some proofs of technical results needed in the paper are deferred to
  \refapp{app:existence-expander}.}
\makeatletter{}%

\section{From \XOR Formulas to Relational Structures}
\label{sec:prelims}

In this paper all structures are finite and defined over a relational
signature $\sigma$.   
We use the letters $X$,~$E$, and $R$ for unary, binary, and \mbox{$r$-ary}
relation symbols, respectively, and let 
$X^\strucA$, $E^\strucA$, and~$R^\strucA$ be their interpretation in a
structure $\strucA$.  
We write~$V(\strucA)$ to denote 
the domain of the structure~$\strucA$. 
The \emph{$\pebblesk$\nobreakdash-variable 
fragment of first-order logic}~$\FOk$ consists
of all first-order formulas that use at most $\pebblesk$~different variables
(possibly re-quantifying them as in Equation~\eqref{eq:FOkExample}).  
We also consider \emph{$\pebblesk$-variable first-order counting logic}
$\FOcntk$, which is the extension of $\FOk$ by counting quantifiers
$\exists^{\geq i}x \varphi(x)$,  stating that there exist at least $i$
elements $\strucAelement\in V(\strucA)$ such that
$(\strucA,\strucAelement)\models \varphi(x)$. 
For a survey of finite variable logics and their applications we refer
the reader to, e.g.,~\cite{Grohe.1998}. 

An \introduceterm{\xorwidth-\XOR{} clause} is a tuple
$(x_1,\ldots,x_\xorwidth,\Boolvala)$ consisting of $\xorwidth$~Boolean
variables and a Boolean value $\Boolvala\in\{0,1\}$. 
We refer to $\xorwidth$ as the \introduceterm{width} of the clause. 
An assignment~$\pos$ \emph{satisfies}
$(x_1,\ldots,x_\xorwidth,\Boolvala)$ 
if $\pos(x_1) + \cdots + \pos(x_\xorwidth) \equiv \Boolvala \pmod{2}$.  
An $\xorwidth$-\XOR{} formula~$\xformf$ is a conjunction of \XOR~clauses of
width at most~$\xorwidth$ and is satisfied by
an assignment~$\pos$  if $\pos$ satisfies all
clauses in~$\xformf$. 

For every $\xorwidth$-\XOR{} formula $\xformf$ on $n$~variables we can
define a pair of $2n$-element structures \mbox{$\strucA=\strucA(\xformf)$}
and $\strucB=\strucB(\xformf)$ that are isomorphic \ifaoif $\xformf$ is
satisfiable. 
The domain of the structures contains two elements $x_i^0$ and $x^1_i$
for each Boolean variable $x_i$.  
There is one unary
 predicate~$X_i$ for every variable~$x_i$ identifying 
the corresponding two elements  $x_i^0$ and~$x^1_i$.
 Hence these unary relations partition the domain of the structures into
 two-element sets, \ie $X_i^\strucA = X_i^\strucB =
\{x_i^0,x_i^1\}$. 
To encode the \XOR clauses, we introduce one $\xorwidthalt$-ary
relation $R_\xorwidthalt$ for every $1\leq \xorwidthalt \leq
\xorwidth$ and set  
\begin{subequations}
\begin{align}
  \label{eq:structure-A}
  \ifthenelse{\boolean{conferenceversion}}
  {\!R_\xorwidthalt^\strucA \!&=\!}                           
  {R_\xorwidthalt^\strucA &=}                           
  \Setdescr{\bigl(x^{\Boolvala_1}_{i_1},\ldots,x^{\Boolvala_\xorwidthalt}_{i_\xorwidthalt}\bigr)} 
  {(x_{i_1},\ldots,x_{i_\xorwidthalt},\Boolvala)
  \ifthenelse{\boolean{conferenceversion}}
  {\!\in\! \xformf, \textstyle\sum_i\! \Boolvala_i \!\equiv\! 0 }
  {\in \xformf,\;\textstyle\sum_i \Boolvala_i \equiv 0 \pmod{2}}} 
  \\ 
\shortintertext{and}
  \label{eq:structure-B}
  \ifthenelse{\boolean{conferenceversion}}
  {\!R_\xorwidthalt^\strucB \!&=\!}
  {R_\xorwidthalt^\strucB &=}
  \Setdescr{\bigl(x^{\Boolvala_1}_{i_1},\ldots,x^{\Boolvala_\xorwidthalt}_{i_\xorwidthalt}\bigr)}
  {(x_{i_1},\ldots,x_{i_\xorwidthalt},\Boolvala)
  \ifthenelse{\boolean{conferenceversion}}
  { \!\in\! \xformf, \textstyle\sum_i\! \Boolvala_i \!\equiv\! \Boolvala }  
  { \in \xformf,\;\textstyle\sum_i \Boolvala_i \equiv \Boolvala \pmod{2}}} 
  \ifthenelse{\boolean{conferenceversion}}
  {}                         
  {\eqperiod}
\end{align}
\end{subequations}
\ifthenelse{\boolean{conferenceversion}}
{(where the sums are taken $\bmod{}\ 2$). Every}
{Every}
bijection $\beta$ between the domains of $\strucA(\xformf)$ and
$\strucB(\xformf)$ that preserves the unary relations~$X_i$ 
can be translated
to an assignment $\pos$ for the \XOR{} formula via the correspondence
\ifthenelse{\boolean{conferenceversion}}
{$\pos(x_i) = 0 \Leftrightarrow \beta(x_i^0) = x_i^0 \Leftrightarrow
  \beta(x_i^1) = x_i^1$ and  $\pos(x_i) = 1 \Leftrightarrow \beta(x_i^0)
  = x_i^1 \Leftrightarrow \beta(x_i^1) = x_i^0$.} 
{\begin{subequations}
    \begin{gather}
      \label{eq:correspondence-1}
      \pos(x_i) = 0 \Leftrightarrow \beta(x_i^0) = x_i^0 \Leftrightarrow
      \beta(x_i^1) = x_i^1
      \\
      \shortintertext{and}
      \label{eq:correspondence-2}
      \pos(x_i) = 1 \Leftrightarrow \beta(x_i^0)
      = x_i^1 \Leftrightarrow \beta(x_i^1) = x_i^0
      \eqperiod
    \end{gather}
  \end{subequations}
} 
\ifthenelse{\boolean{conferenceversion}}
{It is} 
{Moreover, it is} 
not hard to show that such a bijection defines an
isomorphism between 
$\strucA(\xformf)$ and $\strucB(\xformf)$ if and only if the 
corresponding assignment satisfies $\xformf$.
\ifthenelse{\boolean{conferenceversion}}
{}
{See \reffig{fig:xor_encoding_example} for a small example
  illustrating the construction.

\makeatletter{}%

\begin{figure}[t] 
 \centering
  \begin{tikzpicture}[vertex/.style={circle,draw=black,fill=black,inner  sep=0pt,minimum  size=4pt},bluevertex/.style={rectangle,draw=black,fill=blue,inner  sep=0pt,minimum  size=4pt},redvertex/.style={circle,draw=black,fill=red,inner  sep=0pt,minimum  size=4pt},
   diredge/.style={->,shorten <=.5pt, shorten >=.5pt}] %
  
  \node[bluevertex,label=above:{\small $x_7^0$}] (x0a) at (0,2) {};
  \node[bluevertex,label=above:{\small $x_7^1$}] (x1a) at (.5,2) {};
  \node[redvertex,label=below:{\small $x_8^0$}] (y0a) at (0,1.25) {};
  \node[redvertex,label=below:{\small $x_8^1$}] (y1a) at (.5,1.25) {};
 
 \node at (.25,0.25) {$\strucA$};

  \node[bluevertex,label=above:{\small $x_7^0$}] (x0b) at (0+2,2) {};
  \node[bluevertex,label=above:{\small $x_7^1$}] (x1b) at (.5+2,2) {};
  \node[redvertex,label=below:{\small $x_8^0$}] (y0b) at (0+2,1.25) {};
  \node[redvertex,label=below:{\small $x_8^1$}] (y1b) at (.5+2,1.25) {};

 \node at (2.25,0.25) {$\strucB$};

  \draw[diredge] (x0a) -- (y0a);
  \draw[diredge] (x1a) -- (y1a);
  \draw[diredge] (x0b) -- (y1b);
  \draw[diredge] (x1b) -- (y0b);

\end{tikzpicture}
  \caption{Structure encoding of $\xformf=\{(x_7,x_8,1)\}$.}
  \label{fig:xor_encoding_example}
\end{figure}
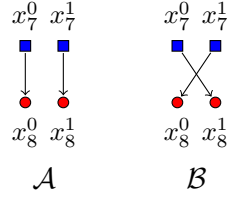

}

This kind of
encodings of \XOR formulas into relational structures have been
very useful for proving lower bounds for finite variable logics in the
past.  
Our transformation of
\XOR clauses of width~$\xorwidth$ into
$\xorwidth$\nobreakdash-ary relational structures resembles the way 
Gurevich and Shelah~\cite{Gurevich.1996} encode \XOR formulas as hypergraphs. 
It is also closely related to the way Cai, Fürer, and
Immerman~\cite{Cai.1992} obtain two non-isomorphic graphs $\graph$
and~$\graphalt$  
from
an unsatisfiable \mbox{$3$-\XOR} formula $\xformf$
in the sense that $\graph$ and $\graphalt$ can be seen to be the
incidence  graphs of our structures  $\strucA(\xformf)$
and~$\strucB(\xformf)$.

In order to prove our main result, we make use of the combinatorial
characterization of quantifier depth of finite-variable logics  in
terms of pebble games for $\FOk$ and $\FOcntk$, 
which are played on two given relational structures. 
Since in our case the structures are based on \XOR{} formulas,
\ifthenelse{\boolean{conferenceversion}}
{for convenience we consider a simplified combinatorial game that is
  played directly on the formulas} 
{for convenience we will consider a simplified combinatorial game that is
  played directly on the \XOR{} formulas} 
rather than on their structure encodings. We first describe this game and then 
show
in \reflem{lem:EquivalentCharacterisations} that this
yields an equivalent characterization. 

The 
\introduceterm{$\roundstd$-round $\pebblesk$-pebble game} is played on an \XOR
formula $\xformf$ by two players,
whom we will refer to as
\firstplayer and \secondplayer.  
A~position in the game is a partial assignment~$\pos$ of at most
$\pebblesk$~variables of~$\xformf$ and the game starts with the empty
assignment.  In each round, \firstplayer can delete some variable
assignments from the current position (he chooses some
$\pos'\subseteq\pos$).  If the current position assigns values to
exactly $\pebblesk$~variables, then \firstplayer has to delete at
least one variable assignment.  Afterwards, \firstplayer chooses 
some currently unassigned variable~$x$ and 
asks for its value.
\secondplayer answers by either $0$ or~$1$ 
(independently of any previous answers to the same question)
and adds this 
\ifthenelse{\boolean{conferenceversion}}
{assignment} 
{variable assignment} 
to the current position.

A winning position for \firstplayer is an assignment falsifying some
clause from $\xformf$. 
\firstplayer wins the $\roundstd$-round $\pebblesk$-pebble game if
he has a strategy to win every play of the $\pebblesk$-pebble game
within at most $\roundstd$~rounds. 
Otherwise, we say that \secondplayer wins (or survives) the
$\roundstd$-round $\pebblesk$-pebble game. 
\firstplayer \emph{wins the $\pebblesk$-pebble game} if he wins the
$\roundstd$-round $\pebblesk$-pebble game within a finite number of
rounds~$\roundstd$. 
Note that if \firstplayer wins the $\pebblesk$-pebble game, then he
can always win the $\pebblesk$-pebble 
within $2^\pebblesk n^{\pebblesk+1}$~rounds, 
because there are 
at most 
$
\sum_{i=0}^{\pebblesk} 2^i \binom{n}{i} \leq
2^\pebblesk n^{\pebblesk+1}
$
different positions with at most $\pebblesk$~pebbles on
$n$\nobreakdash-variable \XOR formulas.    
We say that \firstplayer \emph{can reach a position} $\posalt$ from a
position $\pos$ within $\roundstd$ rounds if he has a strategy such
that in every play of the $\roundstd$\nobreakdash-round
$\pebblesk$\nobreakdash-pebble game starting from position $\pos$ he
either wins or ends up with position~$\posalt$. 

\ifthenelse{\boolean{conferenceversion}}
{}
{%
  As a side remark, we
  note that if we expand the \XOR formula 
  to 
  \CNF,
  then our pebble game is the same as the 
  so-called
  \emph{Boolean existential pebble game} played on 
  this
  \CNF encoding and
  therefore also characterizes the resolution width 
  required for
  the corresponding
  \CNF formula as shown in~\cite{AD08CombinatoricalCharacterization}.  
  Intuitively, it is this correspondence that enables us to apply the
  proof complexity techniques from~\cite{Razborov16NewKind} in our
  setting.  We will not need to use any concepts from proof complexity
  in this paper, however, but will present a self-contained proof, and
  so we do not elaborate further on this connection.
} %
  
Let us now show that the game described above is equivalent to the
pebble game for $\FO^{\pebblesk}$ and to the bijective pebble game for
$\FOcnt^{\pebblesk}$ played on the structures $\strucA(\xformf)$ and
$\strucB(\xformf)$. 

\begin{lemma}\label{lem:EquivalentCharacterisations}
  Let $k, \initialxorwidth, r$ be integers such that
  $r>0$ and $k\geq \initialxorwidth$ 
  and  let $\xformf$ be 
  a  $\initialxorwidth$-\XOR formula 
  giving rise to structures
  $\strucA=\strucA(\xformf)$
  and
  $\strucB=\strucB(\xformf)$
  as described in the paragraph preceding  
  \refeq{eq:structure-A}--\refeq{eq:structure-B}.
  Then the following statements are equivalent: 
  \begin{enumerate}[label=(\alph*)]
  \item
    \label{item:equiv-char-1}
    Player 1 wins the $\roundstd$-round $\pebblesk$-pebble game on
    $\xformf$.  
  \item
    \label{item:equiv-char-2}
    There is a $\pebblesk$-variable first-order sentence $\varphi\in
    \FO^{\pebblesk}$ of quantifier depth $\roundstd$ such that
    $\strucA(\xformf)\models \varphi$ and $\strucB(\xformf)\not\models
    \varphi$. 
  \item
    \label{item:equiv-char-3}
    There is a $\pebblesk$-variable sentence in first-order counting
    logic $\varphi\in \FOcnt^{\pebblesk}$ of quantifier depth $\roundstd$ such
    that $\strucA(\xformf)\models \varphi$ and
    $\strucB(\xformf)\not\models \varphi$. 
  \item
    \label{item:equiv-char-4}
    The $(\pebblesk-1)$-dimensional Weisfeiler--Leman procedure 
    can
    distinguish between $\strucA(\xformf)$ 
    and~$\strucB(\xformf)$ within $\roundstd$ refinement steps. 
  \end{enumerate}
\end{lemma}

\begin{proof}[Proof sketch]
  Let us start by briefly recalling known characterizations in terms
  of \EF games  of
  $\FOk$~\cite{Barwise.1977,Immerman.1982} 
  and~$\FOcntk$~\cite{Cai.1992,Hella.1996}.
  In both cases the game is played by two players,
  \ifthenelse{\boolean{conferenceversion}}
  {referred to as}
  {called} 
  \Spoiler and
  \Duplicator, on the two structures $\strucA$ and~$\strucB$.
  Positions in the games are partial mappings 
  $\partmap =
  \Set{(\strucAelement_1,\strucBelement_1), \ldots, 
    (\strucAelement_i,\strucBelement_i)}$  
  from $V(\strucA)$ to $V(\strucB)$ of size at most~$\pebblesk$.  
  The games start from the empty position and proceed in rounds.
  At the beginning of each round in both games, \Spoiler chooses
  $\partmapalt\subseteq \partmap$ with
  $\setsize{\partmapalt}<\pebblesk$.  
  
 \begin{itemize}
   \item 
     In the $\FOk$-game, \Spoiler then 
     selects
     either some
     $\strucAelement\in V(\strucA)$ or some $\strucBelement\in
     V(\strucB)$ and \Duplicator responds by choosing an element
     $\strucBelement\in V(\strucB)$ or $\strucAelement\in V(\strucA)$
     in the other structure. 
    \item
      In the $\FOcntk$-game, \Duplicator first 
      selects
      a global
      bijection \mbox{$\bijection:V(\strucA)\to V(\strucB)$} and \Spoiler
      chooses some pair $(\strucAelement, \strucBelement)\in
      \bijection$.  
      (If $\setsize{V(\strucA)} \neq \setsize{V(\strucB)} $, 
      \Spoiler
      wins the $\FOcntk$-game immediately.)
 \end{itemize}
 The new position is 
 $\partmapalt\cup\{(\strucAelement, \strucBelement)\}$. 
 Spoiler wins the $\roundstd$-round $\FOk\,/\,\FOcntk$ game if he has
 a strategy to reach within $\roundstd$~rounds 
 \ifthenelse{\boolean{conferenceversion}}
 {a position}
 {a position~$\partmap$}
 that does not define an isomorphism on the induced substructures.  
 Both games characterize equivalence in the corresponding logics:
 Spoiler wins the $\roundstd$-round $\FOk\,/\,\FOcntk$ game if and only
 if  there is a 
 \ifthenelse{\boolean{conferenceversion}}
 {\mbox{depth-$\roundstd$} sentence $\varphi\in \FOk\,/\,\FOcntk$}
 {sentence $\varphi\in \FOk\,/\,\FOcntk$ of quantifier
   depth~$\roundstd$} 
 such that $\strucA\models \varphi$ and
 $\strucB\not\models \varphi$.  

When these games are played on the two structures $\strucA(\xformf)$
and~$\strucB(\xformf)$ obtained from an \XOR formula~$\xformf$, it is
not  hard to  
verify
that both games are 
equivalent to the $\pebblesk$-pebble game on $\xformf$.
To see this,
we identify \Spoiler with \firstplayer, \Duplicator with
\secondplayer, and partial mappings 
\mbox{$\partmap=\setdescr{(x^{a_i}_i,x^{b_i}_i) }{ i\leq \ell}$} 
with partial assignments 
$\pos = \setdescr{x_i\mapsto a_i\oplus b_i}{i\leq\ell}$.  
Because of the $X_i$-relations, we can assume that partial assignments
of any other form will not occur as they are losing positions for
\Duplicator.

If \Spoiler asks for some $x_i^0$ or $x_i^1$ in the
\mbox{$\FOk$-game}, which corresponds 
to a choice by \firstplayer of
$x_i\in\variables(\xformf)$, the only meaningful 
action
for \Duplicator is to choose either $x_i^0$ or $x_i^1$ in the other
structure, corresponding to an assignment to~$x_i$ by \secondplayer.  
With any other choice \Duplicator would lose immediately because of the unary
relations~$X_i$.
Thus, there is a natural correspondence between strategies in the 
\mbox{$\FOk$-game} and the $\pebblesk$\nobreakdash-pebble game.

The players in the $\pebblesk$\nobreakdash-pebble game can be
assumed to have perfect knowledge of the strategy of the other
player. This means that at any given position in the game, without
loss of generality we can think of \firstplayer as being given a
complete truth value assignment to the remaining variables, out of
which he can pick one variable assignment. By the correspondence 
\ifthenelse{\boolean{conferenceversion}}
{discussed above}
{in \refeq{eq:correspondence-1}--\refeq{eq:correspondence-2}}
we see that this can be translated to a bijection $\bijection$
chosen by \Duplicator in the $\FOcntk$-game (which has to preserve the
$X_i$ relations). Therefore, \Spoiler picking some pair of the form
$(x_i^a,x_i^b)$ from $\bijection$ can be viewed as \firstplayer asking
about the assignment to~$x_i$ and getting a response from
\secondplayer in the game on~$\xformf$  
(again using the above-mentioned correspondence between partial mappings
$\partmap$ and partial assignments $\pos$).
Finally, we observe that by design a partial mapping that preserves
the $X_i$-relations defines a local isomorphism if and only if the corresponding
$\pos$ does not falsify any \XOR clause. 

Formalizing the proof sketch above, it is not hard to show that
\ifthenelse{\boolean{conferenceversion}}
{statements \ref{item:equiv-char-1}--\ref{item:equiv-char-3}}
{statements \ref{item:equiv-char-1}--\ref{item:equiv-char-3}
  in
  the lemma} 
are all equivalent.  The equivalence between 
\ref{item:equiv-char-3} and \ref{item:equiv-char-4}
was proven in~\cite{Cai.1992}. 
The lemma follows.
\end{proof} 

\makeatletter{}%

\ifthenelse{\boolean{conferenceversion}}
{\section{Proofs of Main Theorems}}
{\section{Technical Lemmas and Proofs of Main Theorems}}
\label{sec:lower-bound-proof}

To prove our lower bounds of the quantifier depth of finite variable
logics in \refth{thm:maintheorem} and the number of refinement steps
of the Weisfeiler--Leman algorithm 
in \reftwoths{thm:mainWLtheorem}{thm:mainTheoremWLquasipolyTradeoff},
we utilize the characterization in
\reflem{lem:EquivalentCharacterisations} and show that there are 
$\indexn$\nobreakdash-variable \XOR
formulas on which \firstplayer is able to win the $\pebblesk$-pebble
game but cannot do so in significantly less than
$\indexn^{\pebblesk/\log \pebblesk}$~rounds.  
The next lemma states this formally and also provides a trade-off as
the number of pebbles increases.

\begin{lemma}[Main technical lemma]
  \label{lem:MainTheoremXor}
  There is an absolute constant
  $\pebblezero \in \Nplus$ such that 
  \ifthenelse{\boolean{conferenceversion}}
  {for}
  {for integers}
  $\pebbleslower$,
  $\pebblesupper$,
  and $\indexn$ 
  satisfying 
  $\pebblezero \leq
  \pebbleslower
  \leq \pebblesupper 
  \leq \indexn^{1/6}/\pebbleslower$
  there is an \XOR formula 
  $\xformf$ with $\indexn$~variables such that \firstplayer{} wins the
  $\pebbleslower$-pebble game on~$\xformf$, but does not
  win the $\pebblesupper$-pebble game 
  within
  \mbox{$\indexn^{\pebbleslower/(10\log \pebblesupper)-1/5}$ rounds}. 
\end{lemma}

Note that there is a limit to how far
$\pebbleslower$
and
$\pebblesupper$
can be from each other for the lemma to make sense---the statement
becomes vacuous if 
$\pebbleslower \leq 2 \log \pebblesupper$.
Let us see  how this lemma yields the theorems in \refsec{sec:intro}.

\begin{proof}[Proof of \refth{thm:maintheorem}]
  This theorem 
  \ifthenelse{\boolean{conferenceversion}}
  {can be seen to follow immediately from}
  {follows immediately from}
  \reftwolems
  {lem:EquivalentCharacterisations}
  {lem:MainTheoremXor},
  but let us write out the details for clarity.
  By setting 
  $\pebbleslower = \pebblesupper = \pebblesk$
  in \reflem{lem:MainTheoremXor},
  we can find  \XOR formulas with
  $\indexrhs$~variables such that \firstplayer wins the
  $\pebblesk$-pebble game on $\xformf_\indexrhs$ but needs more than
  \mbox{$\indexrhs^{\varepsilon \pebblesk/\log \pebblesk}$
    rounds}
  in order to do so 
  (provided we choose $\varepsilon<1/10$ and $\pebblezero$ large enough). 
  We can then plug these \XOR formulas into
  \reflem{lem:EquivalentCharacterisations}
  to obtain $\indexrhs$-element structures 
  $\strucA_\indexrhs = \strucA(\xformf_\indexrhs)$
  and
  $\strucB_\indexrhs = \strucB(\xformf_\indexrhs)$
  that can be distinguished in the 
  $\pebblesk$-variable fragments of first-order logic~$\FO^{\pebblesk}$
  and first-order counting logic~$\FOcnt^{\pebblesk}$, 
  but where this requires sentences of quantifier depth at 
  \mbox{least $\indexrhs^{\varepsilon \pebblesk / \log \pebblesk}$}.
\end{proof}

\begin{proof}[Proof of \refth{thm:mainWLtheorem}]
  If we let $\xformf_\indexrhs$ be the \XOR formula from
  \reflem{lem:MainTheoremXor} for 
  $\pebbleslower = \pebblesupper = \pebblesk + 1$, 
  then by
  \reflem{lem:EquivalentCharacterisations} it holds that 
  the structures $\strucA(\xformf_\indexrhs)$ and $\strucB(\xformf_\indexrhs)$
  will be distinguished 
  by the $\pebblesk$-dimensional Weisfeiler--Leman algorithm, but only
  after  $\indexrhs^{\varepsilon(\pebblesk+1) / \log
    (\pebblesk+1)}\geq \indexrhs^{\varepsilon\pebblesk / \log\pebblesk}$
  refinement steps.  
  Hence, computing the stable colouring of either of these structures
  requires at least 
  $\indexrhs^{\varepsilon \pebblesk / \log \pebblesk}$ 
  refinement steps (since they would be distinguished earlier if at
  least one of the computations terminated earlier). 
\end{proof}

\begin{proof}[Proof of \refth{thm:mainTheoremWLquasipolyTradeoff}]
  This is similar to the proof of \refth{thm:mainWLtheorem}, but setting
  $\pebbleslower = \lfloor\log \indexrhs^c\rfloor + 1$ 
  and
  \mbox{$\pebblesupper = \Ceiling{\indexrhs^{1/7}} + 1$} 
  in \reflem{lem:MainTheoremXor}. 
\end{proof}

\ifthenelse{\boolean{conferenceversion}}
{The proof of \reflem{lem:MainTheoremXor} splits into two steps.}  
{The proof of the trade-off between the number of pebbles versus number
  of rounds in \reflem{lem:MainTheoremXor} splits into two steps.}  
We first establish a rather weak lower bound on the number of rounds in
the pebble game played on suitably chosen
$\indexlhs$\nobreakdash-variable \XOR formulas 
\mbox{for $\indexlhs \gg \indexrhs$.}  
We then transform this into a much
stronger lower bound for formulas over $\indexrhs$~variables using
hardness condensation.
To help the reader keep track of 
which
results are proven in which
setting, 
in what follows we will write $\pebbleslowerlhs$
and~$\pebblesupperlhs$ to denote parameters depending on~$\indexlhs$
and $\pebbleslower$ and~$\pebblesupper$ to denote parameters
depending on~$\indexrhs$.

To implement the first step in our proof plan, we use tools developed by
Immerman~\cite{Immerman.1981} to establish a lower bound as stated in
the next lemma.

\begin{lemma}
  \label{lem:pyramids}
  For all $\pebblesupperlhs, \indexlhs \geq 3$ 
  there is an 
  $\indexlhs$-variable 3-\XOR{} formula
  $\xformf^{\pebblesupperlhs}_{\indexlhs}$ 
  on which
  \firstplayer
  \begin{enumerate}[label=(\alph*)]
  \item
    \label{item:immerman-a}
    wins the \mbox{$3$-pebble} game, but
  \item
    \label{item:immerman-b}
    \ifthenelse{\boolean{conferenceversion}}{
    does not win the $\bigl(\tfrac{1}{\ceiling{\log \pebblesupperlhs}}
    \indexlhs^{1/(1+\ceiling{\log \pebblesupperlhs})}\! -\! 2\bigr)$-round \mbox{$\pebblesupperlhs$-pebble} game.  
    }{
    does not win the $\pebblesupperlhs$-pebble game
    within
    $\max\bigl(3, \tfrac{1}{\ceiling{\log \pebblesupperlhs}}
    \indexlhs^{1/(1+\ceiling{\log \pebblesupperlhs})} - 2\bigr)$
    rounds.
    }  
  \end{enumerate}
\end{lemma}

We defer the proof of \reflem{lem:pyramids} to \refsec{sec:pyramids}, but
at this point an expert reader might wonder why we would need to prove
this lower bound at all, since a much stronger $\bigomega{\indexlhs}$ 
bound on the number of rounds in the 
pebble game
on \mbox{$4$-\XOR} formulas was 
already obtained by Fürer~\cite{Furer.2001}.  The reason is that in Fürer's
construction 
\firstplayer cannot
win the game with 
few pebbles.
However, it is crucial for the second step of our proof,
where we boost the lower bound but also significantly increase 
the number of pebbles that are needed to win the game,
that \firstplayer is able to win the original game with
very few pebbles.

The second step in the proof of our main technical lemma is carried
out by using the techniques developed by
Razborov~\cite{Razborov16NewKind} and applying them to the \XOR
formulas in \reflem{lem:pyramids}.
Roughly speaking,
if we set
$\pebbleslower = \pebblesupper = \pebblesk$ 
for simplicity, then the number of variables decreases from~$\indexlhs$ 
to \mbox{$\indexrhs \approx \indexlhs^{1/\pebblesk}$}, whereas the
\mbox{$\indexlhs^{1/\log \pebblesk}$ round} lower bound for the
$\pebblesk$-pebble game stays essentially the same and hence becomes
$\indexrhs^{\pebblesk/\log \pebblesk}$ in terms of the new number
of variables~$\indexrhs$.  
The properties of hardness condensation are
summarized in the next lemma, 
which we prove in \refsec{sec:hardness-condensation}.  
To demonstrate the flexibility of this tool we state the lemma in its
most general form---readers who want to see 
an example of
how to apply it to the
\XOR formulas in \reflem{lem:pyramids} can mentally fix
$\initialxorwidth = 3$, 
$\pebbleslowerlhs = 3$,
$\pebblesupperlhs = \pebblesupper$,
$\roundlowerbound \approx \indexlhs^{1/\log \pebblesupper}$, 
and
$\expanderdegree \approx \pebblesupper / 3$
when reading the statement of the lemma below.

\begin{lemma}[Hardness condensation lemma]
  \label{lem:hardnessCondensationXOR}
  There 
  exists
  an absolute constant $\mindegree \in \Nplus$ such that the following holds.
  Let $\xformf$ be an
  \mbox{$\indexlhs$-variable} \mbox{$\initialxorwidth$-\XOR{}} formula 
  and suppose that we can choose parameters
  \mbox{$\pebbleslowerlhs > 0$},
  \mbox{$\pebblesupperlhs\geq \mindegree\pebbleslowerlhs$}
  and
  \mbox{$\roundlowerbound$}
  such that  \firstplayer 
  \begin{enumerate}[label=(\alph*)]
  \item \label{item:condensationlhs-a}
    has a winning strategy for the
    $\pebbleslowerlhs$-pebble game on~$\xformf$, but
  \item \label{item:condensationlhs-b}
    does not win the $\pebblesupperlhs$-pebble game on~$\xformf$
    within $\roundlowerbound$ rounds.    
  \end{enumerate}
  Then for
  any %
  $\expanderdegree$ satisfying
  $\mindegree\leq \expanderdegree \leq \pebblesupperlhs/\pebbleslowerlhs$ and 
  \mbox{$(2\pebblesupperlhs
  \expanderdegree)^{2\expanderdegree} \leq \indexlhs$} 
  there is an $(\expanderdegree\initialxorwidth)$\nobreakdash-\XOR{} formula
  $\xformg$
  with
  $\Ceiling{\indexlhs^{3/\expanderdegree}}$ variables such that 
  \firstplayer
  \begin{enumerate}[label=(\alph*')]
  \item \label{item:condensationrhs-a}
    has a winning strategy for the
    $(\expanderdegree\pebbleslowerlhs)$-pebble game
    on~$\xformg$, but  
  \item \label{item:condensationrhs-b}
    does not win the $\pebblesupperlhs$-pebble game on~$\xformg$ 
    within 
    ${\roundlowerbound}/{(2\pebblesupperlhs)}$ rounds.
  \end{enumerate}
\end{lemma}

Taking \reftwolems{lem:pyramids}{lem:hardnessCondensationXOR}
on faith for now, we are ready to prove our main technical lemma
yielding an $n^{\bigomega{\pebblesk/\log \pebblesk}}$ lower bound on
the number of rounds in the $\pebblesk$-pebble game.

\begin{proof}[Proof of \reflem{lem:MainTheoremXor}]
 Let $\mindegree$ be the  constant in
 \reflem{lem:hardnessCondensationXOR}.  
 We let 
 \begin{equation}
   \label{eq:K0-first}
   \pebblezero \geq 3\mindegree +9
 \end{equation}
 be an absolute constant, the
 precise value of which will be determined by calculations later in
 the proof.
 We are given $\pebblesupper$, $\pebbleslower$, and $\indexrhs$
 satisfying the conditions 
 \begin{equation}
   \label{eq:main-lemma-inequality}
   \pebblezero 
   \leq
   \pebbleslower
   \leq
   \pebblesupper
   \leq
   \indexrhs^{1/6} / \pebbleslower
 \end{equation}
 in \reflem{lem:MainTheoremXor}.
 Let us set
 \begin{subequations}
 \begin{align}
   \label{eq:l-high-value}
   \pebblesupperlhs & \defi \pebblesupper  
   \\
   \shortintertext{and}
   \label{eq:m-value}
   \indexlhs &\defi
   \indexrhs^{\lfloor\pebbleslower/9\rfloor}  
 \end{align}
 and apply
 \reflem{lem:pyramids}
 (which is in order since
 $\pebbleslowerlhs \geq 3$ 
 and
 $\indexlhs \geq 3$ 
 by
 \refeq{eq:K0-first}  
 and~\refeq{eq:main-lemma-inequality}).
 This yields an $\indexlhs$\nobreakdash-variable \mbox{$3$-XOR}
 formula  on which 
 \firstplayer wins the \mbox{$3$-pebble} game but cannot win the
 $\pebblesupperlhs$-pebble game within 
 \begin{equation}
   \label{eq:r-value}
   \roundlowerbound
   \defi
   \tfrac{1}{\ceiling{\log \pebblesupperlhs}}
   \indexlhs^{1/(1 + \ceiling{\log \pebblesupperlhs})} - 2   
 \end{equation}
 \end{subequations}
 rounds.
 As a side remark we note that this lower bound term might vanish if $\pebbleslower$ and $\pebblesupper$
 were to far apart from each other ($\pebbleslower \leq 2 \log
 \pebblesupper$), but recall that in this case the statement of
 Lemma~\ref{lem:MainTheoremXor} becomes vacuous anyway.
 Now we can apply hardness condensation as in 
 \reflem{lem:hardnessCondensationXOR}
 to the 
 formula provided by \reflem{lem:pyramids},
 where we fix parameters
 \begin{subequations}
 \begin{align}
   \label{eq:p-value}
   \initialxorwidth &\defi 3   \eqcomma
   \\
   \label{eq:l-low-value}
   \pebbleslowerlhs &\defi 3    \eqcomma
   \\
   \shortintertext{and}
   \label{eq:delta-value}
   \expanderdegree &\defi  3\floor{\pebbleslower/9}
   \eqperiod
 \end{align}
 \end{subequations}
 To verify that our choice of parameters is legal, note that in
 addition to 
 $\roundlowerbound \geq 1$
 we also have
 $\pebbleslowerlhs > 0$
 and
 \begin{equation}
   \label{eq:hardnesscondensation-condition-1}
   \pebblesupperlhs 
   = 
   \pebblesupper 
   \geq 
   \pebblezero 
   >
   3\mindegree 
   =
   \mindegree\pebbleslowerlhs 
   \eqperiod     
 \end{equation}
 Thus, the assumptions needed for
 \ref{item:condensationlhs-a} and~\ref{item:condensationlhs-b} 
 are satisfied by the \XOR formula obtained from
 \reflem{lem:pyramids}.
 To confirm that 
 $\expanderdegree$ 
 chosen as in~\refeq{eq:delta-value}
 satisfies the conditions 
 in \reflem{lem:hardnessCondensationXOR},
 observe that
 \begin{equation}
   \label{eq:delta-bound-1}
   \mindegree
   \leq 
   3 \lfloor\pebblezero/9\rfloor   
   \leq
   3\lfloor\pebbleslower/9\rfloor  
   = 
   \expanderdegree
   \leq 
   \pebbleslower/3
   \leq 
   \pebblesupper/3 
   = 
   \pebblesupperlhs/\pebbleslowerlhs
   \eqperiod   
 \end{equation}
 Furthermore, since
 $\expanderdegree\leq \pebbleslower/3$
 and
 $\pebblesupperlhs=\pebblesupper \leq \indexrhs^{1/6}/\pebbleslower$
 we get
 \begin{equation}
   \label{eq:delta-bound-2}
   (2\pebblesupperlhs\expanderdegree)^{2\expanderdegree} 
   \leq
   \left( \frac23 \indexrhs^{1/6} \right)^{2\expanderdegree} 
   \leq
   \indexrhs^{\expanderdegree/3} 
   = \indexlhs
   \eqperiod
 \end{equation}
 Note, finally, that 
 $\indexrhs = \indexrhs^{\lfloor\pebbleslower/9\rfloor 3/\expanderdegree} 
 = \indexlhs^{3/\expanderdegree}$.
 Now \reflem{lem:hardnessCondensationXOR} provides us with an 
 $\indexrhs$\nobreakdash-variable $\pebbleslower$-XOR formula on which
 according to~\ref{item:condensationrhs-a} \firstplayer has a winning
 strategy for the \mbox{$(3\expanderdegree)$-pebble} game and hence also for
 the game with $\pebbleslower\geq 9\lfloor\pebbleslower/9\rfloor =
 3\expanderdegree$ pebbles.
 Moreover, by~\ref{item:condensationrhs-b} it holds that
 \firstplayer needs
 more than 
 \mbox{$\roundlowerbound/(2\pebblesupper)$ rounds} to win
 the \mbox{$\pebblesupper$-pebble} game.   
 To complete the proof, we observe that if we choose
 $\pebblezero$ large enough, then for
 $
 \indexrhs > \pebblesupper \geq \pebbleslower \geq \pebblezero
 $
 it holds that
 \begin{align}
   \nonumber
   \frac{\roundlowerbound}{2\pebblesupper} 
   &= 
     \frac{1}{2\pebblesupper \ceiling{\log \pebblesupper}}
     \indexrhs^{\lfloor \pebbleslower/9 \rfloor / 
     (1 + \ceiling{\log \pebblesupper})} - \frac{1}{\pebblesupper} 
   &&
      \bigl[\text{by \refeq{eq:m-value} and  \refeq{eq:r-value}}\bigr]
   \\
   &\geq 
     \frac{6 \indexrhs^{1/5}}{\indexrhs^{1/6} \log \indexrhs}
     \indexrhs^{\lfloor\pebbleslower/9\rfloor/(1+\ceiling{\log
     \pebblesupper})-1/5} - \frac{1}{\pebblesupper}
   &&
      \bigl[\text{since
      $\pebblesupper \leq \indexrhs^{1/6}$}\bigr]
   \\
   \nonumber
   &\geq 
     \indexn^{\pebbleslower/(10\log \pebblesupper)-1/5}
     &&
     \bigl[\text{for large enough
     $\indexrhs$, $\pebblesupper$, and $\pebbleslower$.}\bigr]
 \end{align} 
 We now choose the constant $\pebblezero$ large so that all conditions
 encountered in the calculations above are valid. This establishes the
 lemma.
\end{proof}

\makeatletter{}%

\ifthenelse{\boolean{conferenceversion}}
{\section{\XOR Formulas over Pyramids}}
{\section{\XOR Formulas over High-Dimensional Pyramids}}
\label{sec:pyramids}

We now proceed to establish the $\pebblestd$-pebble game lower bound
stated in \reflem{lem:pyramids}.
Our \XOR formulas will be constructed over directed acyclic graphs
(DAGs) as described in the following definition.

\begin{definition}
  \label{def:xor-formula}
  Let $\digraph$ be 
  a DAG
  with
  sources~$S$ and a unique sink~$z$.
  The \XOR formula $\digraphxor$ contains one
  variable~$v$ for
  every vertex $v\in V(\digraph)$ and consists of the following clauses: 
  \begin{enumerate}[label=(\alph*)]
  \item
    \label{xor-clause-source}
    $(s,0)$ for every source $s\in S$,
  \item 
    \label{xor-clause-non-source}
    $(v,w_1,\ldots,w_\ell,0)$ for all 
    non-sources
    $v\in V(\digraph)\setminus S$ with in-neighbours
    $N^-(v)=\{w_1,\ldots,w_\ell\}$, 
  \item
    \label{xor-clause-sink}
    $(z,1)$ for the unique sink $z$.
  \end{enumerate}
\end{definition}

Note that the formula
$\digraphxor$
is always unsatisfiable, since all 
\ifthenelse{\boolean{conferenceversion}}
{sources}
{source vertices} 
are forced to~$0$ by~\ref{xor-clause-source},
which forces all other vertices to~$0$ in topological order
by~\ref{xor-clause-non-source},
contradicting~\ref{xor-clause-sink}
for the sink.
Incidentally, these formulas are somewhat similar to the 
\emph{pebbling formulas} defined in~\cite{BW01ShortProofs}, which have
been very useful in proof complexity (see the
survey~\cite{Nordstrom13SurveyLMCS} for more details). 
The difference is that pebbling formulas state that a vertex $v$ 
is true if and only if all of its in-neighbours are true,  
whereas $\digraphxor$ states that $v$ is true if and only if the
parity of the number of true in-neighbours is odd.

It is clear that one winning strategy for \firstplayer is to ask first
about the sink~$z$, for which \secondplayer has to answer~$1$ (or lose
immediately) and then about all the in-neighbours of the sink until the
answer for one vertex~$v$ is~$1$ (if there is no such vertex,
\secondplayer again loses immediately). At this point \firstplayer can
forget all other vertices and then ask about the
in-neighbours of~$v$ until a \mbox{$1$-labelled} vertex~$w$ is found,
and then continue in this way to trace a path of \mbox{$1$-labelled} vertices
backwards through the DAG until some source~$s$ is reached, 
which contradicts the requirement that $s$ should be
labelled~$0$. Formalizing this as an induction proof on the depth
of~$\digraph$ shows that if the in-degree is bounded, then
\firstplayer can win the pebble game on $\digraphxor$ with few pebbles
\ifthenelse{\boolean{conferenceversion}}
{as stated next.}
{as stated in the next lemma.}

\begin{lemma}
  \label{lem:UpperBoundPebblingDAG}
  Let $\digraph$ be a DAG with a unique sink and maximal
  in-degree~$\indegreenumber$.   
  Then
  \firstplayer wins the $(\indegreenumber+1)$-pebble game
  on~$\digraphxor$.  
\end{lemma}

As a warm-up for the proof of \reflem{lem:pyramids}, 
let us describe a very weak lower bound from \cite{Immerman.1981} 
for the complete binary tree of height~$h$ (with
edges directed from the leaves to the root), which we will
denote~$\binarytree_h$.
By the lemma above, \firstplayer wins the \mbox{$3$-pebble} game on
$\digraphtoxor(\binarytree_h)$ in $O(h)$~steps by 
propagating~$1$ from the root down to some leaf.  On the other hand,
\secondplayer has the freedom to decide on which path she answers~$1$.
Hence, she can safely 
respond~$0$ 
for a vertex~$v$   %
as long as there is some leaf with a pebble-free path leading to the
lowest pebble labelled~$1$
without passing~$v$.    %
In particular, 
if \secondplayer is asked about vertices at least $\ell$~layers below
the lowest pebbled vertex for which the answer~$1$ was given,
then she can answer~$0$ for $2^\ell-1$~queries.
It follows that the height~$h$ provides a
lower bound on the number of rounds \firstplayer needs to win the
game, even if he has an infinite amount of pebbles.
We remark that this proof in terms of pebble-free paths is somewhat reminiscent
of an argument by Cook~\cite{Cook74ObservationTimeStorageTradeOff} for
the so-called black pebble game corresponding to the pebbling formulas 
in~\cite{BW01ShortProofs} briefly discussed above. 

The downside of this lower bound is that the height is only
logarithmic in the number of vertices and thus too weak for us as we are
shooting for a lower bound of the order of~$n^{1/\log k}$.  
To get a better bound for the black pebble game Cook
instead considered so-called pyramid graphs as
in~\reffig{fig:2figsA}. These will not be sufficient to obtain strong
enough lower bounds for our pebble game, 
\ifthenelse{\boolean{conferenceversion}}
{however.}
{however.%
  \footnote{For readers knowledgeable in pebbling, we comment
    that the problem is that the open-path argument
    in~\cite{Cook74ObservationTimeStorageTradeOff}  
    does not work in a DAG-like setting for the XOR pebble game.
    To see this, consider a pyramid with a vertex row
    $u,v,w$
    and a second row 
    $p,q,r,s$
    immediately below   such that the edges are
    $(p,u),(q,u),(q,v),(r,v),(r,w),(s,w)$.
    Then if the values of  $u,w$ on the upper row
    and
    $p,s$ on the lower row are known, 
    there is still an open path via $(q,v)$ or $(r,v)$, which is enough
    for the black pebbling lower bound for pyramids in
    \cite{Cook74ObservationTimeStorageTradeOff}.
    But in the \XOR pebble game this means that $r$ and~$q$ are already
    fixed because of the \XOR constraints, and so there is no ``open
    path'' with unconstrained vertices.}}
Instead, following Immerman we consider a kind of high-dimensional
generalization of these graphs, for which the lower bound on the
number of rounds in the $k$-pebble game is still linear in the
height~$h$ while the number of vertices is roughly~$h^{\log k}$.

\begin{definition}[\cite{Immerman.1981}]
  \label{def:high-dim-pyramid}
  For $d\geq 1$ we define the
  \introduceterm{$(d+1)$-dimensional pyramid of height $h$}, denoted
  by~$\pyramid^d_h$, to be the following layered DAG.
  We let 
  $\rownumber$, $0\leq \rownumber\leq h$ 
  be the \emph{layer number} and set
  $\pyrquot{d}{\rownumber} 
  \defi \lfloor \rownumber/d\rfloor$ and
  $\pyrrem{d}{\rownumber} 
  \defi \rownumber{}\pmod d$. Hence,
  for any $\rownumber$ we have
  $\rownumber=
  \pyrquot{d}{\rownumber} \cdot d + 
  \pyrrem{d}{\rownumber}$. 
  For integers $x_i \geq 0 $ 
  the vertex set is
  \begin{subequations}
  \begin{equation}
     \ifthenelse{\boolean{conferenceversion}}
     {%
     \begin{split} 
      &V \bigl( \pyramid^d_h \bigr) =
      \big\{ (x_0,\ldots, x_{d-1},\rownumber) \mid 
        \rownumber \leq h ; \\
        &\quad x_i \leq \pyrquot{d}{\rownumber} + 1
      \text{ if $i \!<\! \pyrrem{d}{\rownumber}$}  ; %
      x_i \leq \pyrquot{d}{\rownumber}
      \text{ if $i \!\geq\! \pyrrem{d}{\rownumber}$}\big\} 
      \eqcomma
     \end{split} 
     }{
      V \bigl( \pyramid^d_h \bigr) =
      \Setdescr
      {(x_0,\ldots, x_{d-1},\rownumber) \, }
      {\, 
        \rownumber \leq h ; \,
        x_i \leq \pyrquot{d}{\rownumber} + 1
      \text{ if $i \!<\! \pyrrem{d}{\rownumber}$}  ; \,
      x_i \leq \pyrquot{d}{\rownumber}
      \text{ if $i \!\geq\! \pyrrem{d}{\rownumber}$} } 
      \eqcomma
     } 
  \end{equation}
  where we say that $\rownumber$ is the \introduceterm{layer} of the
  vertex~$(x_0,\ldots, x_{d-1},\rownumber)$.
  The edge set $E \bigl( \pyramid^d_h \bigr)$ consists of
  the pair of edges
\begin{equation}
  \begin{split} 
    \bigl(
    (x_0, \ldots, 
    x_{\pyrrem{d}{\rownumber}}, 
    \ldots, x_{d-1},
    \rownumber+1) ,
    (x_0,\ldots,
    x_{\pyrrem{d}{\rownumber}}, 
    \ldots, x_{d-1},
    \rownumber)
    \bigr)
    \eqcomma
    \\ 
    \bigl(
    (x_0,\ldots,
    x_{\pyrrem{d}{\rownumber} }+1, 
    \ldots,
    x_{d-1},
    \rownumber+1) ,
    (x_0,\ldots, 
    x_{\pyrrem{d}{\rownumber}}, 
    \ldots,
    x_{d-1},\rownumber) 
    \bigr)
  \end{split}
\end{equation}
\end{subequations}
for all vertices 
$(x_0,\ldots, x_{d-1},\rownumber{})\in V(\pyramid^d_h)$ and
layers~%
$\rownumber<h$,
so that every vertex in layer~$\rownumber$ has exactly two in-neighbours
from \mbox{layer $\rownumber+1$}. 
\end{definition}

\ifthenelse{\boolean{conferenceversion}}
{We refer the reader to 
  \reftwofigs{fig:2figsA}{fig:2figsB}
  for illustrations of
  $2$\nobreakdash-dimensional and
  $3$\nobreakdash-dimensional pyramids
  (where all the edges in the figures are assumed to be directed upwards).  
  The vertex~$(0,\ldots,0)$ at the top of the pyramid is the unique sink and
  all vertices at the bottom layer~$h$ 
  are sources.}
{It might be easier to parse 
  \refdef{def:high-dim-pyramid}
  by noting that the $(k d)$th layer of~$\pyramid^d_h$
  is a $d$\nobreakdash-dimensional cube of side length~$k$. 
  Intuitively, we then want to have incoming edges to each vertex~$u$ at the
  $(k  d)$th layer from all vertices~$v$ in the
  $d$\nobreakdash-dimensional cube of side length~$k+1$
  such that all coordinates
  of~$v$ are at distance $0$ or~$+1$ from the coordinates of~$u$. 
  This would give a fan-in larger than~$2$, however, and to avoid this
  we expand in one dimension at a time to obtain a sequence of
  multidimensional cuboids where in each consecutive cuboid the side
  length increases by one in one dimension,
  until $d$~layers later we have a complete cube with side length~$k+1$.
  We refer the reader to 
  \reffig{fig:2figsA}
  for an illustration of a
  $2$\nobreakdash-dimensional pyramid generated by stacking
  $1$\nobreakdash-dimensional cubes on top of one another and to
  \reffig{fig:2figsB}
  for a
  $3$\nobreakdash-dimensional pyramid generated from
  $2$\nobreakdash-dimensional cuboids
  (where all the edges in the figures are assumed to be directed upwards).  
  The vertex~$(0,\ldots,0)$ at the top of the pyramid is the unique sink and
  all vertices at the bottom layer~$h$ 
  are sources. Observe that it follows from the definition that
  $\Setsize{V \bigl( \pyramid^\dimension_h \bigr) }\leq (h+1)^{\dimension+1}$.}

\ifthenelse{\boolean{conferenceversion}}
{%
\makeatletter{}%

\begin{figure}%
  \centering
   \parbox{1.2in}{ %
    \resizebox{!}{4cm}{
    \begin{tikzpicture}
  [vertex/.style={circle,draw=black,fill=black,inner  sep=0pt,minimum  size=3pt},
   diredge/.style={-}] %

  \foreach \height/\xdist/\ydist in {10/.5/.5} {

  \foreach \y in {0,...,\height} {
      \pgfmathsetmacro{\diff}{\height-\y};
      \foreach \x in {0,...,\diff} {
            \node[vertex] (v_\x_\y) at (\x*\xdist-\diff*\xdist/2,\y*\ydist) {}; 
      }
  }
  \foreach \y in {1,...,\height} {
      \pgfmathsetmacro{\diff}{\height-\y};
      \foreach \x in {0,...,\diff} {
          \pgfmathtruncatemacro{\nextx}{1+\x};
          \pgfmathtruncatemacro{\prevy}{\y-1};
          \draw[diredge] (v_\x_\prevy) -- (v_\x_\y);
          \draw[diredge] (v_\nextx_\prevy) -- (v_\x_\y); 
      }
  }

  }%
  \end{tikzpicture}
  }%
         \caption{2D pyramid}%
         \label{fig:2figsA}}%
\hspace{1.4cm}
\begin{minipage}{1.2in}%
\resizebox{!}{4cm}{
  \begin{tikzpicture}
  [vertex/.style={circle,draw=black,fill=black,inner  sep=0pt,minimum  size=4pt},
   diredge/.style={-,shorten <=.5pt, shorten >=.5pt}]

  \foreach \height/\xdist/\ydist/\zdist in {8/1.8/1.2/.7} {

  \foreach \y in {0,...,\height} {
      \pgfmathtruncatemacro{\xmax}{(1+\height-\y)/2};
      \pgfmathtruncatemacro{\zmax}{(\height-\y)/2};
      \foreach \x in {0,...,\xmax} {
          \foreach \z in {0,...,\zmax} {
                  \node[vertex] (v_\x_\y_\z) at (\x*\xdist-\xmax*\xdist/2,\y*\ydist,\z*\zdist-\zmax*\zdist/2) {}; 
          }
      }
  }

  \foreach \y in {1,...,\height} {
      \pgfmathtruncatemacro{\xmax}{(1+\height-\y)/2};
      \pgfmathtruncatemacro{\zmax}{(\height-\y)/2};
      \foreach \x in {0,...,\xmax} {
          \foreach \z in {0,...,\zmax} {
              \pgfmathtruncatemacro{\Expanded}{mod(\height-\y,2)};
              \pgfmathtruncatemacro{\prevy}{\y-1};
              \draw[diredge] (v_\x_\prevy_\z) -- (v_\x_\y_\z); %
              \pgfmathtruncatemacro{\nextx}{\x+1-\Expanded};
              \pgfmathtruncatemacro{\nextz}{\z+\Expanded};
              \draw[diredge] (v_\nextx_\prevy_\nextz) -- (v_\x_\y_\z); %
          }
      }
  }

  }%
  \end{tikzpicture}
}%
         \caption{3D pyramid}%
         \label{fig:2figsB}%
       \end{minipage}%
     \end{figure}%

}
{%
\makeatletter{}%

\newcommand{\verticalspacing}{\vspace{3mm}}

\begin{figure}[tp]%
  \centering
  \subfigure[2D pyramid]%
  {
    \centering
    \label{fig:2figsA}%
    \begin{minipage}{7.5cm}
      \resizebox{!}{7.2cm}{
        \begin{tikzpicture}%
          [vertex/.style={circle,draw=black,fill=black,%
            inner sep=0pt,minimum  size=3pt},%
          diredge/.style={-}] %

          \foreach \height/\xdist/\ydist in {10/.5/.5} {

            \foreach \y in {0,...,\height} {
              \pgfmathsetmacro{\diff}{\height-\y};
              \foreach \x in {0,...,\diff} {
                \node[vertex] (v_\x_\y) at (\x*\xdist-\diff*\xdist/2,\y*\ydist) {}; 
              }
            }
            \foreach \y in {1,...,\height} {
              \pgfmathsetmacro{\diff}{\height-\y};
              \foreach \x in {0,...,\diff} {
                \pgfmathtruncatemacro{\nextx}{1+\x};
                \pgfmathtruncatemacro{\prevy}{\y-1};
                \draw[diredge] (v_\x_\prevy) -- (v_\x_\y);
                \draw[diredge] (v_\nextx_\prevy) -- (v_\x_\y); 
              }
            }
            
          }%
        \end{tikzpicture}
      }%
    \verticalspacing
    \end{minipage}
  }
  \hfill
  \subfigure[3D pyramid]
  {
    \centering
    \label{fig:2figsB}%
    \begin{minipage}{7.5cm}
      \hspace{.8cm}
      \resizebox{!}{7.2cm}{
        \begin{tikzpicture}
          [vertex/.style={circle,draw=black,fill=black,%
            inner sep=0pt,minimum  size=4pt},%
          diredge/.style={-,shorten <=.5pt, shorten >=.5pt}]

          \foreach \height/\xdist/\ydist/\zdist in {8/1.8/1.2/.7} {

            \foreach \y in {0,...,\height} {
              \pgfmathtruncatemacro{\xmax}{(1+\height-\y)/2};
              \pgfmathtruncatemacro{\zmax}{(\height-\y)/2};
              \foreach \x in {0,...,\xmax} {
                \foreach \z in {0,...,\zmax} {
                  \node[vertex] (v_\x_\y_\z) at (\x*\xdist-\xmax*\xdist/2,\y*\ydist,\z*\zdist-\zmax*\zdist/2) {}; 
                }
              }
            }

            \foreach \y in {1,...,\height} {
              \pgfmathtruncatemacro{\xmax}{(1+\height-\y)/2};
              \pgfmathtruncatemacro{\zmax}{(\height-\y)/2};
              \foreach \x in {0,...,\xmax} {
                \foreach \z in {0,...,\zmax} {
                  \pgfmathtruncatemacro{\Expanded}{mod(\height-\y,2)};
                  \pgfmathtruncatemacro{\prevy}{\y-1};
                  \draw[diredge] (v_\x_\prevy_\z) -- (v_\x_\y_\z); %
                  \pgfmathtruncatemacro{\nextx}{\x+1-\Expanded};
                  \pgfmathtruncatemacro{\nextz}{\z+\Expanded};
                  \draw[diredge] (v_\nextx_\prevy_\nextz) -- (v_\x_\y_\z); %
                }
              }
            }
            
          }%
        \end{tikzpicture}
      }%
     \verticalspacing
    \end{minipage}%
  }
  \caption{Examples of high-dimensional pyramids.}
  \label{fig:high-dim-pyramids}
\end{figure}
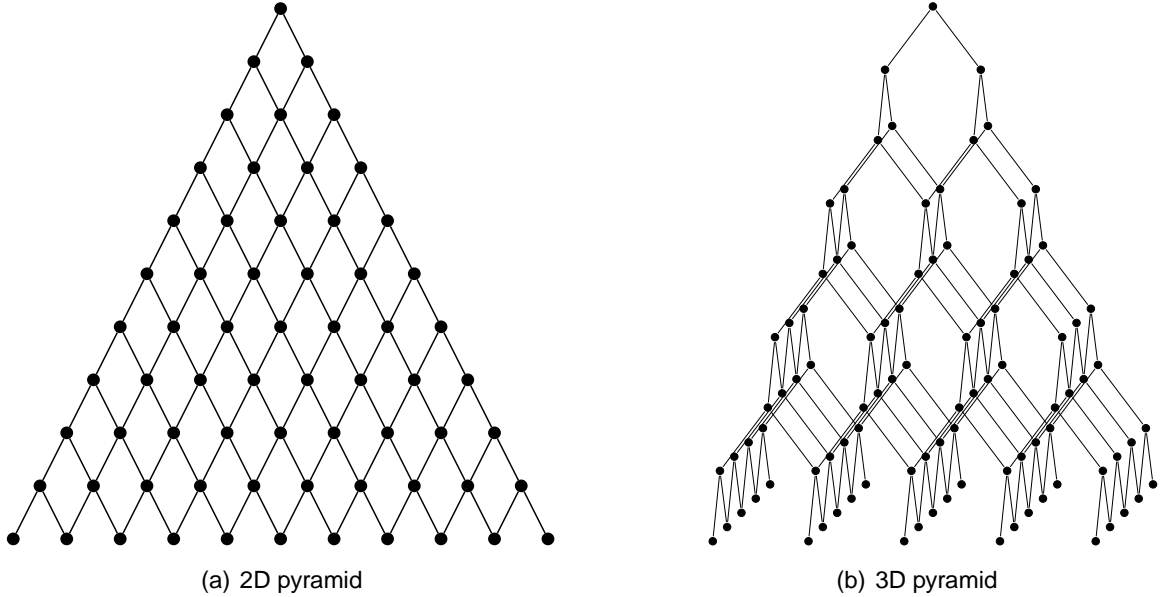%

}

As high-dimensional pyramids have in-degree~$2$,
\reflem{lem:UpperBoundPebblingDAG} implies that \firstplayer wins the
\mbox{$3$-pebble} game on~$\pyramid^d_h$.
Recall that, as discussed in the proof sketch of the lemma,
\firstplayer starts his winning strategy in the \mbox{$3$-pebble} game
by pebbling the sink of the pyramid and its two in-neighbours.
One of them has to be labelled~$1$.  Then he picks up the two other
pebbles and pebbles the two in-neighbours of the vertex marked
with~$1$ and so on. Continuing this strategy, he is able to ``move''
the~$1$ all the way to the bottom, reaching a contradiction, in a
number of rounds that is linear in the height of the pyramid.  This
strategy turns out to be nearly optimal in the sense that in order to
move a $1$ from the top to the bottom in $\pyramid^d_h$, 
as long as the total number of available pebbles is at
most~$2^\dimension$
it makes no sense for \firstplayer at any point in the game to pebble
a vertex that is $\dimension$ or more levels away from the lowest
level containing a pebble.

The next lemma states a key property of pyramids in this regard.  
In order to state it, we need to make a definition.

\begin{definition}
  \label{def:consistent-labelling}
  We refer to a partial assignment $\labelling$ of Boolean values to
  the vertices of a DAG~$\digraph$ as a \introduceterm{labelling} or
  \introduceterm{marking} of~$\digraph$.
  We say that $\labelling$ is
  \introduceterm{consistent} if no clause of
  type~\ref{xor-clause-non-source}
  or~\ref{xor-clause-sink}
  in the XOR formula~$\digraphxorarg{\digraph}$ in 
  \refdef{def:xor-formula}
  is falsified by~$\labelling$.
  \ifthenelse{\boolean{conferenceversion}}
  {}
  {We also say that $\labelling'$ is 
  \introduceterm{consistent with~$\labelling$}     
  if
  $\labelling \cup \labelling'$
  is a consistent labelling of~$\digraph$.}
\end{definition}

That is, a consistent labelling does not violate any constraint on any
non-source vertex, but 
\ifthenelse{\boolean{conferenceversion}}
{source} 
{source vertex} 
constraints~\ref{xor-clause-source} 
may be falsified.
Such labellings are easy to find for high-dimensional pyramids.

\begin{lemma}%
[\cite{Immerman.1981}]
  \label{lem:ImmermansMainTechnicalLemma} 
  Let $\labelling$ be any consistent labelling of all vertices in a
  pyramid~$\pyramid^\dimension_h$ from layer~$0$ to
  layer~$\rownumber$.
  Then for every set $\setS$ of $2^\dimension{}-1$ vertices on or below
  layer $\rownumber+\dimension$ there is a consistent labelling of 
  the
  entire pyramid that extends $\labelling$ and labels all vertices in
  $\setS$ with~$0$. 
\end{lemma}

To get some intuition why \reflem{lem:ImmermansMainTechnicalLemma}
holds, note that
the $\dimension$-dimensional pyramids are constructed in such a way
that they locally look like binary trees.  
In particular, every vertex $v\in V(\pyramid^\dimension_h)$ together
with all its predecessors at distance at most~$\dimension$ form a
complete binary tree.  
By the same argument as for the 
binary trees above, it follows that
if $v$ is labelled with~$1$, \secondplayer can safely answer~$0$
up to $2^\dimension-1$ times when asked about vertices
$\dimension$~layers below~$v$.    
\ifthenelse{\boolean{conferenceversion}}
{However, the full proof 
  of \reflem{lem:ImmermansMainTechnicalLemma}
  is more challenging and requires some quite
  subtle reasoning. We refer the reader to~\cite{Immerman.1981} 
  (or the upcoming full-length version of this paper) for the details.}  
{However, the full proof
  of \reflem{lem:ImmermansMainTechnicalLemma}
  is more challenging and requires some quite
  subtle reasoning. For the convenience of the reader we now present a
  slightly modified version of the proof in~\cite{Immerman.1981}
  with notation and terminology adapted to this paper.

\makeatletter{}%

\newcommand{\veczerodminusone}[1]{\vec{#1}}

\newcommand{\dimindex}{j}
\newcommand{\vala}{a}
\newcommand{\dimValLayerWedge}[3]{\ensuremath{\mathcal W({#1},{#2},{#3})}}
\newcommand{\dimW}{\dimindex}
\newcommand{\valW}{\vala}
\newcommand{\layerW}{\layernumber}

\newcommand{\dimValLayerLeftSlice}[3]%
    {\ensuremath{\mathcal{S}_{\mathsf L}({#1},{#2},{#3})}}
\newcommand{\dimValLayerRightSlice}[3]%
    {\ensuremath{\mathcal{S}_{\mathsf R}({#1},{#2},{#3})}}
\newcommand{\dimS}{j}
\newcommand{\valS}{a}
\newcommand{\layerS}{\layernumber}
\newcommand{\chosenset}{T}
\newcommand{\coordinaterestriction}{\alpha}
\newcommand{\coordinaterestrictionsize}{q}
\newcommand{\lslice}{left slice\xspace}
\newcommand{\lslices}{left slices\xspace}
\newcommand{\rslice}{right slice\xspace}
\newcommand{\rslices}{right slices\xspace}
\newcommand{\initiallabelling}{\widetilde{\labelling}}

\newcommand{\pyramidalxor}{pyramidal~\XOR formula\xspace}
\newcommand{\frustum}{frustum\xspace}
\newcommand{\frustums}{frustums\xspace}
\newcommand{\topex}{top-expanding\xspace}
\newcommand{\indsetlo}{I_{\text{\upshape lo}}}
\newcommand{\indsethi}{I_{\text{\upshape hi}}}
\newcommand{\maxval}{\mathit{slen}}
\newcommand{\maxincr}{\Delta_\mathit{slen}}
\newcommand{\bluecross}{\tikz{\node[cross out,draw=blue,ultra thick,%
    inner sep=0pt,minimum  size=5 pt] at (0,0){ };}}
\DeclareRobustCommand{\bluecrossrobust}{\tikz{\node[cross out,draw=blue,ultra thick,%
    inner sep=0pt,minimum  size=5 pt] at (0,0){ };}}
\DeclareRobustCommand{\whitecirclerobust}{
           \mbox{\tikz{\node[circle,draw=black,fill=white,inner
           sep=0pt,minimum  size=5pt] at (0,0) {};}}}

To formalize the intuitive argument above, we need  some additional
notation and technical definitions.
Let us use the shorthand 
$
\veczerodminusone{x}
= (x_0, \ldots, x_{\dimension-1})
$.
For a pyramid~$\pyramid^\dimension_\pyheight$, a coordinate
\mbox{$\dimindex\in\set{0,\ldots,\dimension-1}$},
and a layer~$\layernumber$, 
we let 
\begin{equation}
  \label{eq:maxval-1}
  \maxval(\dimindex,\layernumber) 
  \defi
  \max
  \Setdescr{x_\dimindex}{(\veczerodminusone{x}, \layernumber)
    \in
    V\bigl( \pyramid^\dimension_\pyheight \bigr)}
\end{equation}
be the side length in the $\dimindex$th dimension of the cuboid in
layer~$\layernumber$,  \ie
the maximal value that can be achieved in the 
$\dimindex$th~coordinate in layer~$\layernumber$,
and for $\layernumberalt \geq \layernumber$ we write
\begin{equation}
  \label{eq:maxval-2}
  \maxincr(\dimindex,\layernumber,\layernumberalt) 
  \defi
  \maxval(\dimindex,\layernumberalt) - 
  \maxval(\dimindex,\layernumber) 
\end{equation}
to denote how much the cuboids in 
$\pyramid^\dimension_\pyheight$
grow in the \mbox{$\dimindex$th dimension}
in betwen layers~$\rownumber$ and~$\rownumberalt$.

We define
the \introduceterm{\frustum}~$\pyramid^\dimension_{\layernumber,\pyheight}$ 
to be the subgraph of~$\pyramid^\dimension_\pyheight$ 
induced on the set
\mbox{$
  \Setdescr
  {\bigl( \veczerodminusone{x},\layernumberalt \bigr)}
  {\layernumberalt \geq \layernumber}
  $}
of all vertices on layer~$\layernumber$ and below.
We say that 
the \introduceterm{wedge}
\introduceterm{\dimValLayerWedge{\dimW}{\valW}{\layerW}} is the
subgraph of~$\pyramid^\dimension_\pyheight$ induced on the vertices
$(x_0, \dots, x_{\dimW-1}, \valW, x_{\dimW+1}, \ldots,
x_{\dimension-1}, \layerW)$ 
with fixed $\dimW$th coordinate $x_{\dimW} = \valW$ 
together with all predecessors of these vertices.
That is, the vertex set of
$\dimValLayerWedge{\dimW}{\valW}{\layerW}$
is
\begin{equation}
  \label{eq:wedge-vertex-set}
  V \bigl( \dimValLayerWedge{\dimW}{\valW}{\layerW} \bigr) =  
  \Setdescr{ 
    \bigl(
    \veczerodminusone{x},
    \layernumberalt \bigr)
    \in V\bigl( \pyramid^\dimension_\pyheight \bigr)
  }
  {
    \layernumberalt\geq \layerW, \valW
    \leq x_{\dimW} \leq \valW +
    \maxincr(\dimW,\layerW,\layernumberalt)
  }  
  \eqperiod
\end{equation}
An important part in our proof will be played by subgraphs obtained by 
deleting wedges from \frustums.  We define these subgraphs next.

Fix a \frustum $\pyramid^\dimension_{\layernumber,\pyheight}$, two
disjoint subsets of coordinates $\indsetlo,\indsethi \subseteq
\{0,\ldots,\dimension-1\}$, and a mapping 
\mbox{$\coordinaterestriction \colon \indsetlo\cup\indsethi\to \Nzero$}
such that
$\coordinaterestriction(\dimindex)\leq
\maxval(\dimindex,\layernumber)$. 
We let the \introduceterm{restricted \frustum}
$\pyramid^\dimension_{\layernumber,\pyheight}%
[\indsetlo,\indsethi,\coordinaterestriction]$ 
be the subgraph of the \frustum~$\pyramid^\dimension_{\layernumber,\pyheight}$
induced on the vertex set
\begin{equation}
  \label{eq:restricted-frustum-vertex-set}
  \Setdescr{
    \bigl( \veczerodminusone{x},\layernumberalt \bigr)
    \in
    V(\pyramid^\dimension_{\layernumber,\pyheight})
  }
  {\,
    \forall
    \dimindex\!\in\!\indsetlo: x_\dimindex>
    \coordinaterestriction(\dimindex) +
    \maxincr(\dimindex,\layernumber,\layernumberalt)
    ;
    \,
    \forall \dimindex\!\in\!\indsethi: x_\dimindex <
    \coordinaterestriction(\dimindex)}
\end{equation}
where no coordinates in
$\indsetlo\cup\indsethi$ are 
expanded, \ie
the cuboids will not grow in size in dimensions
$\indsetlo \union \indsethi$ as we move down the layers. 
For dimensions in $\indsethi$ the coordinate set stays the same, and
for dimensions in $\indsetlo$ the coordinate set shifts by an
additive~$+1$ every time the pyramid graph grows in this direction. 
We say that a layered directed graph is a
\introduceterm{$(\dimension,\coordinaterestrictionsize)$-\frustum{}} 
if it is  a
restricted \frustum $\pyramid^\dimension_{\layernumber,\pyheight}%
[\indsetlo,\indsethi,\coordinaterestriction]$ where
$\coordinaterestrictionsize$ coordinates are restricted, \ie
$\setsize{\indsetlo\disjointunion\indsethi} = \coordinaterestrictionsize$.

To see how restricted \frustums are obtained by deleting wedges from
\frustums, note that after removing the wedge
$\dimValLayerWedge{\dimindex}{\valW}{\layernumber}$ 
from the \frustum
$\pyramid^\dimension_{\layernumber,\pyheight}$, 
the remaining graph is the disjoint union of
the restricted \frustums
$\pyramid^\dimension_{\layernumber,\pyheight}%
[\set{\dimindex},\emptyset,\set{\dimindex\mapsto\valW}]$ 
and
$\pyramid^\dimension_{\layernumber,\pyheight}%
[\emptyset,\set{\dimindex},\set{\dimindex\mapsto\valW}]$.  
\reffig{fig:pyramidWedgeFrustum} 
shows a 3D pyramid with a 
\mbox{wedge \tikz{\node[circle,draw=black,fill=white,inner
    sep=0pt,minimum  size=5pt] at (0,0) {};}} and a restricted
frustum~\bluecross.  

\makeatletter{}%

\begin{figure}[t]%
  \hspace{-1.5cm}
  \resizebox{14cm}{!}{\centering%
\makeatletter{}%

\newcommand{\drawthreedimpyramid}[5]{ 
\tdplotsetmaincoords{180}{180} 
  \tdplotsetrotatedcoords{0}{20}{20}
\begin{tikzpicture}[rotate=20,tdplot_main_coords,
  layerlabel/.style={tdplot_rotated_coords},
  vertex/.style={tdplot_rotated_coords,circle,draw=black,fill=black,inner  sep=0pt,minimum  size=#5 pt},
  wedgevertex/.style={tdplot_rotated_coords,circle,draw=black,fill=white,inner  sep=0pt,minimum  size=#5 pt},
  slicevertex/.style={tdplot_rotated_coords,cross out,draw=blue,ultra thick,inner  sep=0pt,minimum  size=#5 pt},
  frustumvertex/.style={tdplot_rotated_coords,cross out,draw=blue,ultra thick,inner  sep=0pt,minimum  size=#5 pt},
   diredge/.style={-,shorten <=.5pt, shorten >=.5pt,thin}]
  \foreach \height/\xdist/\ydist/\zdist in {#1/#2/#3/#4} {
    \node[layerlabel] at (-6.5,0.2+11*\ydist,0) {\Large layer};
    \node[layerlabel] at (-6.5,0.2+10*\ydist,0) {\Large $0$};
    \node[layerlabel] at (-6.5,0.2+9*\ydist,0) {\Large $1$};
    \node[layerlabel] at (-6.5,0.2+7.5*\ydist,0) {\Large $\cdots$};
    \node[layerlabel] at (-6.5,0.2+6*\ydist,0) {\Large $\layernumber$};
    \node[layerlabel] at (-6.5,0.2+5*\ydist,0) {\Large $\layernumber+1$};
    \node[layerlabel] at (-6.5,0.2+4*\ydist,0) {\Large $\layernumber+\dimension$};
    \node[layerlabel] at (-6.5,0.2+3*\ydist,0) {\Large $\layernumber+\dimension+1$};
    \node[layerlabel] at (-6.5,0.2+1.5*\ydist,0) {\Large $\cdots$};
    \node[layerlabel] at (-6.5,0.2+0,0) {\Large $\pyheight$};

  \newboolean{wedge}
  \newboolean{slice}
  \newboolean{frustum}
  \foreach \y in {0,...,\height} {
      \pgfmathtruncatemacro{\xmax}{(1+\height-\y)/2};
      \pgfmathtruncatemacro{\zmax}{(\height-\y)/2};
      \foreach \x in {0,...,\xmax} {
          \foreach \z in {0,...,\zmax} {
              \setboolean{wedge}{false}
              \ifthenelse{\y=5 \and \x=1}{\setboolean{wedge}{true}}{}
              \ifthenelse{\y=4 \and \x=1}{\setboolean{wedge}{true}}{}
              \ifthenelse{\y=3 \and \x>0 \and \x<3}{\setboolean{wedge}{true}}{}
              \ifthenelse{\y=2 \and \x>0 \and \x<3}{\setboolean{wedge}{true}}{}
              \ifthenelse{\y=1 \and \x>0 \and \x<4}{\setboolean{wedge}{true}}{}
              \ifthenelse{\y=0 \and \x>0 \and \x<4}{\setboolean{wedge}{true}}{}
              \setboolean{slice}{false}
              \ifthenelse{\y=5 \and \x=2}{\setboolean{slice}{true}}{}
              \ifthenelse{\y=4 \and \x=2}{\setboolean{slice}{true}}{}
              \ifthenelse{\y=3 \and \x=3}{\setboolean{slice}{true}}{}
              \ifthenelse{\y=2 \and \x=3}{\setboolean{slice}{true}}{}
              \ifthenelse{\y=1 \and \x=4}{\setboolean{slice}{true}}{}
              \ifthenelse{\y=0 \and \x=4}{\setboolean{slice}{true}}{}
              \setboolean{frustum}{false}
              \ifthenelse{\y=5 \and \x=2}{\setboolean{frustum}{true}}{}
              \ifthenelse{\y=4 \and \x=2}{\setboolean{frustum}{true}}{}
              \ifthenelse{\y=3 \and \x=3}{\setboolean{frustum}{true}}{}
              \ifthenelse{\y=2 \and \x=3}{\setboolean{frustum}{true}}{}
              \ifthenelse{\y=1 \and \x=4}{\setboolean{frustum}{true}}{}
              \ifthenelse{\y=0 \and \x=4}{\setboolean{frustum}{true}}{}
              \ifthenelse{\y=5 \and \x=3}{\setboolean{frustum}{true}}{}
              \ifthenelse{\y=4 \and \x=3}{\setboolean{frustum}{true}}{}
              \ifthenelse{\y=3 \and \x=4}{\setboolean{frustum}{true}}{}
              \ifthenelse{\y=2 \and \x=4}{\setboolean{frustum}{true}}{}
              \ifthenelse{\y=1 \and \x=5}{\setboolean{frustum}{true}}{}
              \ifthenelse{\y=0 \and \x=5}{\setboolean{frustum}{true}}{}
                \ifthenelse{\boolean{wedge}}
                {
                  \node[wedgevertex] (v_\x_\y_\z) at (\x*\xdist-\xmax*\xdist/2,\y*\ydist,\z*\zdist-\zmax*\zdist/2) { }; 
               }
               {
               \ifthenelse{\boolean{frustum}}{
                  \node[frustumvertex] (v_\x_\y_\z) at (\x*\xdist-\xmax*\xdist/2,\y*\ydist,\z*\zdist-\zmax*\zdist/2) { }; 
               }
               {
                  \node[vertex] (v_\x_\y_\z) at (\x*\xdist-\xmax*\xdist/2,\y*\ydist,\z*\zdist-\zmax*\zdist/2) { }; 
               }
               }
          }
      }
  }
  \foreach \y in {1,...,\height} {
      \pgfmathtruncatemacro{\xmax}{(1+\height-\y)/2};
      \pgfmathtruncatemacro{\zmax}{(\height-\y)/2};
      \foreach \x in {0,...,\xmax} {
          \foreach \z in {0,...,\zmax} {
              \pgfmathtruncatemacro{\Expanded}{mod(\height-\y,2)};
              \pgfmathtruncatemacro{\prevy}{\y-1};
              \draw[diredge] (v_\x_\prevy_\z) -- (v_\x_\y_\z); %
              \pgfmathtruncatemacro{\nextx}{\x+1-\Expanded};
              \pgfmathtruncatemacro{\nextz}{\z+\Expanded};
              \draw[diredge] (v_\nextx_\prevy_\nextz) -- (v_\x_\y_\z); %
          }
      }
  }

  }%
  \end{tikzpicture}
  }
  \drawthreedimpyramid{10}{1.8}{1.4}{0.7}{5} %
}
     \caption{Pyramid with 
       wedge \dimValLayerWedge{0}{2}{\layernumber+1}
        \whitecirclerobust{}
       and
       restricted frustum 
       $\pyramid^2_{\layernumber+1,10}[\emptyset,\set{0},\set{0\mapsto 2}]$
       \bluecrossrobust{}.
     }%
  \label{fig:pyramidWedgeFrustum}
\end{figure}

We prove  \reflem{lem:ImmermansMainTechnicalLemma} by inductively
cutting the pyramid into a wedge and restricted \frustums to the left and
right of this wedge.  It will be convenient to focus on
$(\dimension,\coordinaterestrictionsize)$-\frustums
which grow in the dimensions corresponding to the topmost
\mbox{$\dimension - \coordinaterestrictionsize$ layers} 
(as the one in
\reffig{fig:pyramidWedgeFrustum}).  
More formally, we say that an 
$(\dimension,\coordinaterestrictionsize)$-\frustum
$\pyramid^\dimension_{\layernumber,\pyheight}%
[\indsetlo,\indsethi,\coordinaterestriction]$  
is \introduceterm{\topex} if 
\begin{equation}
  \label{eq:top-expanding-condition}
  \setdescr{(\layernumber + \dimindex) \bmod \dimension}{0\leq \dimindex
    \leq \dimension - 1 - \coordinaterestrictionsize}\cap
  (\indsetlo\cup\indsethi)
  =
  \emptyset
  \eqperiod   
\end{equation}
As we have done for digraphs with unique sinks in
\refdef{def:xor-formula}, we identify with each (restricted) frustum
$\pyramid$ the corresponding XOR formula $\digraphxorarg{\pyramid}$
containing all clauses given in
\refdef{def:xor-formula}\ref{xor-clause-source}
and~\ref{xor-clause-non-source}.   
We do not include hard-coded labels on the sources at the top layer
as in~\ref{xor-clause-sink}
but instead will always provide a labelling of that layer.%
\footnote{For readers more familiar with proof complexity language,
  our subgraphs correspond to formulas obtained by applying
  restrictions to 
  $\Digraphxorarg{\pyramid^\dimension_{\pyheight}}$.}

We now state our inductive claim.
\Reflem{lem:ImmermansMainTechnicalLemma} follows immediately once this
claim has been established, as the subgraph of a pyramid
$\pyramid^\dimension_{\pyheight}$ on or below layer~$\layernumber$ is
an (unrestricted) \topex 
\mbox{$(\dimension,0)$-\frustum $\pyramid^\dimension_{\layernumber,\pyheight}$} 
and, in particular,
$\Digraphxorarg{\pyramid^\dimension_{\pyheight}}$ with all vertices on
or above layer~$\layernumber$ consistently labelled is equivalent to
$\Digraphxorarg{\pyramid^\dimension_{\layernumber,\pyheight}}$  with
the same labelling of layer~$\layernumber$. 

\begin{claim}
  \label{claim:inductiveclaimpyramids}
  Let $\pyramid$ be a \topex
  $(\dimension,\coordinaterestrictionsize)$-\frustum and
  $\labelling_\layernumber$ 
  be a labelling of its top layer $\layernumber$.  
  Then for every set $\setS$ of
  $2^{\dimension-\coordinaterestrictionsize}-1$ vertices on or below
  layer $\layernumber+\dimension-\coordinaterestrictionsize$ there is a
  consistent labelling 
  of all vertices in~$\pyramid$
  that extends
  $\labelling_\layernumber$ and labels 
  every vertex in
  $\setS$ with~$0$. 
\end{claim}

The following proposition summarizes the core properties of frustums
that we will use when establishing
Claim~\ref{claim:inductiveclaimpyramids}.

\begin{proposition}
  \label{propo:frustums}
  Let
  $\pyramid=\pyramid^\dimension_{\layernumber,\pyheight}%
  [\indsetlo,\indsethi,\coordinaterestriction]$ 
  be a restricted \frustum   
  with a labelling~$\labelling_\layernumber$ of all vertices in the
  top layer~$\layernumber$.
  Then the following holds:
  \begin{enumerate}[label=(\alph*)]
  \item 
    \label{item:propo_frustums-a}
    There is a labelling~$\labelling_{\layernumber+1}$
    of all vertices in layer $\layernumber+1$ 
    of~$\pyramid$
    that is consistent 
    with~$\labelling_{\layernumber}$.
  \item 
    \label{item:propo_frustums-b}
    Let $\labelling_{\layernumber+1}$ be any labelling of
    layer~$\layernumber+1$ 
    in~$\pyramid$
    that is consistent with 
    $\labelling_{\layernumber}$ and suppose that  $\pyramid$ expands
    coordinate 
    $\dimindex
    $ 
    from layer $\layernumber$
    to layer $\layernumber+1$, \ie
    $\dimindex = \pyrrem{\dimension}{\layernumber}$ 
    and
    $\dimindex\notin\indsetlo\cup\indsethi$.
    Then 
    for any vertex
    $(\veczerodminusone{y}, \layernumber + 1)$ in~$\pyramid$
    it holds that 
    the    labelling 
    $\labelling^{\veczerodminusone{y}}_{\layernumber+1}$ 
    defined by 
    \begin{equation*}
      \labelling^{\veczerodminusone{y}}_{\layernumber+1}%
      (\veczerodminusone{x},\layernumber+1)
      \defi 
      \begin{cases}
        1 - \labelling_{\layernumber+1}(\veczerodminusone{x},\layernumber+1)
        &
        \text{if  $x_i = y_i$ for all  $i \neq \dimindex$,} 
        \\
        \labelling_{\layernumber+1}(\veczerodminusone{x}, \layernumber+1)
        &
        \text{otherwise}
      \end{cases}
    \end{equation*}
    is also consistent with~$\labelling_{\layernumber}$.
  \end{enumerate}
\end{proposition}

\begin{proof}
  We first note that the set of
  \XOR
  constraints between two layers
  $\layernumber$ and $\layernumber+1$ 
  can be partitioned into
  several connected
  components.  
  Each of the components forms a ``one-dimensional line'' in the
  direction of the expanding coordinate
  $\dimindex = \pyrrem{\dimension}{\layernumber}$.  
  More formally, two vertices from layers $\layernumber$ 
  and~$\layernumber+1$ 
  are on the same line if and only if they agree on
  the coordinates $x_i$ for all $i\neq\dimindex$.  
  These lines form the connected components of the graph induces on
  layers $\layernumber$ and~$\layernumber+1$.  
  All such lines between layers $\layernumber$ and~$\layernumber+1$
  isomorphic and their shape depends on whether
  $\dimindex\in \indsetlo$, $\dimindex\in \indsethi$, or
  $\dimindex\notin\indsetlo\cup\indsethi$.  
  See \reffig{fig:between_layers} for an illustration.
  We remark that in
  \reftwofigs{fig:between_layersB}{fig:between_layersC} the vertex
  with in-degree $1$ and its predecessor form a binary XOR clause in
  $\digraphxorarg{\pyramid}$.    

\makeatletter{}%

\newcommand{\captionspacing}{\vspace{1mm}}

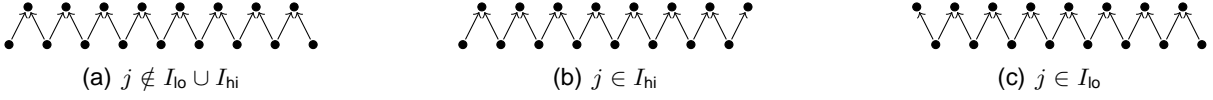
\begin{figure}
  \centering
  \subfigure[$\dimindex\notin\indsetlo\cup\indsethi$]%
  {
    \label{fig:between_layersA}
    {%
      \begin{tikzpicture}%
        [vertex/.style={circle,draw=black,fill=black,%
          inner sep=0pt,minimum  size=3pt},%
        diredge/.style={->,shorten <=.5pt, shorten >=.5pt}] 
        \foreach \height/\xdist/\ydist in {8/.5/.5} {

          \foreach \y in {0,1} {
            \pgfmathsetmacro{\diff}{\height-\y};
            \foreach \x in {0,...,\diff} {
              \node[vertex] (v_\x_\y) at (\x*\xdist-\diff*\xdist/2,\y*\ydist) {}; 
            }
          }
          \foreach \y in {1} {
            \pgfmathsetmacro{\diff}{\height-\y};
            \foreach \x in {0,...,\diff} {
              \pgfmathtruncatemacro{\nextx}{1+\x};
              \pgfmathtruncatemacro{\prevy}{\y-1};
              \draw[diredge] (v_\x_\prevy) -- (v_\x_\y);
              \draw[diredge] (v_\nextx_\prevy) -- (v_\x_\y); 
            }
          }
          
        }%
      \end{tikzpicture}
    }%
    \captionspacing
  } 
  \hfill
  \subfigure[$\dimindex\in\indsethi$]%
  {
    \label{fig:between_layersB}
    {%
      \begin{tikzpicture}%
        [vertex/.style={circle,draw=black,fill=black,%
          inner sep=0pt,minimum  size=3pt},%
        diredge/.style={->,shorten <=.5pt, shorten >=.5pt}] 
        \foreach \height/\xdist/\ydist in {8/.5/.5} {
          
          \pgfmathtruncatemacro{\heightminusone}{\height-1};
          \pgfmathtruncatemacro{\heightminustwo}{\height-2};
          \foreach \y in {0,1} {
            \pgfmathsetmacro{\diff}{\height-\y};
            \foreach \x in {0,...,\heightminusone} {
              \node[vertex] (v_\x_\y) at (\x*\xdist-\diff*\xdist/2,\y*\ydist) {}; 
            }
          }
          \foreach \y in {1} {
            \pgfmathsetmacro{\diff}{\height-\y};
            \foreach \x in {0,...,\heightminustwo} {
              \pgfmathtruncatemacro{\nextx}{1+\x};
              \pgfmathtruncatemacro{\prevy}{\y-1};
              \draw[diredge] (v_\x_\prevy) -- (v_\x_\y);
              \draw[diredge] (v_\nextx_\prevy) -- (v_\x_\y); 
            }
            \draw[diredge] (v_\heightminusone_0) -- (v_\heightminusone_1);
          }
          
        }%
      \end{tikzpicture}
    }%
  } 
  \hfill
  \subfigure[$\dimindex\in\indsetlo$]%
  {
    \label{fig:between_layersC}
    {%
      \begin{tikzpicture}%
        [vertex/.style={circle,draw=black,fill=black,%
          inner sep=0pt,minimum  size=3pt},%
        novertex/.style={circle,draw=white,fill=white,%
          inner sep=0pt,minimum  size=3pt},%
        diredge/.style={->,shorten <=.5pt, shorten >=.5pt}] 
        \foreach \height/\xdist/\ydist in {8/.5/.5} {
          
          \pgfmathtruncatemacro{\heightminusone}{\height-1};
          \pgfmathtruncatemacro{\heightminustwo}{\height-2};
          \foreach \y in {0,1} {
            \pgfmathsetmacro{\diff}{\height-\y};
            \pgfmathtruncatemacro{\start}{1-\y};
            \foreach \x in {\start,...,\diff} {
              \node[vertex] (v_\x_\y) at (\x*\xdist-\diff*\xdist/2,\y*\ydist) {}; 
            }
          }
          \foreach \y in {1} {
            \pgfmathsetmacro{\diff}{\height-\y};
            \foreach \x in {1,...,\diff} {
              \pgfmathtruncatemacro{\nextx}{1+\x};
              \pgfmathtruncatemacro{\prevy}{\y-1};
              \draw[diredge] (v_\x_\prevy) -- (v_\x_\y);
              \draw[diredge] (v_\nextx_\prevy) -- (v_\x_\y); 
            }
            \draw[diredge] (v_1_0) -- (v_0_1);
          }

          \node[novertex] (v_0_0) at (0*\xdist-\height*\xdist/2,0*\ydist) {}; 
          
        }%
      \end{tikzpicture}
    }%
  }
  \caption{%
    Shapes of
    connected component between layers $\layernumber$ and
    $\layernumber+1$
    expanding in dimension~$\dimindex$.
  }  
  \label{fig:between_layers}
\end{figure} 

  If the layer is not expanding (as depicted in
  \reftwofigs{fig:between_layersB}{fig:between_layersC}) 
  and the upper layer~$\layernumber$
  is entirely labelled, then it is not hard to see that there is a
  unique consistent labelling of the lower-level vertices of each line
  (determined by propagating values from right to left in
  \reffig{fig:between_layersB}
  and from left to right in
  \reffig{fig:between_layersC}).
  As all lines are disjoint this gives a unique labelling of the
  entire layer $\layernumber+1$ that is consistent with the labelling
  of layer~$\layernumber$.  
  If the layer expands (\ie, if
  $\dimindex\notin\indsetlo\cup\indsethi$ as illustrated in
  \reffig{fig:between_layersA}), then we have more freedom.  
  Indeed, if we label either the rightmost or the leftmost vertex at
  the bottom layer with 0, then we have the same situation as in
  \reftwofigs{fig:between_layersB}{fig:between_layersC}, 
  respectively.  
  This concludes the proof of item~\ref{item:propo_frustums-a}
  in the proposition.

  For item~\ref{item:propo_frustums-b},
  first observe that the condition
  $\dimindex\notin\indsetlo\cup\indsethi$
  means that we are in the case depicted in
  \reffig{fig:between_layersA}.
  This means that if
  we have a  consistent labelling of 
  the upper and lower part of a line,
  then flipping all values 
  at the lower level
  yields another consistent labelling.  
  This is 
  so since
  every XOR clause contains exactly two
  vertices from the 
  lower part.
  Hence, flipping both of these vertices does not change the parity of the
  variables in the XOR clause but leaves the clause satisfied.
  As all 
  lines between layers~$\layernumber$ \mbox{and~$\layernumber + 1$} are 
  disconnected from each other,
  flipping all values in one line gives another consistent labelling
  for the whole layer $\layernumber+1$, which is precisely 
  what is claimed in item~\ref{item:propo_frustums-b}. The proposition
  follows.
\end{proof}

\begin{proof}[Proof of Claim~\ref{claim:inductiveclaimpyramids}]
  The proof is by induction over 
  decreasing values of $\coordinaterestrictionsize$,
  the base case being
  $\coordinaterestrictionsize = \dimension$.
  As 
  $\setsize{\setS} = 2^{\dimension - \coordinaterestrictionsize} - 1 =  0$
  if
  $\coordinaterestrictionsize = \dimension$, 
  in this case we only have to ensure that there is a 
  consistent labelling of the entire \frustum that is consistent with
  the labelling of the top layer.  
  This follows from inductively applying
  \refpr{propo:frustums}\ref{item:propo_frustums-a}
  layer by layer.

  For the 
  inductive step,
  assume that the claim 
  holds
  for all
  \topex $(\dimension,\coordinaterestrictionsize+1)$-\frustums.
  We   want prove it for a \topex
  $(\dimension,\coordinaterestrictionsize)$-\frustum
  $\pyramid=\pyramid^\dimension_{\layernumber,\pyheight}%
  [\indsetlo,\indsethi,\coordinaterestriction]$.  
  Let 
  $\dimindex 
  = \pyrrem{\dimension}{\layernumber} 
  = \layernumber
  \bmod   %
  \dimension$ 
  be the 
  dimension
  that
  expands from layer $\layernumber$ to $\layernumber+1$.  
  As $\pyramid$ is \topex and $\coordinaterestrictionsize<\dimension$
  we have $\dimindex\notin\indsetlo\cup\indsethi$.  
  For some well-chosen
  $\valW \in [0, \maxval(\dimindex,\layernumber+1)]$  
  to be specified shortly,  we 
  partition 
  $\pyramid$ on and below layer $\layernumber+1$ into the
  wedge 
  \begin{subequations}
    \begin{equation}
      \label{eq:split-wedge}
      \dimValLayerWedge{\dimindex}{\valW}{\layernumber+1}    
    \end{equation}
    and two disjoint $(\dimension,\coordinaterestrictionsize+1)$-\frustums: the ``right'' frustum
    \begin{equation}
      \label{eq:split-frustum-1}
      \pyramid^\dimension_{\layernumber+1,\pyheight}
      [\indsetlo\cup\set{\dimindex},\indsethi,
      \coordinaterestriction\cup\set{\dimindex \mapsto\valW}]    
    \end{equation}
    and the ``left'' frustum (depicted by \bluecross{} in Figure~\ref{fig:pyramidWedgeFrustum})
    \begin{equation}
      \label{eq:split-frustum-2}
      \pyramid^\dimension_{\layernumber+1,\pyheight}
      [\indsetlo,\indsethi\cup\set{\dimindex},
      \coordinaterestriction\cup\set{\dimindex\mapsto \valW}]
      \eqperiod
    \end{equation}
  \end{subequations}
  We choose the position $\valW$ of the wedge so that both
  $(\dimension,\coordinaterestrictionsize+1)$-\frustums
  in~\refeq{eq:split-frustum-1}
  and~\refeq{eq:split-frustum-2}
  contain at  most 
  $(\setsize{\setS}-1) / 2 \leq 2^{\dimension-(\coordinaterestrictionsize+1)}-1$ 
  vertices from
  $\setS$ 
  (which implies
  that the wedge~\refeq{eq:split-wedge}
  contains at least one vertex from $\setS$).   
  To be more specific,
  we choose the largest $\valW\geq 0$ such that 
  the 
  left frustum~\refeq{eq:split-frustum-2} contains at most
  $(\setsize{\setS}-1) / 2$ vertices 
  $\setS_\valW \subseteq \setS$. Such an $\valW$ exists as $\setS_0=\emptyset$.
  
  If $\valW$ reached the maximum $\maxval(\dimindex,\layernumber+1)$,
  then empty right frustum~\refeq{eq:split-frustum-1} clearly contains
  no vertices from $\setS$.   
  Otherwise let $\setS_{\valW+1}$ be the set of vertices from $\setS$
  left of the wedge {at position $\valW+1$}.  
  By the choice of $\valW$ we have
  $\setsize{\setS_{\valW+1}}>(\setsize{\setS}-1) / 2$.  
  Furthermore, because all vertices in $\setS$ are below layer
  $\layernumber+1$ it follows that $\setS_{\valW+1}\setminus
  \setS_\valW$ is contained in the wedge~\refeq{eq:split-wedge} at
  position $\valW$.  
  Hence the right frustum~\refeq{eq:split-frustum-1} contains at most
  $\setsize{\setS\setminus\setS_{\valW+1}} \leq (\setsize{\setS}-1) /
  2$ vertices. 

  Now we proceed as follows.
  First we use
  \refpr{propo:frustums}\ref{item:propo_frustums-a}
  to obtain any consistent labelling of all vertices in
  layer~$\layernumber+1$. 
  Consider the set of all vertices
  $v = (\veczerodminusone{x},\layernumber+1)$ in layer
  $\layernumber+1$ with $x_\dimindex=\valW$,
  \ie the topmost vertices in the
  wedge~\refeq{eq:split-wedge}.
  Note that this set of vertices form a hyperplane through, and
  perpendicular to, the disconnected parallel lines discussed in
  \refpr{propo:frustums}.
  We go over these vertices~$v$ one by one and 
  flip every \mbox{$1$-labelled}~$v$ to~$0$.
  As $\dimindex$ is the expanding coordinate from layer $\layernumber$
  \mbox{to $\layernumber+1$}, 
  after every such flip we can apply
  \refpr{propo:frustums}\ref{item:propo_frustums-b}
  to relabel the rest of the line through~$v$ as needed.
  In this way, all vertices   in the top layer of the
  wedge~\refeq{eq:split-wedge}, get labelled by~$0$,
  and we label all other vertices in the wedge 
  with~$0$ also. 
  It follows from repeated application of
  \refpr{propo:frustums}\ref{item:propo_frustums-b}
  that the end result is a labelling of \mbox{layer~$\layernumber+1$} that is
  consistent with~$\labelling_\layernumber$.
  Note that this labels every vertex from $\setS$ within the wedge
  with $0$ and moreover, layer
  $\layernumber+1$ contains no vertices from $\setS$ outside of the
  wedge.
  This is because if some vertex from $\setS$ is on layer
  $\layernumber+1$, then we have by the assumption in
  Claim~\ref{claim:inductiveclaimpyramids} that $\setsize{\setS}=1$
  and hence this one labelled vertex is guaranteed to be in the wedge
  by the choice of $\valW$.
  In this way we obtain a
  labelling for the top layer of both
  $(\dimension,\coordinaterestrictionsize+1)$-\frustums,
  and 
  we then  apply induction to consistently label all vertices in both
  frustums   in such a way that
  every vertex in $\setS$ is set to $0$.  
  Now we argue that the disjoint union of the all-zero
  labelling of the wedge and the consistent labellings
  of both
  frustums is a consistent labelling of $\pyramid$. 
  Clearly, every (non-source) clause in $\digraphxorarg{\pyramid}$
  that is entirely contained in the wedge is satisfied by the
  all-zero labelling of the wedge.  
  In the same way,  every clause that is entirely contained in one of
  the two
  frustums is satisfied by their consistent labellings.  
  It remains to consider clauses that contain variables from the wedge
  as well as from one of the frustums.  
  By construction, those clauses have the form $(v,w_1,w_2,0)$, for a
  vertex $v$ with in-neighbours $w_1$, $w_2$, where the vertex $v$ and
  one of its neighbours (say $w_1$) is within the frustum and the
  other neighbour is inside the wedge.  
  As the edge $(w_1,v)$ inside the frustum forms the binary clause
  $(v,w_1,0)$ it follows that the consistent labelling of the frustum
  guarantees that the parity of $v$ and $w_1$ is even.  
  Because $w_2$ is labelled $0$ in the wedge, the merged labelling
  satisfies $(v,w_1,w_2,0)$.  
  The claim follows.
\end{proof}

As noted above, our proof of 
Claim~\ref{claim:inductiveclaimpyramids}
also establishes
\reflem{lem:ImmermansMainTechnicalLemma}.

}

In~\cite{Immerman.1981} $\log n$-dimensional pyramids (where $n$ is
the number of vertices) are used to prove a
$\Bigomega{2^{\sqrt{\log n}}}$ lower bound on the quantifier depth of
full first-order counting logic.  The next lemma
shows
that if we instead choose the dimension to be logarithmic in the
number of variables (\ie pebbles)
in the game, 
we get an improved quantifier depth
lower bound for the $\pebblestd$-variable fragment.

\begin{lemma}\label{lem:pyramidLowerBound}
  For every $\dimension\geq 2$ and height~$h$, \firstplayer does not
  win the $2^\dimension$-pebble game on
  $\Digraphxorarg{\pyramid^\dimension_{h}}$ within 
  $h/(\dimension-1)- 1$
  rounds.  
\end{lemma}

\begin{proof}
  We show that \secondplayer has a counter-strategy to answer
  consistently for at least $\lfloor h/(\dimension-1)\rfloor-1$ rounds and
  therefore \firstplayer needs at least $\lfloor h/(\dimension-1)\rfloor >
  h/(\dimension- 1)-1$ rounds to win.  
  Starting at the top layer~$\rownumber_1=0$,
  she maintains the invariant that at the start of round~$r$ 
  she has a consistent labelling of all vertices from layer~$0$
  to layer~$\rownumber_{r}$ with the property that
  there is no pebble on layers~$\rownumber_{r}+1$
  to~$\rownumber_{r} + \dimension - 1$.

  Whenever \firstplayer places a pebble on or above layer~$\rownumber_r$, 
  \secondplayer responds
  according to the consistent labelling and whenever \firstplayer
  puts a pebble on or below layer $\rownumber_r + \dimension$, she
  answers~$0$,
  and in both cases sets $\rownumber_{r+1} = \rownumber_r$.
  Note that as long as \firstplayer places pebbles in this way, the
  game can go on forever. 
  Since there are never more  than
  $2^\dimension - 1$
  pebbles left on vertices on or below layer $\rownumber_r + \dimension$
  (when \firstplayer runs out of pebbles the next move must be a removal),
  the conditions needed for 
  \reflem{lem:ImmermansMainTechnicalLemma} 
  to apply are never violated.

  Thus, the interesting case is when \firstplayer places a pebble
  between layer $\rownumber_r + 1$ and 
  $\rownumber_r + \dimension-1$.
  Then \secondplayer uses 
  \reflem{lem:ImmermansMainTechnicalLemma} to extend her labelling to
  the first layer $\rownumber_{r+1} > \rownumber_r$ such
  that there is no pebble on layers $\rownumber_{r+1} + 1$ to
  $\rownumber_{r+1} + \dimension -1$, after which she answers the query
  according to the new   labelling.  
  \ifthenelse{\boolean{conferenceversion}}
  {Note that}
  {It is worth noting that} 
  when \secondplayer skips downward from
  layer~$\rownumber_{r}$ to 
  layer~$\rownumber_{r+1}$ she might jump over a lot of layers in one
  go, but if so there is at least one pebble for every
  \mbox{$(\dimension-1)$th layer} forcing such a big jump.
  We see that following this strategy \secondplayer survives for at
  least $\lfloor h/(\dimension-1)\rfloor-1$ rounds, and this establishes
  the lemma.
\end{proof}

Putting the pieces together, we can now present the lower bound for the
$\pebblestd$-pebble game in \reflem{lem:pyramids}.

\begin{proof}[Proof of \reflem{lem:pyramids}]
  Recall that we want to prove that
  for all \mbox{$\pebblesupperlhs\geq 3$} and 
  \mbox{$\indexlhs \geq 3$} there is an 
  $\indexlhs$-variable \mbox{$3$-\XOR{}} formula
  $\xformf$ 
  on which
  \firstplayer
  wins the \mbox{$3$-pebble} game but
  cannot win the $\pebblesupperlhs$-pebble game
  within 
  $\tfrac{1}{\ceiling{\log \pebblesupperlhs}}
  \indexlhs^{1/(1 + 
    \ceiling{\log \pebblesupperlhs})} - 2$ 
  rounds.
  If $\indexlhs < (5\ceiling{ \log \pebblesupperlhs })^{(\ceiling{ \log \pebblesupperlhs } + 1)}$, then the round
  lower bound is trivial and we let $\xformf$ be, \eg, the 3-variable
  formula $\Digraphxorarg{\pyramid^1_1}$ plus $\indexlhs-3$ auxiliary
  variables on which \firstplayer{} needs 3 rounds to win.
  Otherwise, we choose the formula to be
  $\xformf =  \Digraphxorarg{\pyramid^\dimension_h}$
  for parameters
  $\dimension
   = %
   \ceiling{ \log \pebblesupperlhs }
  $
  and
  $h 
  = %
  \Floor{m^{1/(d+1)}}-1$. 
  Note that $\pyramid^\dimension_h$ contains less than $(h+1)^{d+1}\leq m$
  vertices and we can add dummy variables to reach exactly $m$.
  Since the graph~$\pyramid^\dimension_h$ has in-degree~$2$,
  \reflem{lem:UpperBoundPebblingDAG} says that \firstplayer wins the
  \mbox{$3$-pebble} game
  as claimed in
  \reflem{lem:pyramids}\ref{item:immerman-a}.
  The lower bound for the \mbox{$\pebblesupperlhs$-pebble} game 
  in \reflem{lem:pyramids}\ref{item:immerman-b}
  follows from \reflem{lem:pyramidLowerBound} and the oberservation
  that because $h\geq 5\dimension-1$ we have $h/(\dimension-1)\geq
  (h+1)/\dimension$ and hence
  \begin{align}
    \label{eq:1}
    h/(\dimension-1)- 1 \leq (h+1)/\dimension- 1
    &= \tfrac{1}{\ceiling{\log \pebblesupperlhs}}
      \FLOOR{\indexlhs^{1/(1 + 
      \ceiling{\log \pebblesupperlhs})}} - 1 \\
    &\leq \tfrac{1}{\ceiling{\log \pebblesupperlhs}}
      \indexlhs^{1/(1 + 
      \ceiling{\log \pebblesupperlhs})} - 2
      \eqperiod
  \end{align}
  The lemma follows.
\end{proof}

\makeatletter{}%
\newcommand{\boundaryexp}[3]%
    {$({#2},{#1},{#3})$\nobreakdash-boundary expander\xspace}
\newcommand{\nmboundaryexp}[5]%
    {${#1}\times{#2}$ \boundaryexp{#3}{#4}{#5}}
\newcommand{\boundaryexpnodeg}[2]%
    {$({#1},{#2})$\nobreakdash-boundary expander\xspace}
\newcommand{\nmboundaryexpnodeg}[4]{${#1}\times{#2}$ \boundaryexpnodeg{#3}{#4}}

\newcommand{\expguarantee}{\expansionguarantee}
\newcommand{\expfactor}{\expansionfactor}
\newcommand{\expdegree}{\expanderdegree}
\newcommand{\lsize}{\leftsize}
\newcommand{\rsize}{\rightsize}

\newcommand{\boundaryexpnodegstd}%
    {\boundaryexpnodeg{\expguarantee}{\expfactor}}
\newcommand{\boundaryexpstd}%
    {\boundaryexp{\expdegree}{\expguarantee}{\expfactor}}
\newcommand{\nmboundaryexpstd}%
    {\nmboundaryexp{\lsize}{\rsize}{\expdegree}{\expguarantee}{\expfactor}}
\newcommand{\nmboundaryexpnodegstd}%
    {\nmboundaryexpnodeg{\lsize}{\rsize}{\expguarantee}{\expfactor}}

\newcommand{\aboundaryexpnodegstd}%
    {an \boundaryexpnodeg{\expguarantee}{\expfactor}}
\newcommand{\aboundaryexpstd}%
    {an \boundaryexp{\expdegree}{\expguarantee}{\expfactor}}
\newcommand{\annmboundaryexpstd}%
    {an \nmboundaryexp{\lsize}{\rsize}{\expdegree}{\expguarantee}{\expfactor}}
\newcommand{\annmboundaryexpnodegstd}%
    {an \nmboundaryexpnodeg{\lsize}{\rsize}{\expguarantee}{\expfactor}}

\newcommand{\Aboundaryexpnodegstd}%
    {An \boundaryexpnodeg{\expguarantee}{\expfactor}}
\newcommand{\Aboundaryexpstd}%
    {An \boundaryexp{\expdegree}{\expguarantee}{\expfactor}}
\newcommand{\Annmboundaryexpstd}%
    {An \nmboundaryexp{\lsize}{\rsize}{\expdegree}{\expguarantee}{\expfactor}}
\newcommand{\Annmboundaryexpnodegstd}%
    {An \nmboundaryexpnodeg{\lsize}{\rsize}{\expguarantee}{\expfactor}}

\section{Hardness Condensation}
\label{sec:hardness-condensation}

In this section we establish \reflem{lem:hardnessCondensationXOR},
which shows how to convert an \XOR formula into a harder formula over
fewer variables. As discussed in the introduction, 
this part of our construction relies heavily on Razborov's recent
paper~\cite{Razborov16NewKind}. We follow his line of
reasoning closely below, but translate it from proof complexity to 
a pebble game argument for bounded variable logics.

A key technical concept in the proof is graph expansion. Let us define
the particular type of expander graphs we need and then discuss some
crucial properties of these graphs. We use standard 
notation, letting
$\expandergraph =
(\leftvertexset \disjointunion  \rightvertexset,E)$ 
denote a bipartite graph with left vertex set~$\leftvertexset$
and
right vertex set~$\rightvertexset$.
We let
$N^\graph \bigl( \leftvertexsubset \bigr) 
=
\Setdescr{v}{\{u,v\}\in E(\graph),u\in \leftvertexsubset}$
\ifthenelse{\boolean{conferenceversion}}
{denote the right neighbours}
{denote the set of neighbour vertices on the right}
of a left vertex subset~$\leftvertexsubset \subseteq \leftvertexset$
(and vice versa for right vertex subsets).

\begin{definition}[Boundary expander]
  A bipartite graph 
  $\expandergraph = (\leftvertexset \disjointunion
  \rightvertexset,E)$ 
  is an
  \introduceterm{\nmboundaryexpnodegstd{} graph}
  if
  \mbox{$\setsize{\leftvertexset}=\leftsize$},
  \mbox{$\setsize{\rightvertexset}=\rightsize$},
  and for every set 
  $\leftvertexsubset\subseteq\leftvertexset$,
  $\setsize{\leftvertexsubset}\leq \expansionguarantee$,
  it holds that
  $\Setsize{\boundary^{\expandergraph}(\leftvertexsubset)} \geq
  \expansionfactor\setsize{\leftvertexsubset}$, 
  where the \emph{boundary}
  $\boundary^{\expandergraph}(\leftvertexsubset)$ is the set of all 
  $v \in N^{\graph}(\leftvertexsubset)$
  having a unique neighbour in~$U'$, meaning that 
  \mbox{$\Setsize{N^{\graph}(v)\cap\leftvertexsubset}=1$}. 
  \Aboundaryexpstd is  \aboundaryexpnodegstd where additionally 
  $\Setsize{N^\graph(u)} \leq \expanderdegree$ 
  for all
  $u\in\leftvertexset$,
  \ie the graph has left degree bounded by~$\expanderdegree$. 
\end{definition}

In what follows, we will omit~$\graph$ from the notation when 
the graph is clear from context.

In any \boundaryexpnodegstd
\ifthenelse{\boolean{conferenceversion}}
{with $\expansionfactor > 0$}
{with expansion $\expansionfactor > 0$}
it holds that any left vertex subset
$\leftvertexsubset \subseteq \leftvertexset$
of size 
$\setsize{\leftvertexsubset} \leq \expansionguarantee$
has a partial matching into~$\rightvertexset$
\ifthenelse{\boolean{conferenceversion}}
{where the} 
{where in addition the} 
vertices in~$\leftvertexsubset$ can be ordered
in such a way that every vertex $u_i \in \leftvertexsubset$ is
matched to a vertex outside of the neighbourhood of
the preceding vertices $u_1, \ldots, u_{i-1}$.
The proof of this fact is sometimes
referred to as a \introduceterm{peeling argument}.

\begin{lemma}[Peeling lemma] 
  \label{lem:peeling_lemma}
  Let
  $\expandergraph = (\leftvertexset \disjointunion \rightvertexset,E)$ 
  be   \aboundaryexpnodegstd with
  $\expansionguarantee\geq1$
  and
  $\expansionfactor > 0$.
  Then
  for every set $\leftvertexsubset \subseteq \leftvertexset$,
  $|\leftvertexsubset|=\alternativetoell\leq\expansionguarantee$ there is an
  ordering $u_1,\ldots,u_\alternativetoell$ of its vertices and a sequence of
  vertices $v_1,\ldots,v_\alternativetoell \in \rightvertexset$ such that $v_i\in
  N(u_i)\setminus N(\{u_1,\ldots,u_{i-1}\})$. 
\end{lemma}

\ifthenelse{\boolean{conferenceversion}}
{%
\begin{proof}[Proof sketch]
  Fix any \mbox{$v_\alternativetoell \in \boundary(\leftvertexsubset)$} 
  and let
  $u_\alternativetoell \in \leftvertexsubset$ 
  be the unique vertex
  such that 
  $\Setsize{N(v_\alternativetoell)\cap\leftvertexsubset}=\{u_\alternativetoell\}$.
  Then  it holds that
  $v_\alternativetoell \in 
  N(u_\alternativetoell) \setminus 
  N \bigl(\leftvertexsubset\setminus\{u_\alternativetoell\}\bigr)$.
  By induction  we can now find sequences
  $u_1,\ldots,u_{\alternativetoell-1}$ and $v_1,\ldots,v_{\alternativetoell-1}$ 
  for~$\leftvertexsubset\setminus\{u_\alternativetoell\}$
  such that $v_i\in
  N(u_i)\setminus N(\{u_1,\ldots,u_{i-1}\})$,
  to which we can append
  $u_\alternativetoell$ and~$v_\alternativetoell$ at the end. 
  The lemma follows.
\end{proof}
}
{%
\begin{proof}
  The proof is by induction on $\alternativetoell$.
  The base case $\alternativetoell=1$ 
  is immediate
  since
  $\expansionguarantee\geq1$
  and
  $\expansionfactor > 0$ 
  implies that no left vertex can be isolated.
  For the 
  inductive step,
  suppose the lemma holds 
  for~$\alternativetoell - 1$.
  To 
  construct
  the sequence $v_1,\ldots,v_\alternativetoell$ we first fix
  $v_\alternativetoell$ to 
  be any vertex in $\boundary(\leftvertexsubset)$, which
  has to exist since
  $\Setsize{\boundary(\leftvertexsubset)}
  \geq \mbox{$\expansionfactor \setsize{\leftvertexsubset} > 0$}$.
  The fact that $v_\alternativetoell$ is in the boundary of~$\leftvertexsubset$
  means that there is a unique~$u_\alternativetoell \in
  \leftvertexsubset$
  such that 
  $\Setsize{N(v_\alternativetoell)\cap\leftvertexsubset}=\{u_\alternativetoell\}$.
  Thus, for this pair $(u_\alternativetoell, v_\alternativetoell)$ it holds that
  $v_\alternativetoell \in 
  N(u_\alternativetoell) \setminus 
  N \bigl(\leftvertexsubset\setminus\{u_\alternativetoell\}\bigr)$.
  By the induction hypothesis we can now find sequences
  $u_1,\ldots,u_{\alternativetoell-1}$ and $v_1,\ldots,v_{\alternativetoell-1}$ 
  for~$\leftvertexsubset\setminus\{u_\alternativetoell\}$
  such that $v_i\in
  N(u_i)\setminus N(\{u_1,\ldots,u_{i-1}\})$,
  to which we can append
  $u_\alternativetoell$ and~$v_\alternativetoell$ at the end. The lemma follows.
\end{proof}
}
For a right vertex subset
$\rightvertexsubset \subseteq \rightvertexset$ 
in
$\expandergraph = (\leftvertexset \disjointunion \rightvertexset,E)$
we define the \introduceterm{kernel} 
$\Ker \bigl( \rightvertexsubset \bigr) \subseteq \leftvertexset$
to be the set of all left vertices whose entire neighbourhood is
contained in 
$\rightvertexsubset$, \ie
\begin{equation}
  \label{eq:kernel}
  \Ker \bigl( \rightvertexsubset \bigr) =
  \Setdescr{\mbox{$u\in\leftvertexset$}}{N(u)\subseteq\rightvertexsubset}
  \eqperiod
\end{equation}
We let 
$\expandersubgraph{\expandergraph}{\rightvertexsubset}$ 
denote the subgraph of $\expandergraph$ induced on
$
\bigl(
\leftvertexset\setminus\Ker(\rightvertexsubset) 
\bigr)
\disjointunion
\bigl(
\rightvertexset\setminus\rightvertexsubset
\bigr)
$.
In other words,
we obtain
$\expandersubgraph{\expandergraph}{\rightvertexsubset}$
from~$\expandergraph$ by first deleting $\rightvertexsubset$ and
afterwards all isolated vertices from $\leftvertexset$
(assuming that there were no isolated left vertices before, which is
true if $\expandergraph$ is expanding). 

\ifthenelse{\boolean{conferenceversion}}
{The next lemma states that for any small enough right vertex
  set~$\rightvertexsubset$ in an expander~$\expandergraph$ we can}
{The next lemma states that if $\expandergraph$ is an expander graph,
  then for any small enough right vertex set~$\rightvertexsubset$ we
  can always}
find a \introduceterm{closure}
$\closure\bigl(\rightvertexsubset\bigr) \supseteq \rightvertexsubset$
with a small kernel such that 
\ifthenelse{\boolean{conferenceversion}}
{$\expandersubgraph{\expandergraph}{\closure(\rightvertexsubset)}$
  has good expansion.}   
{the subgraph
  $\expandersubgraph{\expandergraph}{\closure(\rightvertexsubset)}$
  has good boundary expansion.}   
\ifthenelse{\boolean{conferenceversion}}
{The proof of this lemma (albeit with slightly different parameters)
  can be found  in~\cite{Razborov16NewKind} and is also provided in
  the full-length version of this paper.} 
{The proof of this lemma (albeit with slightly different parameters)
  can be found  in~\cite{Razborov16NewKind}, but we also include it in
  \refapp{app:existence-expander}
  for completeness.}

\begin{lemma}[\cite{Razborov16NewKind}]
  \label{lem:ClosedSet}
  Let $\expandergraph$ be an
  \boundaryexpnodeg{\expguarantee}{2}.
  Then for every $\rightvertexsubset \subseteq \rightvertexset$ with
  $\setsize{\rightvertexsubset}\leq \expansionguarantee/2$ 
  there exists a subset
  $\closure\bigl(\rightvertexsubset\bigr) \subseteq \rightvertexset$ 
  with
  $\closure(\rightvertexsubset) \supseteq \rightvertexsubset$
  such that 
  $\Setsize{\Ker \bigl( \closure \bigl( \rightvertexsubset \bigr)\bigr)}
  \leq \Setsize{\rightvertexsubset}$
  and the induced subgraph
  $\expandersubgraph{\expandergraph}{\closure(\rightvertexsubset)}$ is
  an
  \boundaryexpnodeg{\expguarantee/2}{1}.
\end{lemma}

\ifthenelse{\boolean{conferenceversion}}
{In order for \reftwolems{lem:peeling_lemma}{lem:ClosedSet} to be
  useful, we need to know that there exist good expanders.
  This can be established by a standard probabilistic argument.
  A proof of the next lemma is given in~\cite{Razborov16NewKind} and can
  also be found in  the full version of this paper.}  
{Note that 
  \reflem{lem:ClosedSet} 
  is a purely existential result. We do not know how the closure is
  constructed and, in particular, 
  if we want to choose closures of minimal size, then 
  $\rightvertexset_1 \subseteq \rightvertexset_2$
  does not necessarily imply
  $\closure(\rightvertexset_1) \subseteq \closure(\rightvertexset_2)$.

  In order for \reftwolems{lem:peeling_lemma}{lem:ClosedSet} to be
  useful, we need to know that there exist good enough boundary expanders.
  To prove this, one can just fix a left vertex set~$\leftvertexset$ of
  size~$\leftsize$ and a right vertex set~$\rightvertexset$ of
  size~$\rightsize$ and then for every $u\in\leftvertexset$ choose
  $\expanderdegree$~neighbours from~$\rightvertexset$ uniformly and
  independently at random.
  \mbox{A standard} probabilistic argument shows that 
  with high probability
  this random graph  is an
  \nmboundaryexp{\leftsize}{\rightsize}{\expdegree}{\expguarantee}{2}
  for appropriately chosen parameters.
  We state this formally as a lemma below. A similar lemma is proven
  in~\cite{Razborov16NewKind} but we also 
  provide a proof in \refapp{app:existence-expander} 
  for the convenience of the reader.}

\begin{lemma}%
  \label{lem:expanderexistnew}
  There is 
  an absolute constant $\mindegree \in \Nplus$ 
  such that for all 
  integers
  $\expanderdegree$, 
  $\expansionguarantee$,  
  and~$\leftsize$
  satisfying
  $
  \expanderdegree \geq \mindegree
  $
  and 
  \mbox{$(\expansionguarantee\expanderdegree)^{2\expanderdegree} \leq \leftsize$}
  there exist
  \nmboundaryexp{\leftsize}{\Ceiling{{\leftsize}^{3/\expanderdegree}}}
  {\expdegree}{\expguarantee}{2}{}s.
\end{lemma}

\ifthenelse{\boolean{conferenceversion}}
{}
{%
  For readers familiar with expander graphs from other contexts, 
  it might be worth pointing out that the parameters 
  above
  are different from what tends to be the standard expander graph
  settings of 
  $\expansionguarantee = \bigomega{\leftsize}$
  and
  $\expanderdegree = \bigoh{1}$.
  Instead, in  \reflem{lem:expanderexistnew} we  have
  $\expansionguarantee$ growing sublinearly  in~$\leftsize$
  and 
  $\expanderdegree$
  need not be constant
  (although we still need
  $\expanderdegree \lessapprox \log \leftsize / \log \log \leftsize$
  in order to satisfy the conditions of the lemma).
} %

In what follows, unless otherwise stated
$\expandergraph = (\leftvertexset \disjointunion \rightvertexset,E)$ 
will be an 
\boundaryexpnodeg{\expguarantee}{2} for
$\expguarantee =  2\pebblesk$.
We will use such expanders
when we do \XOR substitution in our formulas as described formally in the next
definition.  In words, variables in the \XOR formula are identified
with left vertices~$\leftvertexset$ in~$\expandergraph$, the pool of
new variables is the right vertex set~$\rightvertexset$, and every
variable $u \in \leftvertexset$ in an \XOR clause is replaced by
an exclusive or
$\bigoplus_{v \in N(u)} v$
over its neighbours
$v \in N(u)$.
\ifthenelse{\boolean{conferenceversion}}
{}
{We emphasize that in ``standard'' \xorification as found in the proof
  complexity literature all new substituted variables would be
  distinct, \ie 
  $N(u_1) \cap N(u_2) = \emptyset$
  for $u_1 \neq u_2$. While this often makes formulas harder, it also
  increases the number of variables. Here, 
  we use the approach in~\cite{Razborov16NewKind} to instead recycle
  variables from a much smaller set~$V$ in the substitutions, thus decreasing
  the total number of variables.}

\begin{definition}[\XOR substitution with recycling]
  \label{def:xor-substitution}
  Let $\xformf$ be an \XOR formula with 
  $\variables(\xformf) = \leftvertexset$
  and let
  $\expandergraph = (\leftvertexset \disjointunion \rightvertexset,E)$ 
  be a bipartite graph. 
  For every clause $\constraint = (u_1,\ldots,u_\alternativetoell,\Boolvala)$ in
  $\xformf$ 
  we let $\constraintsubst$ be the clause
  $(v_1^1,\ldots,v_1^{z_1},\ldots,
  v_\alternativetoell^1,
  \ldots,
  v_\alternativetoell^{z_\alternativetoell},\Boolvala)$, 
  where
  $N(u_i)=\{v^1_i,\ldots,v^{z_i}_i\}$ 
  for all $1\leq
  i\leq \alternativetoell$.
  Taking unions,   we let
  $\xformsubst$ be the \XOR formula
  $\xformsubst =
  \setdescr{\constraintsubst}{\constraint\in\xformf}$. 
\end{definition}

When using an
\nmboundaryexp{\leftsize}{{\leftsize}^{3/\expanderdegree}}
{\expdegree}{\expguarantee}{2}
as in \reflem{lem:expanderexistnew} 
for substitution in an $\leftsize$\nobreakdash-variable \XOR 
formula~$\xformorig$ as described in \refdef{def:xor-substitution},   
we obtain a new \XOR formula $\xformsubst$ 
where the number of variables have decreased significantly
to~$\leftsize^{3/\expanderdegree}$. 
The next lemma, which is at the heart of our logic-flavoured version
of hardness condensation, states that a round lower bound for the
\mbox{$\pebblesk$-pebble} game on~$\xformorig$ implies a round lower
bound for the \mbox{$\pebblesk$-pebble} game on~$\xformsubst$.  

\begin{lemma}
  \label{lem:HardnessCondensationGameStyle}
  Let
  $\pebblesk$
  be a positive integer and   let
  $\expandergraph$ 
  be an \nmboundaryexpnodeg{\lsize}{\rsize}{2 \pebblesk}{2}.
  Then if $\xformorig$ is an \XOR formula 
  over $\lsize$~variables such that
  \secondplayer wins the \mbox{$\roundorig$-round}
  \mbox{$\pebblesk$-pebble} game on~$\xformorig$,
  she also wins the \mbox{$\roundorig/(2\pebblesk)$-round}
  \mbox{$\pebblesk$-pebble} game on~$\xformsubst$.
\end{lemma}

\ifthenelse{\boolean{conferenceversion}}{}
{By way of comparison with~\cite{Razborov16NewKind},
  we remark that a straightforward translation of Razborov's technique
  would start with formulas on which \firstplayer can win with few
  pebbles, but needs an almost linear number of rounds to win the
  game, even if he has an infinite amount of 
  pebbles.%
  \footnote{In terms of resolution, this corresponds to formulas that are
    refutable in small width, but where every resolution refutation
    has almost linear depth.}
  Applying this without modification to Immerman's construction, we would
  obtain very weak bounds (and, in particular, nothing interesting for
  constant~$\pebblesk$). Instead, as input to our hardness condensation
  lemma we use a construction that has a round lower bound
  of~$\indexn^{1/\log \pebblesk}$, and show that for hardness
  condensation it is not necessary that the original formula is hard
  over the full range.}

Before embarking on a formal proof of 
\reflem{lem:HardnessCondensationGameStyle},
which is rather technical and will take the rest of this section, let
us discuss the intuition behind it.  The main idea to obtain a good
strategy for \secondplayer on the substituted formula~$\xformsubst$ is
to think of the game as being played on~$\xformorig$ and simulate the
survival strategy there for as long as possible 
(which is where 
\ifthenelse{\boolean{conferenceversion}}
{expansion} 
{boundary expansion} 
comes into play).

Let 
$\expandergraph = (\leftvertexset \disjointunion \rightvertexset,E)$
be an
\boundaryexpnodeg{2\pebblesk}{2} 
as stated in the lemma.
We have
$\variables(\xformorig) = \leftvertexset$ and $\variables(\xformsubst)
= \rightvertexset$. 
Given a strategy for \secondplayer in the $\roundorig$-round
$\pebblesk$-pebble game on $\xformorig$, we want to 
convert this into a winning strategy for \secondplayer for
the $\roundorig/(2\pebblesk)$-round $\pebblesk$-pebble
game on $\xformsubst$. 
A first approach 
(which will not quite work) is the following. 

While playing on the substituted formula $\xformsubst$, \secondplayer
simulates the game on~$\xformorig$. 
For every position $\possubst$ in the game on $\xformsubst$, she
maintains a corresponding position $\posorig$ on $\xformf$, which is
defined on all variables whose entire neighbourhood in the expander is
contained in the domain of  $\possubst$, \ie
$\variables(\posorig)=\Ker(\variables(\possubst))$. 
The assignments of $\posorig$ should be 
defined in such a way that they are
\emph{consistent with} $\possubst$, \ie so that
$\posorig(u)
=  %
\bigoplus_{v\in N(u)}\possubst(v)$. 
It then follows from 
\ifthenelse{\boolean{conferenceversion}}
{the definition of \xorification}
{the description of \xorification in \refdef{def:xor-substitution}}
that $\posorig$ falsifies an \XOR
clause of $\xformorig$ if and only if $\possubst$ falsifies an \XOR
clause of $\xformsubst$.  

Now \secondplayer wants to play
in such a way that if $\possubst$ changes to
$\possubst'$ in one round of the game on~$\xformsubst$, then the
corresponding position $\posorig$ also changes to $\posorig'$ in one
round of the game on~$\xformorig$.
Intuitively, this should 
be done as follows.
Suppose that starting from a position $\possubst$, \firstplayer asks
for a variable $v\in \rightvertexset$.  
If $v$ is not the last 
unassigned
vertex in a neighbourhood of some
$u\in\leftvertexset$, \ie
$\Ker(\variables(\possubst))=\Ker(\variables(\possubst)\cup\{v\})$,
then \secondplayer can make an arbitrary choice as $\posorig =
\posorig'$ is consistent with both choices. 
If $v$ is the last free vertex in the neighbourhood of exactly one
vertex $u$, \ie $\{u\}=\Ker(\variables(\possubst)\cup\{v\})\setminus
\Ker(\variables(\possubst))$, then \secondplayer assumes that she was
asked for $u$ in the simulated game on $\xformorig$.  
If
in her strategy for the $\roundorig$-round $\pebblesk$-pebble game on
$\xformorig$ she 
would answer
with an assignment $a\in\{0,1\}$ 
which would yield the new position 
$\posorig'=\posorig\cup\{u\mapsto a\}$,
then in the game on $\xformsubst$ she now sets $v$ to the right value
$b\in\{0,1\}$ 
so
that the new position
$\possubst'=\possubst\cup\{v\mapsto b\}$ satisfies the consistency
property $\posorig'(u)=\bigoplus_{v\in N(u)}\possubst'(v)$.  
If \secondplayer could follow this strategy, then the number of rounds
she would survive the game on $\xformsubst$ would be  lower-bounded by
the number of rounds she survives in the game on $\xformorig$. 

There is a gap in this intuitive argument, however, namely
how to handle the case when the
queried variable $v$ completes the neighbourhood of two (or more)
vertices
$u_1$, $u_2$ at the same time.
If it holds that
$\set{u_1,u_2} \subseteq \Ker(\variables(\possubst)\cup\{v\})\setminus
\Ker(\variables(\possubst))$,
then we have serious 
\ifthenelse{\boolean{conferenceversion}}
{problems, as $u_1$ and $u_2$ could guide to two different ways of
  assigning~$v$, implying}
{problems. Following the strategy above for $u_1$ and $u_2$ separately can
  yield two different and conflicting ways of assigning~$v$, meaning}
that for the new position~$\possubst'$ there will
be no  consistent assignment $\posorig'$ of $\Ker(\variables(\possubst'))$.  

To circumvent this problem and implement the proof idea above, 
we will use the boundary expansion of~$\expandergraph$
to ensure that this problematic case does not occur. 
For instance, suppose that the graph 
$\restrictedexpander = 
\expandersubgraph{\expandergraph}{\variables(\possubst)}$, 
which is the induced subgraph of~$\expandergraph$ on
$\leftvertexset\setminus\variables(\posorig)$ and 
$\rightvertexset\setminus\variables(\possubst)$,
has boundary expansion at least~$1$. 
Then the bad situation  described above with two variables $u_1,u_2$
having neighbourhood
$N^{\restrictedexpander}(u_1)=N^{\restrictedexpander}(u_2)=\{v\}$ in
$\restrictedexpander$ 
cannot arise, since this would imply
\mbox{$\boundary^{\restrictedexpander}(\{u_1,u_2\}) = \emptyset$}, 
contradicting the  expansion properties of~$\restrictedexpander$.
Unfortunately, we cannot ensure boundary expansion of
$\expandersubgraph{\expandergraph}{\variables(\possubst)}$ for every
position $\possubst$, but we can apply \reflem{lem:ClosedSet} and
extend the current position to a larger one that is defined on
$\closure(\variables(\possubst))$ and has the desired expansion
property.  
Since \reflem{lem:ClosedSet} ensures that the domain
$\Ker(\closure(\variables(\possubst)))$ of our assignment~$\posorig$
under construction is bounded by 
$\setsize{\posorig} \leq \setsize{\possubst} \leq \pebblesk$, 
such an extension will still be good enough. 

We now proceed to present a formal proof.  When doing so, it turns out
to be convenient for us to prove the contrapositive of the statement
discussed above. That is, instead of transforming a strategy for
\secondplayer in the $\roundorig$-round $\pebblesk$-pebble game
on~$\xformorig$ to a strategy for the $\roundorig/(2\pebblesk)$-round
$\pebblesk$-pebble game on~$\xformsubst$
for an \boundaryexpnodeg{2\pebblesk}{2}~$\expandergraph$,
we will show that a winning strategy for \firstplayer in the game
on the substituted formula~$\xformsubst$ can be used to obtain a
winning strategy for \firstplayer in the game on the original
formula~$\xformorig$. 

Suppose that $\possubst$ is any position in the $\pebblesk$-pebble game
on~$\xformsubst$, \ie a partial assignment of variables
in~$\rightvertexset$.
Since 
$\expandergraph$ is a \boundaryexpnodeg{2\pebblesk}{2} and
$\setsize{\possubst} \leq \pebblesk$, 
we can apply
\reflem{lem:ClosedSet} to obtain a superset
$\closure(\variables(\possubst))\supseteq \variables(\possubst)$
having the properties that
$\setsize{\Ker ( \closure ( \variables(\possubst) ) )}
\leq \setsize{\variables(\possubst)}$
and 
the induced subgraph
$\expandersubgraph{\expandergraph}{\closure(\variables(\possubst))}$ is
a \boundaryexpnodeg{\pebblesk}{1}.
For the rest of this section, fix a minimal such set
$\closure(\rightvertexsubset)$ for
every \mbox{$\rightvertexsubset = \variables(\possubst)$} corresponding to a
position~$\possubst$ in the $\pebblesk$-pebble game.  This will allow
us to define formally what we mean by \introduceterm{consistent}
positions in the two games on~$\xformorig$ and~$\xformsubst$ as
described next.

\begin{definition}[Consistent positions]
  \label{def:consistent-positions}
  Let $\posorig$ be a 
  position in the pebble game on~$\xformorig$, \ie a 
  partial assignment of variables in
  $\leftvertexset$, and let $\possubst$ be a partial assignment of
  variables in $\rightvertexset$
  corresponding to a position in the pebble game on~$\xformsubst$.
  We say that $\posorig$ \emph{is consistent with} $\possubst$
  if there exists an extension $\possubstextended\supseteq\possubst$ with
  $\variables(\possubstextended) =
  N \bigl( \variables(\posorig) \bigr) \cup \variables(\possubst)$
  such that
  for all
  $u\in\variables(\posorig)$
  it holds that
  \mbox{$\posorig(u)=\bigoplus_{v\in N(u)}\possubstextended(v)$}. 

  Let $\possubst$ be a position in the $\pebblesk$-pebble game on the
  XOR-substituted formula~$\xformsubst$ and let
  $\closure(\rightvertexsubset)$ be the fixed, minimal closure of
  $\possubst$ chosen above.
  Then we let
  $\posset(\possubst)$ 
  denote   the set of all positions~$\posorig$   with
  $\variables(\posorig)=\Ker(\closure(\variables(\possubst)))$ 
  in the pebble game on~$\xformorig$
  that are
  consistent with~$\possubst$. 
\end{definition}

Observe that for
$\posorig_1 \subseteq \posorig$
and
$\possubst_1 \subseteq \possubst_2$
it holds that if
$\posorig_2$ is consistent with~$\possubst_1$ then so is~$\posorig_1$, 
and if
$\posorig_1$ is consistent with~$\possubst_2$ then
$\posorig_1$ is consistent also with~$\possubst_1$.
Furthermore, by \reflem{lem:ClosedSet} we have 
$\setsize{\posorig} \leq \setsize{\possubst}$
for all
\mbox{$\posorig \in \posset(\possubst)$}.  
The next claim states the core inductive argument.

\begin{claim}
  \label{claim:inductionGameStyle}
  Let $\possubst$ be a position  on $\xformsubst$
  for an \boundaryexpnodeg{2\pebblesk}{2}~$\expandergraph$
  and suppose that 
  \firstplayer wins the $\roundi$-round $\pebblesk$-pebble game
  on~$\xformsubst$ from position~$\possubst$.  
  Then \firstplayer has a strategy to win the $\pebblesk$-pebble
  game on~$\xformorig$ within $2\pebblesk \roundi$ rounds from every
  position $\posorig\in\posset(\possubst)$. 
\end{claim}

\ifthenelse{\boolean{conferenceversion}}
{We note that   this claim 
  is just a stronger (contrapositive) version of 
  \reflem{lem:HardnessCondensationGameStyle}.}
{We note that this claim
  is just a stronger version of (the
  contrapositive of) \reflem{lem:HardnessCondensationGameStyle}.}

\begin{proof}[Proof of \reflem{lem:HardnessCondensationGameStyle}
  assuming Claim~\ref{claim:inductionGameStyle}]
  Note that if $\roundorig/(2\pebblesk)<1$, then the lemma is
  trivially true,
  as \firstplayer always needs at least one round to win the
  pebble game from the empty position. 
  Otherwise, we apply Claim~\ref{claim:inductionGameStyle} 
  with parameters
  $\possubst=\emptyset$ and 
  $\roundi=\roundorig/(2\pebblesk)$.
  Since
  \mbox{$\posset(\emptyset)=\{\emptyset\}$}, 
  we directly get the contrapositive statement of
  \reflem{lem:HardnessCondensationGameStyle}
  that if  \firstplayer{} wins the $\roundorig/(2\pebblesk)$-round
  $\pebblesk$-pebble game on $\xformsubst$, then he wins the
  $\roundorig$-round $\pebblesk$-pebble game on $\xformorig$. 
\end{proof}

All that remains for us to do now is to establish
Claim~\ref{claim:inductionGameStyle}, after which the hardness
condensation lemma will follow easily.

\begin{proof}[Proof of Claim~\ref{claim:inductionGameStyle}]
  The proof is by induction on $\roundi$.
  For the base case $\roundi=0$ we have to show that if $\possubst$
  falsifies an \XOR clause in~$\xformsubst$, then every assignment
  $\posorig\in\posset(\possubst)$ falsifies an \XOR clause
  in~$\xformorig$. 
  But if 
  $\possubst$ falsifies a clause of $\xformsubst$,
  which by construction has the form $\constraintsubst$ for some
  clause $\constraint$ from~$\xformorig$,
  then by
  \reftwodefs{def:xor-substitution}{def:consistent-positions}
  it holds that  every $\posorig\in\posset(\possubst)$ falsifies~$\constraint$. 
  
  For the induction step, suppose that the statement holds for $i-1$
  and assume that \firstplayer{} wins the $\roundi$-round
  $\pebblesk$-pebble game on $\xformsubst$ from position $\possubst$. 
  The $\roundi$th~round
  consists  of two steps:
  \begin{enumerate}
  \item 
    \firstplayer first chooses a subassignment
    $\possubstpart\subseteq\possubst$.
  \item
    He then asks for the value of one variable 
    $v\in \rightvertexset \setminus \Vars{\possubstpart}$,
    to which \secondplayer{} 
    chooses
    an assignment
    $\valueb \in \set{0,1}$ yielding the new position
    $\possubstpart \cup \set{v \mapsto \valueb}$. 
  \end{enumerate}
  As \firstplayer{} has a strategy to win from $\possubst$ within
  $\roundi$ rounds, it follows that he can win from both
  $\possubstpart \cup \set{v\mapsto 0}$ and
  $\possubstpart \cup \set{v\mapsto 1}$ 
  within $\roundi-1$ rounds. 
  By the inductive assumption we then 
  \ifthenelse{\boolean{conferenceversion}}
  {immediately obtain the following
    statement for the set of assignments
    \begin{equation}
      \label{eq:beta-star-v}
      \posset \bigl( \possubstpart\ast v \bigr) \defi 
      \textstyle \bigcup_{a \in \set{0,1}}
      \posset \bigl( \possubstpart \cup \set{v\mapsto a} \bigr) 
    \end{equation}}
  {deduce for the set of assignments
    \begin{equation}
      \label{eq:beta-star-v}
      \posset \bigl( \possubstpart\ast v \bigr) \defi
      \posset \bigl( \possubstpart \cup \set{v\mapsto 0} \bigr) \cup
      \posset \bigl( \possubstpart \cup \set{v\mapsto 1 } \bigr)    
    \end{equation}}
  consistent with either
  $\possubstpart \cup \set{v\mapsto 0}$ 
  \ifthenelse{\boolean{conferenceversion}}
  {or $\possubstpart \cup \set{v\mapsto 1}$.}
  {or $\possubstpart \cup \set{v\mapsto 1}$
    that the following statement holds.}

  \begin{subclaim}
    \label{subclaim:ihyp}
    \firstplayer{} can win the $\pebblesk$-pebble game on $\xformorig$
    within
    $2\pebblesk(\roundi-1)$ rounds
    from all positions in~%
    $\posset \bigl( \possubstpart\ast v \bigr)
    $.
  \end{subclaim}

  Note that a position is in 
  $\posset \bigl( \possubstpart\ast v \bigr)$
  if it is consistent with either 
  $\possubstpart\cup\{v\mapsto 0\}$ 
  or
  $\possubstpart\cup\{v\mapsto 1\}$.  
  Therefore,  
  $\posset \bigl( \possubstpart\ast v \bigr)$ 
  is the set of all positions  over 
  $\Ker \bigl( \closure \bigl( \possubstpart \bigr) \cup \set{v} \bigr)$ 
  that are consistent with~$\possubstpart$. 
  What remains to show is that from  every position
  $\posorig\in\posset(\possubst)$ \firstplayer{} can reach some
  position in
  $\posset \bigl( \possubstpart\ast v \bigr)$ within $2\pebblesk$~rounds.  
  We split the proof into two steps, corresponding to the two steps in
  the move of \firstplayer from position~$\possubst$.
  \begin{subclaim}
    \label{subclaim:statement-a}
    From every position
    $\posorig \in \posset(\possubst)$ \firstplayer{} can reach some
    position in $\posset \bigl( \possubstpart \bigr)$ 
    for 
    $\possubstpart \subseteq \possubst$
    within $\pebblesk$~rounds. 
  \end{subclaim}

  \begin{subclaim}
    \label{subclaim:statement-b}
    From every position
    $\posorig\in\posset \bigl( \possubstpart \bigr)$
    \firstplayer{} can reach some
    position in $\posset \bigl( \possubstpart\ast v \bigr)$
    within $\pebblesk$~rounds. 
  \end{subclaim}
  
  \ifthenelse{\boolean{conferenceversion}}
  {%
    We now establish Subclaim~\ref{subclaim:statement-b}. 
    The proof of Subclaim~\ref{subclaim:statement-a} is similar and
    deferred to the full-length version of the paper.}
  {%
    Let us establish Subclaims~\ref{subclaim:statement-a}
    and~\ref{subclaim:statement-b} in reverse order.}

  \begin{subproof}%
    [Proof of Subclaim~\ref{subclaim:statement-b}]
    \firstplayer starts with 
    \ifthenelse{\boolean{conferenceversion}}
    {some assignment
      $\posorigstart\in\posset\bigl( \possubstpart \bigr)$ defined over}
    {an assignment
      $\posorigstart\in\posset\bigl( \possubstpart \bigr)$, 
      which is defined over the variables}
    $\kstart = 
    \Ker\bigl( 
    \closure\bigl( \variables \bigl( \possubstpart \bigr) \bigr) 
    \bigr)$,
    and wants to reach some assignment 
    \mbox{$\posorigend \in \posset\bigl( \possubstpart\ast v \bigr)$}
    defined over the variables
    $\kend = 
    \Ker\bigl( 
    \closure \bigl( \variables \bigl( \possubstpart \bigr) 
    \cup \set{v} \bigr) 
    \bigr)$. 
    
    If 
    $
    \Ker\bigl( 
    \closure\bigl( \variables \bigl( \possubstpart \bigr) \bigr) 
    \bigr)
    =
    \Ker\bigl( 
    \closure \bigl( \variables \bigl( \possubstpart \bigr) 
    \cup \set{v} \bigr) 
    \bigr)$, 
    then \firstplayer can choose 
    $\posorigend = \posorigstart$.
    To see this, note that if $\posorigstart$ assigns a value to some
    $u \in N(v)$, then since
    $\posorigstart\in\posset\bigl( \possubstpart \bigr)$
    it holds by
    \refdef{def:consistent-positions}
    that
    $
    N(u) \subseteq
    \closure\bigl( \variables \bigl( \possubstpart \bigr) \bigr) 
    $,
    and thus
    $\posorigstart$
    is already consistent with
    $\possubstpart\cup\set{v\mapsto \valueb}$ 
    for some
    $\valueb \in \set{0,1}$.
    Hence, \firstplayer need not ask any question in this case, but
    the induction hypothesis immediately yields the desired
    conclusion.

    The more interesting case is when
    $
    \Ker\bigl( 
    \closure\bigl( \variables \bigl( \possubstpart \bigr) \bigr) 
    \bigr)
    \neq
    \Ker\bigl( 
    \closure \bigl( \variables \bigl( \possubstpart \bigr) 
    \cup \set{v} \bigr) 
    \bigr)$.
    Now \firstplayer first deletes all assignments of variables in
    $\kstart \setminus \kend$ from~$\posorigstart$
    to get~$\posorig_0$.
    Since
    $\posorig_0 \subseteq \posorigstart$
    and $\posorigstart$ is consistent with~$\possubstpart$ by assumption, 
    $\posorig_0$ is also consistent with~$\possubstpart$.
    Afterwards, he asks for all variables in 
    $\leftvertexsubset = \kend \setminus \kstart$.  
    We need to argue that regardless of how \secondplayer answers, it
    holds that \firstplayer reaches a position that is consistent
    with~$\possubstpart$ .This is where the peeling argument in
    \reflem{lem:peeling_lemma}
    is needed.

    As discussed above, by our choice of the closure
    $\closure\bigl( \variables\bigl( \possubstpart \bigr) \bigr)$
    (obtained using \reflem{lem:ClosedSet})
    we know that the bipartite graph
    $\restrictedexpander
    =
    \expandersubgraph{\expandergraph}
    {\closure\bigl( \variables\bigl( \possubstpart \bigr) \bigr)}$
    is 
    a \boundaryexpnodeg{\pebblesk}{1} 
    and furthermore that for
    $\leftvertexsubset = \kend \setminus \kstart$
    it holds that
    $\setsize{\leftvertexsubset}
    \leq 
    \setsize{\kend} 
    \leq
    \Setsize{ \variables\bigl( \possubstpart \bigr)\cup \set{v}} 
    \leq 
    \pebblesk
    $,
    as observed right after
    \refdef{def:consistent-positions}.
    Hence, we can apply 
    \reflem{lem:peeling_lemma} to~$\restrictedexpander$ 
    and~$\leftvertexsubset$ to get an ordered sequence
    $u_1,\ldots,u_\alternativetoell$ satisfying
    \mbox{$N^{\restrictedexpander}(u_i)\setminus
      N^{\restrictedexpander}(\{u_1,\ldots,u_{i-1}\})\neq \emptyset$}.  
    We will think of \firstplayer as querying the (at
    most~$\pebblesk$) vertices in~$\leftvertexsubset$ in this order, 
    after which he ends up with a position~$\posorigend$ defined on
    the variables~$\kend$. 

    To argue that the position~$\posorigend$ obtained in this way is
    consistent  with~$\possubstpart$ independently of how
    \secondplayer answers,  and is hence contained in $\posset\bigl(
    \possubstpart\ast v \bigr)$, we show inductively that all positions
    encountered during the     transition from $\posorigstart$
    to~$\posorigend$ are  consistent with $\possubstpart$.  
    As already noted, this holds for the position~$\posorig_0$
    obtained from~$\posorigstart$ by deleting all assignments  of
    variables in $\kstart \setminus \kend$. 
    For the induction step,
    let $i\geq 0$ and
    assume inductively
    that the current    position~$\posorig_i$ over 
    \begin{equation}
      \label{eq:transitional-position-domain}
      \leftvertexset_i
      \defi
      (\kstart \cap \kend) \cup
      \setdescr{u_j}{1\leq j \leq i}
    \end{equation}
    is consistent with $\possubstpart$.
    Now \firstplayer asks about the variable~$u_{i+1}$ and \secondplayer
    answers with a value~$\valuea_{i+1}$.  
    Since $\posorig_i$ is consistent with $\possubstpart$, there is an
    assignment $\possubstextended\supseteq \possubstpart$ that sets the
    variables $v\in N(\variables(\posorig_i))$ to the right values such
    that
    $\posorig_i(u) = \bigoplus_{v\in N(u)}\possubstextended(v)$ 
    for all
    $u \in \variables(\posorig_i)$.  
    By our ordering of
    \ifthenelse{\boolean{conferenceversion}}
    {$\leftvertexsubset = \set{u_1,\ldots,u_\alternativetoell}$}
    {$\leftvertexsubset = \Set{u_1,\ldots,u_\alternativetoell}$}
    chosen above
    we know that $u_{i+1}$ has at least one
    neighbour 
    on the right-hand side~$\rightvertexset$
    that is neither contained in
    \mbox{$N^{\expandergraph}(\leftvertexset_i)
      =
      N^{\expandergraph}(\variables(\posorig_i))$} 
    nor in the domain of~$\possubstpart$.  
    Hence, regardless of which value~$a_{i+1}$ \secondplayer chooses
    for her answer
    we can extend the assignment~$\possubstextended$ to the
    variables 
    $N^{\expandergraph}(u_{i+1})\setminus
    \bigl( N^{\expandergraph}\bigl( \variables(\posorig_i) \bigr)\cup
    \variables\bigl( \possubstpart \bigr) \bigr)$ 
    in such a way that
    $\bigoplus_{v\in
      N(u_{i+1})}\possubstextended(v)=\valuea_{i+1}$. 
    This shows that $\posorig_{i+1}$ defined over 
    $\leftvertexset_{i+1}
    =
    (\kstart \cap \kend) \cup
    \setdescr{u_j}{1\leq j \leq i+1}
    $
    is consistent with~$\possubstpart$.
    Subclaim~\ref{subclaim:statement-b} now follows by the induction
    principle. 
  \end{subproof}

  \ifthenelse{\boolean{conferenceversion}}{}
  {Before proving  Subclaim~\ref{subclaim:statement-a}, 
    we should perhaps point out why this claim is not vacuous. 
    Recalling the discussion just below
    \reflem{lem:ClosedSet}, this is because the condition
    $\rightvertexset_1 \subseteq \rightvertexset_2$
    does not allow us to conclude that
    $\closure(\rightvertexset_1) \subseteq \closure(\rightvertexset_2)$.

    \begin{subproof}%
    [Proof of Subclaim~\ref{subclaim:statement-a}]
    The proof is similar to that of
    Subclaim~\ref{subclaim:statement-b} above. 
    \firstplayer{} starts with an
    assignment~$\posorigstart\in\posset(\possubst)$ and wants to reach
    some 
    assignment in~$\posset(\possubstpart)$
    for $\possubstpart \subseteq \possubst$
    within $\pebblesk$~rounds.  
    By assumption,
    $\posorigstart$ is consistent with~$\possubst$ and
    therefore (since $\possubstpart\subseteq \possubst$) is also
    consistent with~$\possubstpart$.  
    \firstplayer deletes all assignments from the domain
    $\kstart=\Ker(\closure(\variables(\possubst)))$ of $\posorigstart$
    that do not occur in the domain
    $\kend=\Ker(\closure(\variables(\possubstpart)))$ of positions in
    $\posset(\possubstpart)$,
    resulting in the position
    $\posorig_0 \subseteq \posorigstart$ that is consistent
    with~$\possubstpart$.   
    Next, he applies 
    \reflem{lem:peeling_lemma} to
    $\restrictedexpander
    =
    \expandersubgraph{\expandergraph}{\closure(\variables(\possubst))}$
    to obtain an ordering of the remaining variables
    $\kend\setminus\kstart$. 
    In the same way as above he can query the variables in this order
    while maintaining the invariant that the current position is
    consistent with~$\possubstpart$. 
  \end{subproof}
  }%

  Combining
  Subclaims~\ref{subclaim:ihyp}, \ref{subclaim:statement-a}
  and~\ref{subclaim:statement-b},
  we conclude that
  \firstplayer wins from every position
  \mbox{$\posorig\in\posset(\possubst)$} 
  within $2\pebblesk\roundi$~rounds.  
  This concludes the proof of Claim~\ref{claim:inductionGameStyle}. 
\end{proof}

We are finally in a position to give a formal proof of
\reflem{lem:hardnessCondensationXOR}.

\begin{proof}[Proof of Lemma~\ref{lem:hardnessCondensationXOR}]
  Let $\mindegree \in \Nplus$ be the constant in
  \reflem{lem:expanderexistnew}.
  Suppose we are given an $\indexlhs$-variable \mbox{$\initialxorwidth$-XOR}
  formula~$\xformf$ and parameters $\pebbleslowerlhs$,
  $\pebblesupperlhs$, $\roundlowerbound$, $\expanderdegree$ 
  satisfying the conditions  in
  \reflem{lem:hardnessCondensationXOR}
  that
  $
  \pebblesupperlhs/\pebbleslowerlhs
  \geq
  \expanderdegree 
  \geq
  \mindegree
  $ 
  and  
  $(2\pebblesupperlhs
  \expanderdegree)^{2\expanderdegree} \leq \indexlhs$.

  Fix
  $\pebblesk \defi \pebblesupperlhs$
  and
  $\expansionguarantee \defi 2\pebblesupperlhs$.
  Since
  $(\expansionguarantee\expanderdegree)^{2\expanderdegree} \leq
  \indexlhs$ and $\expanderdegree\geq\mindegree$,
  we appeal to \reflem{lem:expanderexistnew} to obtain an
  \nmboundaryexp{\leftsize}{\ceiling{{\leftsize}^{3/\expanderdegree}}}
  {\expdegree}{\expguarantee}{2}{}
  $\expandergraph =  (\leftvertexset \disjointunion
  \rightvertexset,E)$,
  and applying \xorification \wrt~$\expandergraph$
  we construct the formula
  $\xformg \defi \xformf\substituted$. 
  Clearly, 
  $\xformg$
  is an $(\expanderdegree\initialxorwidth)$\nobreakdash-\XOR{} formula
  with
  $\Ceiling{\indexlhs^{3/\expanderdegree}}$ variables.
  We want to prove that \firstplayer has a winning strategy for the
  $(\expanderdegree\pebbleslowerlhs)$-pebble game
  on~$\xformg$
  as guaranteed by
  \reflem{lem:hardnessCondensationXOR}\ref{item:condensationrhs-a},
  but that he  does not win the 
  \mbox{$\pebblesupperlhs$-pebble} game on~$\xformg$ 
  within 
  ${\roundlowerbound}/{(2\pebblesupperlhs)}$~rounds
  as stated in
  \reflem{lem:hardnessCondensationXOR}\ref{item:condensationrhs-b}.

  For the upper bound 
  in   \reflem{lem:hardnessCondensationXOR}\ref{item:condensationrhs-a},
  we recall that
  \firstplayer{} has a winning strategy in the
  \mbox{$\pebbleslowerlhs$-pebble} game on~$\xformf$
  by assumption~\ref{item:condensationlhs-a} in the lemma.
  He can use this strategy to win the
  $(\expanderdegree\pebbleslowerlhs)$-pebble game on~$\xformg$ as
  follows. 
  Whenever his strategy tells him to ask for a variable $u\in
  \leftvertexset = \variables(\xformf)$, he instead asks for the at
  most $\expanderdegree$ variables in
  $N(u)\subseteq \rightvertexset = \variables(\xformg)$ 
  and assigns to $u$ the value
  that corresponds to the parity of the answers \secondplayer{}
  gives for~%
  $N(u)$. 
  In this way, he can simulate his strategy on $\xformf$ 
  until he reaches
  an assignment that contradicts an \XOR{} clause
  $\constraint$ from $\xformf$.  
  As the corresponding assignment of the variables $\setdescr{v}{v\in
    N(u),u\in \variables(\constraint)}$ falsifies the constraint
  $\subst{\constraint}{\expandergraph}\in \xformg$, 
  at this point \firstplayer wins the
  $(\expanderdegree\pebbleslowerlhs)$-pebble game on~$\xformg$. 
  
  The lower bound 
  in \reflem{lem:hardnessCondensationXOR}\ref{item:condensationrhs-b}
  follows immediately
  from \reflem{lem:HardnessCondensationGameStyle}.
  By assumption~\ref{item:condensationlhs-b} 
  in \reflem{lem:hardnessCondensationXOR},
  \firstplayer does not win the 
  \mbox{$\pebblesupperlhs$-pebble} game on~$\xformf$
  within $\roundlowerbound$~rounds.
  Since 
  $\expandergraph$ 
  is an \nmboundaryexpnodeg{\lsize}{\rsize}{2 \pebblesk}{2},
  \reflem{lem:HardnessCondensationGameStyle}
  says that that  he does not win the
  $\pebblesupperlhs$-pebble game on 
  $\xformg = \xformf\substituted$ 
  within
  ${\roundlowerbound}/{(2\pebblesupperlhs)}$~rounds either. 
  This concludes the proof of
  \reflem{lem:hardnessCondensationXOR}
\end{proof}

\makeatletter{}%

\section{Concluding Remarks}
\label{sec:conclusion}

In this paper we prove an
$n^{\bigomega{\pebblestd/\log \pebblestd}}$ lower bound
on the minimal quantifier depth of $\FOk$ and $\FOcntk$ sentences that
distinguish two finite $n$-element relational structures, nearly
matching the trivial $n^{\pebblestd-1}$ upper bound.  
By the known connection to the $\wldim$-dimensional Weisfeiler--Leman
algorithm, 
this result implies
near-optimal 
$n^{\bigomega{\wldim/\log \wldim}}$ lower bounds
also on the number of refinement steps of this algorithm. 
The key technical ingredient in our proof is the hardness condensation
technique recently introduced by Razborov~\cite{Razborov16NewKind} in
the context of proof complexity, which we translate into the language
of finite variable logics and use to reduce the domain size of
relational structures while maintaining the minimal quantifier depth
required to distinguish them.

An obvious open problem is to improve 
our
lower bound.  One way to
achieve this would be to strengthen the lower bound on the number of
rounds in the $\pebblestd$-pebble game on \mbox{$3$-\XOR{}} formulas
in \reflem{lem:pyramids} from 
$\indexn^{ 1 / \log \pebblestd}$ 
to $\indexn^\delta$ for some 
$\delta \gg 1 / \log \pebblestd$.  
By the hardness condensation lemma this would directly improve our
lower bound from $\indexn^{\bigomega{\wldim/\log\wldim}}$ to
$\indexn^{\bigomega{\delta\wldim}}$. 

The structures on which our lower bounds hold are $n$-element
relational structures of arity~$\bigtheta{\wldim}$ and
size~$\indexn^{\bigtheta{\wldim}}$.  We would have liked to have this
results also for structures of bounded arity, such as graphs.
However, the increase of the arity is inherent in the method of
amplifying hardness by making \XOR substitutions.
An optimal lower bound of $\indexn^{\bigomega{\pebblestd}}$ on the
quantifier depth required to distinguish two \mbox{$n$-vertex} graphs has
been obtained by \theauthorCB in an earlier work~\cite{Berkholz.2014}
for the \emph{existential-positive fragment} of~$\FOk$.  Determining
the quantifier rank of full $\FOk$ and $\FOcntk$ on \mbox{$n$-vertex}
graphs remains an open problem.

\newcommand{\inputSizeNotLittleN}{(\|\strucA\|+\|\strucB\|)}

\ifthenelse{\boolean{conferenceversion}}
{Another open question} 
{Another open question related to our results} 
concerns the complexity of finite variable equivalence for
non-constant $\pebblestd$.  What is the complexity of deciding, given
two structures and a parameter~$\pebblestd$, whether
the structures
are equivalent in $\FOk$ or~$\FOcntk$?  As this problem can be solved
in time $\inputSizeNotLittleN^{O(\pebblestd)}$, it is in \EXPTIME
if $\pebblestd$ is part of the input.  
It has been conjectured 
that this problem is
\mbox{\EXPTIME}-complete~\cite{GraedelKolaitisLibkinMarxSpencerVardiWeinstein.2007},
but it is not even known whether it is \mbox{\NP-hard.} 
Note that the quantifier depth is connected to the computational  
complexity of the equivalence problem by the fact that an upper bound
of the form $\indexn^{\bigoh{1}}$ on $\indexn$-element structures  
would have implied that testing equivalence is in \PSPACE{}. 
Hence, our lower bounds on the quantifier depth
can be seen as
a necessary requirement for establishing \EXPTIME-hardness of the
equivalence problem.

\makeatletter{}%

\ifthenelse{\boolean{conferenceversion}}
{\acks}
{\section*{Acknowledgements}}

We are very grateful to Alexander~Razborov for patiently explaining
\ifthenelse{\boolean{conferenceversion}}
{his  hardness condensation technique}
{the  hardness condensation technique in~\cite{Razborov16NewKind}}
during numerous and detailed discussions.
We also want to thank Neil~Immerman for
helping us find relevant references to previous works and explaining
some of the finer details in~\cite{Immerman.1981}.
\TheauthorCB wishes to acknowledge useful feedback from the
participants of  Dagstuhl Seminar~15511
\emph{The~Graph Isomorphism Problem}.
Finally, we are most indebted to the anonymous \emph{LICS} reviewers
for very detailed comments that helped improve the manuscript
considerably. 

Part of the work of \theauthorCB was performed while at KTH Royal
Institute of Technology supported by a fellowship within the
Postdoc-Programme of the German Academic Exchange Service (DAAD).
The research of \theauthorJN was supported by the
European Research Council under the European Union's Seventh Framework
Programme \mbox{(FP7/2007--2013) /} ERC grant agreement no.~279611
and by
Swedish Research Council grants 
\mbox{621-2010-4797}
and
\mbox{621-2012-5645}.

\bibliography{refArticles,refBooks,refOther,refLocal}%

\bibliographystyle{alpha}

\appendix

\makeatletter{}%
\section{Existence and Properties of Expander Graphs}
\label{app:existence-expander}

In this appendix we present proofs of
\reftwolems{lem:ClosedSet}{lem:expanderexistnew}, 
starting with the latter lemma.
We again remark that most of this material can already be
found in a similar form in~\cite{Razborov16NewKind}, although the
exact
parameters are somewhat different. 
It also seems appropriate to point out 
that there is a significant overlap with essentially
identical technical lemmas in~\cite{BN16Supercritical}.  

Just to avoid ambiguity, let us state explicitly that even though we
have the Euler number~$\eulernumbervariable$ appearing below, 
we still think of all logarithms as being taken to base~$2$ (though
this should not really matter).

\begin{lem:expanderexists}
  There is 
  an absolute constant $\mindegree \in \Nplus$ 
  such that for all 
  $
  \expanderdegree
  $, 
  $\expansionguarantee$,  
  $\leftsize$
  satisfying
  $
  \expanderdegree \geq \mindegree
  $
  and 
  \mbox{$(\expansionguarantee\expanderdegree)^{2\expanderdegree} \leq \leftsize$}
  there exist
  \nmboundaryexp{\leftsize}{\Ceiling{{\leftsize}^{3/\expanderdegree}}}
  {\expdegree}{\expguarantee}{2}{}s.
\end{lem:expanderexists}

\begin{proof}
  Let $\leftvertexset$ and $\rightvertexset$ be two disjoint sets of
  vertices of size $\setsize{\leftvertexset}=\leftsize$ and
  $\setsize{\rightvertexset}=\rightsize =
  \Ceiling{\leftsize^{3/\expanderdegree}}$.  
  For every $u\in\leftvertexset$ we choose $\expanderdegree$ times a
  neighbour $v\in\rightvertexset$ uniformly at random with
  repetitions.  
  This 
  yields
  a bipartite graph $\expandergraph = (\leftvertexset
  \disjointunion  \rightvertexset,E)$  of left-degree at most
  $\expanderdegree$.  
  In the sequel we show that 
  $\expandergraph$ is likely to be an
  \boundaryexp{\expdegree}{\expguarantee}{2}. 

  First note that for every set
  $\leftvertexsubset\subseteq\leftvertexset$ all neighbours $v\in
  \nbhd(\leftvertexsubset)\setminus\boundary(\leftvertexsubset)$ that
  are not in the boundary of~$\leftvertexsubset$ have at least two
  neighbours in~$\leftvertexsubset$.  
  Since there are at most
  $\expanderdegree\setsize\leftvertexsubset-\setsize{\boundary(\leftvertexsubset)}$
  edges between $\leftvertexsubset$ and
  $\nbhd(\leftvertexsubset)\setminus\boundary(\leftvertexsubset)$, it
  follows that
  $\setsize{\nbhd(\leftvertexsubset)\setminus\boundary(\leftvertexsubset)}\leq
  (\expanderdegree\setsize\leftvertexsubset-\setsize{\boundary(\leftvertexsubset)})/2$
  and hence  
  \begin{equation}
    \setsize{\nbhd(\leftvertexsubset)} 
    \leq
    \frac{\setsize{\boundary(\leftvertexsubset)} +
      \expanderdegree\setsize\leftvertexsubset}
    {2}
    \eqperiod
    \label{eq:neighbourhoodbound}
  \end{equation}
  If $\expandergraph$ is not an
  \boundaryexp{\expdegree}{\expguarantee}{2}, then there is a set
  $\leftvertexsubset$ of size
  $\Setsize{\leftvertexsubset} = 
  \leftvertexsubsetsize \leq
  \expansionguarantee$ 
  that has a boundary~$\boundary(\leftvertexsubset)$ 
  of size 
  $\Setsize{\boundary(\leftvertexsubset)} < 
  2 \leftvertexsubsetsize$ 
  and from~\eqref{eq:neighbourhoodbound} it then
  follows that
  $\setsize{\nbhd(\leftvertexsubset)}
  <
  (1+\expanderdegree/2)\leftvertexsubsetsize$.  
  By a union bound argument (and relaxing to non-strict inequalities)
  we obtain 
  \begin{subequations}
    \begin{align}
      &\Pr[\expandergraph \text{ is not an
        \boundaryexp{\expdegree}{\expguarantee}{2}}] 
      \\ 
       \leq \ & %
        \sum^{\expansionguarantee}_{\leftvertexsubsetsize=1}        
        \ \sum_{\leftvertexsubset \subseteq \leftvertexset ;\,
          \setsize{\leftvertexsubset} = \leftvertexsubsetsize}
      \Pr\big[
      \setsize{\boundary(\leftvertexsubset)} \leq
      2 \leftvertexsubsetsize 
      \big] 
      \\ 
      \leq \ & %
        \sum^{\expansionguarantee}_{\leftvertexsubsetsize=1}
        \ \sum_{\leftvertexsubset \subseteq \leftvertexset ;\,
          \setsize{\leftvertexsubset} = \leftvertexsubsetsize}
        \Pr\big[
        \setsize{\nbhd(\leftvertexsubset)}
        \leq
        (1+\expanderdegree/2)\leftvertexsubsetsize
        \big] 
      \\ 
      \leq \ & %
        \sum^{\expansionguarantee}_{\leftvertexsubsetsize=1}
        \binom{\leftsize}{\leftvertexsubsetsize}
        \binom{\rightsize}{(1+\expanderdegree/2)\leftvertexsubsetsize}
        \left(
        \frac{(1+\expanderdegree/2)\leftvertexsubsetsize}{\rightsize}
        \right)^{\expanderdegree\leftvertexsubsetsize} 
        \label{eq:binom}
      \\ 
      \leq \ & %
        \sum^{\expansionguarantee}_{\leftvertexsubsetsize=1}
        \leftsize^\leftvertexsubsetsize
        \left(
        \frac{\eulernumber\rightsize}
        {(1+\expanderdegree/2)\leftvertexsubsetsize}
        \right)^{(1+\expanderdegree/2)\leftvertexsubsetsize}
        \left(
        \frac{(1+\expanderdegree/2)\leftvertexsubsetsize}{\rightsize}
        \right)^{\expanderdegree\leftvertexsubsetsize} 
        \label{eq:binombound}
      \\ 
      = \ & %
        \sum^{\expansionguarantee}_{\leftvertexsubsetsize=1}
        \leftsize^\leftvertexsubsetsize
        (\eulernumber\rightsize)^{(1+\expanderdegree/2)\leftvertexsubsetsize}
        \left(
        (1+\expanderdegree/2)\leftvertexsubsetsize
        \right)^{(\expanderdegree/2-1)\leftvertexsubsetsize}
        \rightsize^{-\expanderdegree\leftvertexsubsetsize} 
      \\ 
      \label{eq:line-minus-two}
      \leq \ & %
        \sum^{\expansionguarantee}_{\leftvertexsubsetsize=1} 
        \rightsize^{(\expanderdegree/3)\ell} 
        (\eulernumber\rightsize)^{(1+\expanderdegree/2)\leftvertexsubsetsize}
        \left(
        (1+\expanderdegree/2)\leftvertexsubsetsize
        \right)^{(\expanderdegree/2-1)\leftvertexsubsetsize} 
        \rightsize^{-\expanderdegree\leftvertexsubsetsize} 
      \\
      \label{eq:line-minus-one}
      = \ & %
        \sum_{\ell=1}^\expansionguarantee
        \rightsize^{(\expanderdegree/3)\ell}
        \rightsize^{\frac{\log \eulernumber}{\log \rightsize} 
        (1+\expanderdegree/2)\leftvertexsubsetsize}
        \rightsize^{\frac{1}{\log \rightsize}
        \log\bigl((\expanderdegree/2+1)\ell \bigr) 
        (\expanderdegree/2-1)\leftvertexsubsetsize}
        \rightsize^{(-\expanderdegree/2+1)\ell}
      \\
      \label{eq:expander_proof_bigsum}
    \leq \ & %
      \sum_{\ell=1}^\expansionguarantee
      \rightsize^{\bigl(
      \frac{\log \eulernumber}
      {\log \rightsize}
      \expanderdegree+\frac{1}{\log \rightsize}
      \log ( \expanderdegree\expansionguarantee ) 
      (\expanderdegree/2 - 1) - \expanderdegree/6 + 1
      \bigr)
      \ell} 
      \eqcomma
    \end{align}
  \end{subequations}
where in going from~\eqref{eq:binom} to~\eqref{eq:binombound} we use
  the inequality
  $
  \binom{n}{k}
  \leq
  \bigl(
  \frac{\eulernumbervariable{}n}{k}
  \bigr)^k
  $
  for
  $\eulernumbervariable \approx \eulernumberfourdigits$
  denoting the Euler number,
  and in going from~\refeq{eq:line-minus-one}
  to~\refeq{eq:expander_proof_bigsum}
  we assume that
  $\expanderdegree \geq 2$ 
  and also use that
  $\ell\leq\expansionguarantee$.   

  In order to show that the
  expression~\refeq{eq:expander_proof_bigsum} is bounded away
  from~$1$---which implies that 
  $\expandergraph$ 
  is an 
  \boundaryexp{\expdegree}{\expguarantee}{2}
  with constant probability---%
  it suffices to study the exponent and prove that
  there is a constant $\smallepsilon > 0$ such that
    \begin{equation}
    \label{eq:exponent-lt-zero}
    \frac{\log \eulernumber}{\log \rightsize}
    \expanderdegree + 
    \frac{1}{\log \rightsize}
    \log ( \expanderdegree\expansionguarantee ) 
    \left( \frac{\expanderdegree}{2} - 1 \right) 
    - \frac{\expanderdegree}{6} + 1 
    \leq -\smallepsilon
    < 0
    \eqcomma
  \end{equation}
  which holds if there is a constant 
  $\smallepsilonalt = \smallepsilon / \expanderdegree$ 
  such that
  \begin{equation}
    \label{eq:exponent-div-Delta-lt-zero}
    \frac{\log \eulernumber}{\log \rightsize} +
    \frac{1}{\log \rightsize} 
    \log( \expanderdegree\expansionguarantee ) 
    \left(
      \frac{1}{2} - \frac{1}{\expanderdegree} 
    \right) 
    - \frac{1}{6} + \frac{1}{\expanderdegree}  
    \leq - \smallepsilonalt
    < 0
    \eqperiod
  \end{equation}
  Since
  $
  (\expansionguarantee\expanderdegree)^{2\expanderdegree} 
  \leq 
  \leftsize 
  \leq \rightsize^{\expanderdegree/3}$ 
  we have 
  $\expansionguarantee\expanderdegree \leq \rightsize^{1/6}$,
  and it follows that
  \begin{subequations}
    \begin{align}
      &\frac{\log \eulernumber}{\log \rightsize} +
        \frac{\log ( \expanderdegree\expansionguarantee ) 
        (1 / 2 - 1 / \expanderdegree)}
        {\log \rightsize}
        - \frac{1}{6} + \frac{1}{\expanderdegree} 
      \\
      \leq \ 
      &\frac{\log \eulernumber}{\log \rightsize} +
        \frac{\log \bigl( \rightsize^{1/6} \bigr)}{2 \log \rightsize}
        - \frac{1}{6} + \frac{1}{\expanderdegree} 
      \\
      \label{eq:before-exponent-punchline}
      = \ 
      &\frac{\log \eulernumber}{\log \rightsize} 
        - \frac{1}{12} + \frac{1}{\expanderdegree} 
      \\
      \label{eq:exponent-punchline}
      \leq \ & {-\smallepsilonalt} < 0
    \end{align}
  \end{subequations}
  where
  we can make the last inequality hold for
  $\smallepsilonalt$ small enough and
  $\rightsize$ and $\expanderdegree$ large enough.
  Fix $\smallepsilonalt$ satisfying $0 < \smallepsilonalt < 1/12$
  and choose~$\rightsize_0$ so that
  the inequality between
  \refeq{eq:before-exponent-punchline}
  and~\refeq{eq:exponent-punchline} holds for any
  $\rightsize \geq \rightsize_0$
  and
  $\expanderdegree\geq 13$.
  Then we obtain that
  \eqref{eq:expander_proof_bigsum}
  is bounded by 
  $\sum_{\ell=1}^\expansionguarantee \rightsize^{-\smallepsilonalt \ell}$. 
  Insisting in addition \mbox{that
  $\rightsize \geq  3^{1/\smallepsilonalt}$},
  we can upper-bound   \eqref{eq:expander_proof_bigsum} by 
  \begin{align}
    \sum_{\ell=1}^\expansionguarantee \rightsize^{-\smallepsilonalt
    \ell} 
    \leq 
    \sum_{\ell=1}^\infty \bigl(\tfrac13\bigr)^\ell 
    \leq 
    \tfrac12 
    \eqperiod
  \end{align}
  It remains to calculate how to set $\mindegree$
  to make sure that all of these conditions hold.
  Note that by assumption we have
  \mbox{$(\expansionguarantee\expanderdegree)^{2\expanderdegree}
    \leq \leftsize$},
  which implies that
  $\expanderdegree^{\expanderdegree} \leq \leftsize$.  
  It follows that 
  we will always have
  $
  \rightsize 
  =
  \Ceiling{\leftsize^{3/\expanderdegree}}
  \geq
  (\expanderdegree^\expanderdegree)^{3/\expanderdegree} 
  = 
  \expanderdegree^3 
  \geq 
  (\mindegree)^3$
  and hence it is sufficient to choose
  $
  \mindegree
  \geq
  \max \bigl( 
  {\rightsize_0}^{1/3} ,
  3^{1/3\smallepsilonalt} , 
  13
  \bigr)$.  
  This concludes the proof of the lemma.
\end{proof}

We next prove that in a good enough boundary expander it holds that
for any small enough right vertex set~$\rightvertexsubset$ there is a superset
$\closure\bigl(\rightvertexsubset\bigr) \supseteq \rightvertexsubset$
with a small kernel such that the induced subgraph
$\expandersubgraph{\expandergraph}{\closure(\rightvertexsubset)}$  
(obtained from~$\expandergraph$ by deleting $\rightvertexsubset$ and
then all isolated vertices from $\leftvertexset$) 
is also a good boundary expander.
Recall that we refer to this set $\closure\bigl(\rightvertexsubset\bigr)$
as the \introduceterm{closure} of~$\rightvertexsubset$.

\begin{lem:closedset}
  Let $\expandergraph$ be an
  \boundaryexpnodeg{\expguarantee}{2}.
  Then for every $\rightvertexsubset \subseteq \rightvertexset$ with
  $\setsize{\rightvertexsubset}\leq \expansionguarantee/2$ 
  there exists a subset
  $\closure\bigl(\rightvertexsubset\bigr) \subseteq \rightvertexset$ 
  with
  $\closure(\rightvertexsubset) \supseteq \rightvertexsubset$
  such that 
  $\Setsize{\Ker \bigl( \closure \bigl( \rightvertexsubset \bigr)\bigr)}
  \leq \Setsize{\rightvertexsubset}$
  and the induced subgraph
  $\expandersubgraph{\expandergraph}{\closure(\rightvertexsubset)}$ is
  an
  \boundaryexpnodeg{\expguarantee/2}{1}.
\end{lem:closedset}

\newcommand{\vset}{V}
\newcommand{\uset}{U}
\newcommand{\utilde}{\leftvertexsubset}
\newcommand{\ubar}{\overline{U}}
\renewcommand{\ubar}{U^*}
\newcommand{\finindex}{\tau}

\begin{proof}
  Let 
  $\expandergraph=(\leftvertexset \disjointunion \rightvertexset,E)$ 
  be an 
  \boundaryexpnodeg{\expguarantee}{2} and let
  $\rightvertexsubset\subseteq\rightvertexset$
  have size
  \mbox{$\setsize{\rightvertexsubset}\leq \expansionguarantee/2$}. 
  We construct an increasing sequence
  $\rightvertexsubset=\vset_0\subset \vset_1 \subset \cdots \subset
  \vset_\finindex=\closure(\rightvertexsubset)$ such that
  $\expandersubgraph{\expandergraph}{\vset_\finindex}$ is an
  \boundaryexpnodeg{\expguarantee/2}{1}
  as follows.

  If
  $\expandersubgraph{\expandergraph}{\vset_0}$ 
  is not an
  \boundaryexpnodeg{\expguarantee/2}{1}, then there exists a set
  $\uset_1$ of size 
  $\setsize{\uset_1} \leq \expguarantee/2$ 
  such that
  $\Setsize{\boundary^{\expandersubgraph{\expandergraph}{\vset_0}}(\uset_1)}
  \leq \setsize{\uset_1}$.  
  Delete
  $\uset_1$ and all its neighbours
  from~$\expandersubgraph{\expandergraph}{\vset_0}$.  
  If the resulting graph is not an
  \boundaryexpnodeg{\expguarantee/2}{1}, we repeat this process and
  iteratively delete vertex sets that do not satisfy the expansion
  condition. 
  Formally, for $i\geq 1$ fix $\uset_i$ to be any set of size 
  $\setsize{\uset_i} \leq \expguarantee/2$ 
  such that  
  \begin{equation}
    \label{eq:U-i-set-condition}
    \Setsize{\boundary^{\expandersubgraph{\expandergraph}{\vset_{i-1}}}(\uset_i)}
    \leq
    \setsize{\uset_i}
    \eqcomma    
  \end{equation}
  where we set 
  \begin{equation}
    \label{eq:V-i-set-definition}
    \vset_{i} \defi \vset_0
    \cup
    \bigcup_{j=1}^{i} \nbhd^{\expandergraph}(\uset_j)
  \end{equation}
  (and where we note that, formally speaking, what is deleted at the
  \mbox{$i$th step} is $\nbhd^{\expandergraph}(\uset_i)$ together with the
  kernel~$\ker(\nbhd^{\expandergraph}(\uset_i))$ of this right vertex
  set).  
  Since all sets~$\uset_i$ constructed above are non-empty,
  this process must terminate  for some
  $i=\finindex$ and the resulting graph
  $\expandersubgraph{\expandergraph}{\vset_\finindex}$ is then an
  \boundaryexpnodeg{\expguarantee/2}{1}
  (note that an empty graph without vertices vacuously satisfies the
  expansion  condition).     It remains to verify that the size condition
  $\setsize{\Ker(\vset_\finindex)}\leq \setsize{\vset_0}$
  for the kernel of the closure of~$\rightvertexsubset$ holds.
  This is immediately implied by the following inductive claim. 
  
  \begin{claim}
    \label{claim:closure}
    Let 
    $\vset_{-1} = \uset_0 = \emptyset$ and suppose that $i\geq0$. 
    Then for $\uset_i$ satisfying~\refeq{eq:U-i-set-condition}
    and $\vset_i$ defined by~\refeq{eq:V-i-set-definition}
    the following properties hold:
    \begin{enumerate}
    \item
      \label{item:proofclosure1} 
      For all $\leftvertexsubset$ such that $\Ker(\vset_{i-1})\cup
      \uset_{i}\subseteq \leftvertexsubset \subseteq \Ker(\vset_i)$ 
      we have
      $\Setsize{\boundary^{\expandergraph}(\leftvertexsubset)\setminus\vset_0}
      \leq
      \setsize{\Ker(\vset_i)}$. 
    \item
      \label{item:proofclosure2} 
      The kernel of~$\vset_i$ has size
      $\setsize{\Ker(\vset_i)}\leq \setsize{\vset_0}$.
    \end{enumerate} 
  \end{claim}

\newcommand{\propclosureone}{Property~\ref{item:proofclosure1}\xspace} 
\newcommand{\propclosuredtwo}{Property~\ref{item:proofclosure2}\xspace} 
  \newcommand{\claimclosure}{Claim~\ref{claim:closure}\xspace}

  For $i=0$, 
  \propclosureone  in  \claimclosure
  follows because 
  $\leftvertexsubset
  \subseteq \Ker(\vset_0)$ 
  implies
  that
  $\boundary^{\expandergraph}(\leftvertexsubset)\subseteq \vset_0$.  
  For 
  \propclosuredtwo,
  suppose that $\setsize{\Ker(\vset_0)}\leq
  \expansionguarantee$.  
  Then expansion implies 
  $
  2\setsize{\Ker(\vset_0)} \leq
  \setsize{\boundary^{\expandergraph}(\Ker(\vset_0))}$,
  and combining this with
  $\boundary^{\expandergraph}(\ker(\vset_0)) \subseteq \vset_0$
  we obtain
  $\setsize{\Ker(\vset_0)} \leq \frac{1}{2}\setsize{\vset_0}$.
  If instead
  $\setsize{\Ker(\vset_0)} > \expansionguarantee$,
  then we can find a subset
  $\leftvertexsubset\subseteq \Ker(\vset_0)$ of size
  $\setsize{\leftvertexsubset} = \expansionguarantee$.  
  By expansion we have
  $\setsize{\boundary^{\expandergraph}(\leftvertexsubset)}
  \geq
  2\expansionguarantee$,
  which is a contradiction because
  as argued above we should have
  $\Setsize{\boundary^{\expandergraph}(\leftvertexsubset)}
  \leq 
  \setsize{\vset_0}
  \leq 
  \expansionguarantee/2$.

  For the induction step, suppose that both properties hold
  \mbox{for $i-1$}. 
  Let $\ubar = \Ker(\vset_{i-1})\cup \uset_{i}$
  and consider any 
  $\leftvertexsubset$
  satisfying
  $
  \ubar
  \subseteq \leftvertexsubset \subseteq \Ker(\vset_i)$. 
  We claim that every boundary element in
  $\boundary^{\expandergraph}(\utilde)$ is
  either a boundary element from 
  $\boundary^{\expandergraph}(\ubar)$ or
  is
  contained in
  $\vset_0$.  
  To see this, note that since
  $\utilde\subseteq \Ker(\vset_i)$ we have
  $\boundary^{\expandergraph}(\utilde) 
  \subseteq \vset_i 
  = \vset_0 \cup
  \bigcup_{j=1}^{i} \nbhd^{\expandergraph}(\uset_j)$.  
  Furthermore, it can be observed that
  $\bigcup_{j=1}^{i} \uset_j \subseteq
  \ubar\subseteq \utilde$
  (this is basically due to the fact that
  $\nbhd(\Ker(\rightvertexsubset)) \subseteq \rightvertexsubset$
  for any~$\rightvertexsubset$). 
  Hence,
  if
  $v \in \boundary^{\expandergraph}(\utilde) \setminus \vset_0$,
  then it must hold that
  $v \in \bigcup_{j=1}^{i} \nbhd^{\expandergraph}(\uset_j)$,
  and so the unique neighbour of~$v$ on the left is
  contained in
  $\bigcup_{j=1}^{i} \uset_j$ 
  and 
  therefore
  also in $\ubar$, implying that
  $v\in \boundary(\ubar)$. 
  This yields that
  \begin{equation}
    \label{eq:boundary-U-prime}
    \boundary^{\expandergraph}(\utilde)\setminus \vset_0 \subseteq
    \boundary^{\expandergraph}(\ubar)\setminus \vset_0 
  \end{equation}
  as claimed, and in what follows we will show
  \begin{equation}
    \label{eq:boundary-U-star}
    \Setsize{\boundary^{\expandergraph}(\ubar) \setminus \vset_0}
    =
    \Setsize{\boundary^{\expandergraph}(\Ker(\vset_{i-1}) \cup
      \uset_{i}) \setminus \vset_0}
    \leq 
    \setsize{\Ker(\vset_i)}
  \end{equation}
  in order
  to prove \propclosureone.

  Note that by
  construction
  every vertex in $\vset_{i-1}\setminus\vset_0$
  has at least one neighbour in $\Ker(\vset_{i-1})$. 
  It follows that all new boundary vertices 
  in
  $
  \boundary^{\expandergraph}(\Ker(\vset_{i-1})\cup\uset_{i})
  \setminus
  \boundary^{\expandergraph}(\Ker(\vset_{i-1}))
  $ 
  are either from~$\vset_0$  or from the boundary
  $\boundary^{\expandersubgraph{\expandergraph}{\vset_{i-1}}}(\uset_i)$
  of $\uset_i$ outside of $\vset_{i-1}$.   
  Therefore we have
  \begin{equation}
    \label{eq:boundary-ker-V-iminus1-U-i}
    \boundary^{\expandergraph}(\ubar) \setminus \vset_0
    =
    \boundary^{\expandergraph} \bigl( \Ker(\vset_{i-1}) \cup
    \uset_{i} \bigr)
    \setminus \vset_0 
    \subseteq
    \bigl( \boundary^{\expandergraph}(\Ker(\vset_{i-1})) \setminus \vset_0\bigr)
    \disjointunion
    \boundary^{\expandersubgraph{\expandergraph}{\vset_{i-1}}}(\uset_i)    
    \eqperiod
  \end{equation}
  Since by assumption $\uset_i$  does not satisfy the expansion
  condition we know that 
  \begin{equation}
    \label{eq:boundary-G-induced-U-i}
    \Setsize{\boundary^{\expandersubgraph{\expandergraph}{\vset_{i-1}}}(\uset_i)}
    \leq \setsize{\uset_i}    
  \end{equation}
  and by the inductive hypothesis concerning \propclosureone we have
  \begin{equation}
    \label{eq:using-IH-prop-1}
    \Setsize{\boundary^{\expandergraph}(\Ker(\vset_{i-1})) \setminus \vset_0}
    \leq
    \setsize{\Ker(\vset_{i-1})} \eqperiod     
  \end{equation}
  Combining
  \refeq{eq:boundary-U-prime}
  with
  \mbox{\refeq{eq:boundary-ker-V-iminus1-U-i}--\refeq{eq:using-IH-prop-1}}
  we deduce that
  \begin{equation}
    \label{eq:induction-step-punchline-prop1}
    \Setsize{\boundary^{\expandergraph}(\leftvertexsubset)
      \setminus \vset_0}
    \leq
    \Setsize{\boundary^{\expandergraph}(\Ker(\vset_{i-1})
      \cup \uset_{i})
      \setminus \vset_0} 
    \leq 
    \setsize{\Ker(\vset_{i-1})} + \setsize{\uset_i}
    \leq
    \setsize{\Ker(\vset_i)}
    \eqcomma
  \end{equation}
  where the final inequality holds since
  $\Ker(\vset_{i-1})$ and $\uset_i$ are disjoint subsets of
  $\Ker(\vset_i)$.   
  This concludes the inductive step for \propclosureone.

  To establish \propclosuredtwo, assume first that
  $\setsize{\Ker(\vset_i)} \leq   \expansionguarantee$.  
  Then by expansion and 
  \propclosureone applied to 
  $\utilde = \Ker(\vset_i)$ 
  we have 
  \begin{equation}
    \label{eq:prop2-1}
    2\setsize{\Ker(\vset_i)}
    \leq
    \Setsize{\boundary^{\expandergraph}(\Ker(\vset_i))} 
    \leq 
    \setsize{\vset_0} + \setsize{\Ker(\vset_i)}    
  \end{equation}
  and hence 
  \begin{equation}
    \label{eq:prop2-2}
    \setsize{\Ker(\vset_i)}
    \leq
    \setsize{\vset_0}
  \end{equation} 
  as desired.
  If instead
  $\setsize{\Ker(\vset_i)} >
  \expansionguarantee$,  
  then by the inductive hypothesis we know that
  $\setsize{\Ker(\vset_{i-1})}\leq \expansionguarantee/2$ 
  and by
  construction we have
  $\setsize{\uset_i}\leq\expansionguarantee/2$. 
  Therefore 
  we can find
  a set 
  $\leftvertexsubset$ 
  of size
  $\setsize{\leftvertexsubset} = \expansionguarantee$
  satisfying the condition
  $\Ker(\vset_{i-1})\cup\uset_i\subseteq\leftvertexsubset\subseteq
  \Ker(\vset_i)$
  in \propclosureone.
  From the expansion properties of~$\expandergraph$ we conclude that
  $\setsize{\boundary(\leftvertexsubset)}\geq2\expansionguarantee$,
  which is a contradiction because 
  for sets~$\leftvertexsubset$ 
  satisfying the conditions  in \propclosureone we
  derived~\refeq{eq:induction-step-punchline-prop1}, 
  which implies that
  $\setsize{\boundary(\leftvertexsubset)}\leq \setsize{\vset_0}
  +\setsize{\Ker(\vset_{i-1})} +\setsize{\uset_i} \leq
  3\expansionguarantee/2$. 
\end{proof}

\end{document}

